\documentclass[prx,twocolumn,showpacs,amsmath,amssymb,superscriptaddress,floatfix,reprint, nofootinbib]{revtex4-2}
\usepackage{macros}
\makeatletter
\def\l@subsubsection#1#2{}
\makeatother

\makeatletter 
    
\renewcommand\onecolumngrid{
\do@columngrid{one}{\@ne}%
\def\set@footnotewidth{\onecolumngrid}
\def\footnoterule{\kern-6pt\hrule width 1.5in\kern6pt}%
}

\renewcommand\twocolumngrid{
        \def\footnoterule{
        \dimen@\skip\footins\divide\dimen@\thr@@
        \kern-\dimen@\hrule width.5in\kern\dimen@}
        \do@columngrid{mlt}{\tw@}
}%

\makeatother

\begin{document}

\title{Non-Universality from Conserved Superoperators in Unitary Circuits}
\date{\today}
\author{Marco Lastres}
\affiliation{Technical University of Munich, TUM School of Natural Sciences, Physics Department, 85748 Garching, Germany}
\affiliation{Munich Center for Quantum Science and Technology (MCQST), Schellingstr. 4, 80799 M\"unchen, Germany}
\author{Frank Pollmann}
\affiliation{Technical University of Munich, TUM School of Natural Sciences, Physics Department, 85748 Garching, Germany}
\affiliation{Munich Center for Quantum Science and Technology (MCQST), Schellingstr. 4, 80799 M\"unchen, Germany}
\author{Sanjay Moudgalya}
\affiliation{Technical University of Munich, TUM School of Natural Sciences, Physics Department, 85748 Garching, Germany}
\affiliation{Munich Center for Quantum Science and Technology (MCQST), Schellingstr. 4, 80799 M\"unchen, Germany}
\email{sanjay.moudgalya@gmail.com}
\begin{abstract}
An important result in the theory of quantum control is the ``universality'' of $2$-local unitary gates, i.e. the fact that any global unitary evolution of a system of $L$ qudits can be implemented by composition of $2$-local unitary gates.
Surprisingly, recent results have shown that universality can break down in the presence of symmetries: in general, not all globally symmetric unitaries can be constructed using $k$-local symmetric unitary gates.
This also restricts the dynamics that can be implemented by symmetric local Hamiltonians.
In this paper, we show that obstructions to universality in such settings can in general be understood in terms of superoperator symmetries associated with unitary evolution by restricted sets of gates.
These superoperator symmetries lead to block decompositions of the operator Hilbert space, which dictate the connectivity of operator space, and hence the structure of the dynamical Lie algebra.
We demonstrate this explicitly in several examples by systematically deriving the superoperator symmetries from the gate structure using the framework of commutant algebras, which has been used to systematically derive symmetries in other quantum many-body systems. 
We clearly delineate two different types of non-universality, which stem from different structures of the superoperator symmetries, and discuss its signatures in physical observables. 
In all, our work establishes a comprehensive framework to explore the universality of unitary circuits and derive physical consequences of its absence.
\end{abstract}
\maketitle
\tableofcontents

\section{Introduction} 
Understanding the landscape of unitary operations that can be implemented with a set of ``elementary" gates has been one of the primary goals of many areas of quantum information theory, and has led to the entire field of quantum control.
This question is also closely related to the notion of \textit{complexity} of a given unitary operation, which has been under intense study with applications to many disparate areas of physics, from quantum computation to high energy physics~\cite{nielsen2005geometric, brown2018second, haferkamp2022linear,  piroli2022random, bulchandani2021smooth}.
In the context of quantum many-body physics, determining the minimum depth of circuits needed to perform certain operations lies at the heart of classifying various phases of matter, such as Symmetry Protected Topological (SPT) phases~\cite{chen2010local, huang2015quantum}, and also leads to interesting notions of Quantum Cellular Automata (QCAs), which have several applications of their own~\cite{lent1993quantum, arrighi2019overview, farrelly2020review}.
The dynamics of unitary evolution by a set of elementary gates is also a question of great interest, e.g., the ensemble of such evolutions is said to form a \textit{$k$-design} at time $t_k$ if it reproduces the $k$-th moments of the ensemble of Haar random matrices~\cite{harrow2009random, brandao2016local, hunterjones2019unitary}. 
These concepts are now also being extended to cases with symmetries~\cite{hearth2023unitary, li2024designslocalrandomquantum}, and various physical consequences of symmetries for unitary dynamics are being explored~\cite{rakovszky2018diffusive, khemani2018operator, rakovszky2019subballistic, friedman2019spectral, huang2020dynamics, zhou2020diffusive, moudgalya2021spectral,  ogunnaike2023unifying, moudgalya2023symmetries}, leading to a rich landscape of possibilities.
Understanding the full class of possible operations given a set of unitary gates in general is a notoriously hard problem.
However, dramatic simplifications occur by imposing two conditions on this problem. 
First, we can understand the set of unitaries generated from a set of $k$-local unitary gates, i.e., those with support on exactly $k$ consecutive sites, without imposing any conditions on the depth of the circuit.
Second, we only specify broad restrictions on the set of elementary gates, e.g., we assume continuous control so that the allowed elementary unitaries are one-parameter families of the form $e^{i \theta h_\alpha}$ for any $\theta$.
Under these conditions, if no constraints are imposed on $\{h_\alpha\}$, it is clear that any unitary operation can always be generated, and this is a standard result in the theory of quantum control~\cite{d2007introduction}.
Restricting instead the set $\{h_\alpha\}$ to only contain symmetric operators (under the action of some symmetry group $G$) one might again expect that \textit{all} symmetric unitaries can be generated starting from a set of strictly local symmetric gates, assuming continuous control and no depth restrictions on the unitary circuit.
Rather surprisingly,  Marvian and collaborators~\cite{marvian2020locality, hulse2021qudit, marvian2022rotationally, marvian2024abelian} have recently found the converse even for simple symmetries such as $U(1)$ and $SU(2)$.
This means that in general there are obstructions to constructing globally symmetric unitaries from strictly $k$-local symmetric sets of gates.
However, these interesting results rely heavily on the group theory of the particular on-site symmetry groups involved, and the general conditions for the appearance of such non-universality given a set of gates is still lacking. 
For example, particular sets of gates can have symmetries that do not have simple on-site group structures, and it is not clear how to describe universality in those cases.
Moreover, questions on the universality under a \textit{subset} of symmetric gates, in particular or the \textit{amount} of non-universality under such conditions remain unanswered.
For example, symmetric gates that have a Gaussian or matchgate structure result in a well-known type of non-universality, since the product of two Gaussian unitary gates is also a Gaussian unitary gate, hence precluding the generation of general symmetric unitary gates. 
In this work, we unify all these kinds of non-universality and present a systematic framework to understand this problem, and illustrate the precise algebraic conditions that guarantee non-universality. 
The core idea we use is that of \textit{commutant algebras}, which has been applied in a variety of settings to understand block decompositions of the Hilbert space under a given set of operations, e.g., in the context of decoherence-free subspaces~\cite{lidar1998decoherence, lidar2003decoherencereview}, virtual subsystems and reference frames~\cite{zanardi2001virtual, bartlett2007reference}, and quantum error correction~\cite{poulin2005stabilizer}.
More recently, this framework has been used to understand various kinds of symmetries in quantum many-body systems~\cite{moudgalya2021hilbert, moudgalya2022from, moudgalya2022exhaustive, moudgalya2023numerical}, and the corresponding block-diagonalization of symmetric operators into quantum number sectors. 
Due to the block diagonalization, the existence of a symmetry can be equivalently interpreted as an obstruction in connectivity of states in the Hilbert space under time evolution by symmetric operators, i.e., in more standard language, states within different quantum number sectors cannot be connected to each other under symmetric evolution.
The novelty in the commutant framework is that the symmetries that lead to block-diagonalization can either be of many different types beyond the conventional on-site symmetry groups usually considered in the literature, e.g., they could be generalized symmetries that have a categorical structure~\cite{lootens2021MPO}, or even more unconventional symmetries generated by non-local operators with no obvious simple underlying structure~\cite{rakovszky2020statistical, moudgalya2021hilbert, moudgalya2022exhaustive}.
These kinds of unconventional symmetries lead to the better understanding of phenomenon of weak ergodicity breaking~\cite{serbyn2020review, moudgalya2021review,  papic2021review, chandran2022review}, where the apparent block-diagonalization of the time-evolution operators cannot be explained by more conventional symmetries~\cite{moudgalya2021hilbert, moudgalya2022from}. 
In this work, we study this problem of non-universality of a set of gates with continuous control in terms of the so-called \textit{Dynamical Lie Algebra} (DLA) of the generators of the unitaries. 
Given a set of unitary gates of the form $\{e^{i \theta h_\alpha}\}$, the DLA is the Lie algebra generated by $\{h_\alpha\}$, obtained by taking nested commutators of the these operators and their linear combinations.
While the study of the DLA is a standard tool in the literature on quantum control~\cite{d2007introduction}, here we view it from the point of view of superoperators that act on the space of operators.
In particular, the DLA is obtained by repeated adjoint actions of the commutators of $\{h_\alpha\}$, which are superoperators, on the set of the generators $\{h_\alpha\}$.
The structure of the DLA can then be completely understood using the connectivity of operator Hilbert space under these adjoint actions, which in turn can then be obtained by studying the superoperator symmetries of the adjoint action.
This is analogous to the fact that the connectivity of the physical Hilbert space under a set of operators can be obtained by studying the symmetries of the operators. 
In this work we show that the origin of non-universality in symmetric systems, demonstrated in the previous works~\cite{marvian2020locality, hulse2021qudit, marvian2022rotationally,  marvian2024abelian} can easily be traced to \textit{superoperator symmetries} of superoperators responsible for time-evolution of operators. 
This allows us to provide a clear criterion for non-universality for \textit{arbitrary} gate sets with continuous control, which also includes cases where the representation theory for the symmetries is not yet well-developed, e.g., for the unconventional symmetries in the context of Hilbert space fragmentation~\cite{moudgalya2021hilbert}.
This also allows us to use previously developed methods~\cite{moudgalya2023numerical} to compute superoperator symmetries, and hence test non-universality, numerically on finite-size systems.
This also leads to a clear two-fold classification for the non-universality of a set of gates, as opposed to a four-fold classification proposed in earlier literature \cite{marvian2024abelian}.
First, the gates can exhibit {weak non-universality}, where the non-universality is mild, and all the superoperator symmetries are derived from the physical symmetries.
This appears to be the generic case, {and it implies \textit{semi-universality}, where the non-universality is only due to the fact that relative phases between symmetry sectors cannot be controlled,} and has been the focus of earlier studies of non-universality~\cite{marvian2020locality, marvian2024abelian}.
Second, the gates can exhibit \textit{strong non-universality}, where superoperator symmetries not derived from the physical symmetries can exist.  
In such cases, non-universality can be observed in observables such as the entanglement entropies and Out-of-Time-Ordered Correlation (OTOC) or higher point correlations functions, and we demonstrate examples of this. %
These two conditions also connect to earlier results obtained in the literature on the quantum simulation of a particular set of gates from another set of gates~\cite{zimboras2015symmetry}, where the problem of dynamical Lie algebras was approached from a different perspective, and we elucidate the precise connections.
This paper is organized as follows.
In Sec.~\ref{sec:bondcommutant}, we review the general framework of commutant algebras and discuss its connections to Hilbert space decomposition. 
Then in Sec.~\ref{sec:superoperatoralgebra}, we discuss the core results of this work, which is the application of the commutant framework to study the connectivity of operators in Hilbert space under the action of certain superoperators, and we illustrate the connection to dynamical Lie algebras and the generation of unitaries.
We also provide numerical methods to study this problem. 
In Sec.~\ref{sec:nonuniversality}, we discuss the implications of these results on the connectivity of operators space to the non-universality of unitary circuits.
There we discuss the two classes of non-universality that have different origins.
Finally, in Sec.~\ref{sec:physicalimplications} we discuss the connection between non-universality and other physical phenomena, with particular reference to Out-of-Time-Ordered-Correlators (OTOCs), Rényi entropies, and the appearence of $k$-designs.
We close in Sec.~\ref{sec:conclusions} with a summary of open questions. 
The appendices provide technical details on various parts of the main text.

\begin{figure}[t]
\includegraphics[width=0.91\columnwidth]{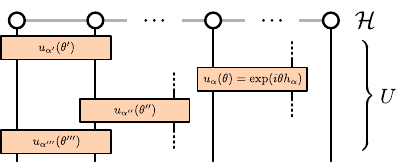}
\caption{Schematic representation of the setup considered for studying non-universality in Sec~\ref{sec:superoperatoralgebra}. Given a set of generators $\{h_\alpha\}$ we study the set of unitaries that can be obtained as arbitrary products of $u_\alpha(\theta)=\exp(i\theta h_\alpha)$. The hermitian generators $\{h_\alpha\}$ need not be local as in the picture presented here.
The main question we study in this work is if this setup can generate \textit{all} global unitaries with the same symmetries as $\{u_\alpha(\theta)\}$.
\label{fig:circuit}}
\end{figure}

\section{Bond and Commutant Algebras}\label{sec:bondcommutant}
We now review the concepts of bond algebras and their associated commutant algebras.
These kinds of objects have been studied in various parts of the quantum information literature~\cite{lidar1998decoherence, zanardi2001virtual, poulin2005stabilizer, bartlett2007reference}, for example, in decoherence-free subspaces~\cite{lidar2003decoherencereview, lidar2014dfs}.
More recently, they were shown to naturally arise in quantum many-body physics when analyzing the symmetries of \textit{families} of Hamiltonians or unitaries that are built starting from a given set of local interactions \cite{moudgalya2021hilbert, moudgalya2022from}.
We will take this quantum many-body physics point of view when introducing the main concepts related to this framework; in the next section we will adapt the language to the question of universality.
\subsection{Definitions}
For concreteness let us consider a finite-dimensional many-body Hilbert space $\mc H=\mc H_\mrm{loc}^{\ot L}$ and Hamiltonians $H(\{J_\alpha\})$ or unitaries $U(\{J_\alpha\})$ of the form
\begin{equation}\label{eq:bondhamiltonian}
H(\{J_\alpha\})=\sum_\alpha J_\alpha H_\alpha,\quad
U(\{J_\alpha\}) = \prod_\alpha e^{i J_\alpha H_\alpha}
\end{equation}
where $\{H_\alpha\}$ is a set of hermitian interaction terms, which we usually take to be strictly local on a lattice, and $\{J_\alpha\}$ is a set of arbitrary real coefficients.
We define the ``commutant algebra'' $\mc C$ associated to the operators defined in \eqref{eq:bondhamiltonian}, as the set of \textit{all} operators on the Hilbert space $\mc H$ that commute with each interaction term
\begin{equation}
    Q\in\mc C \iff [Q,H_\alpha]=0\quad\forall\alpha.
\end{equation}
For any finite system size, the set $\mc C$ is a finite-dimensional complex vector space that contains the identity $\1$, and it is also closed under matrix multiplication and hermitian adjoint: therefore $\mc C$ is a {finite-dimensional} \textit{von Neumann algebra}~\cite{landsman1998lecture, harlow2017}.
This definition provides a generalized notion of symmetry: it is not restricted to the usual on-site symmetry groups generated by local charges $Q=\sum_{j=1}^L Q_j$, but instead includes any symmetry operator compatible with the local structure of the Hamiltonian or the unitary gates.
If we wish to emphasize the initial set of gates, in the following we will use one of the following alternative notation:
\begin{equation}
    \mc C\defeq\mrm{comm}(\{H_\alpha\}).
\end{equation}
Together with the commutant algebra we can define the ``bond algebra'' $\mc A$ to be the algebra generated by the interaction terms themselves:
\begin{equation}
    \mc A\defeq\llangle\{H_\alpha\}\rrangle
\end{equation}
where the notation $\llangle\cdot\rrangle$ indicates the von Neumann algebra generated by the inner set (by including the identity $\1$ and performing linear combinations, multiplications, and hermitian adjoints).
It can be easily seen that all operators in $\mc A$ still commute with all operators in $\mc C$, but a stronger statement is actually true: $\mc A$ is the set of \textit{all} operators that commute with $\mc C$.
This is due to the Double Commutant Theorem for von Neumann algebras \cite{landsman1998lecture, harlow2017, moudgalya2022from}, which states that for hermitian $\{H_\alpha\}$
\begin{equation}\label{eq:doublecomm}
    \mrm{comm}(\mrm{comm}(\{H_\alpha\}))=\llangle\{H_\alpha\}\rrangle.
\end{equation}
This crucial fact illustrates the symmetry between $\mc A$ and $\mc C$ in the statements below.
For example, notice that the algebra $\mc Z\defeq\mc A\cap\mc C$ is the center\footnote{The center of an algebra is the set of all operators in the algebra that commute with all other operators in the algebra.} of both $\mc A$ and $\mc C$ (we will therefore refer to it as ``the center'').

We can briefly illustrate this framework by showing how it represents conventional symmetries \cite{moudgalya2022from}. 
Focusing on a Hilbert space of spin-1/2 d.o.f.'s with $\mH_{\loc} = \mathbb{C}^2$, the symmetry algebra generated by a single $U(1)$ global charge of the form $N_\mrm{tot}=\frac{1}{2}\sum_j (\1+Z_j)$ can be understood as the commutant of the algebra $\bondx{U(1)}$ of operators generated by a set of 2-local interaction terms:
\begin{equation}
    \bondx{U(1)}=\llangle\{X_jX_{j+1}+Y_jY_{j+1},\, Z_jZ_{j+1},\, Z_j\}\rrangle.
\end{equation}
The commutant algebra has the form 
\begin{equation}
    \commx{U(1)} = \llangle\{N_\mrm{tot}\}\rrangle=\mrm{span}(\{\1,N_\mrm{tot},N_\mrm{tot}^2,...\}).
\label{eq:CU1}
\end{equation}
An orthogonal basis of $\commx{U(1)}$ is given by the projectors onto the charge sectors $N_\mrm{tot}=n$
\begin{equation}
    \commx{U(1)}=\mrm{span}(\{\Pi_n\}_{n=0}^{L}),\label{eq:U1basis}
\end{equation}
and therefore $\mrm{dim}(\commx{U(1)})=L+1$.

Since the $U(1)$ symmetry is Abelian, the commutant algebra $\mC_{U(1)}$ coincides with its center.
This is no longer the case for non-Abelian groups such as $SU(2)$; here the commutant is the algebra generated by the symmetry operators $\commx{SU(2)} = \lgen \{S^\alpha_{\tot} \defn \sum_j{S^\alpha_j}\}\rgen$, which is the universal enveloping algebra of the $\mf s\mf u(2)$ Lie algebra.
This algebra can be shown to be the commutant of the bond algebra $\bondx{SU(2)}=\llangle\{\vec S_j\cdot\vec S_{j+1}\}\rrangle$ \cite{moudgalya2022from},  which is related to the group algebra of the permutation group $S_L$; this is related to the Schur-Weyl duality~\cite{fulton2013representation}.
The center of the two algebras is generated by the Casimir element $\centx{SU(2)}=\lgen\{ S_\mrm{tot}^2 \}\rgen$.

We refer readers to Ref.~\cite{moudgalya2022from} for several additional examples of conventional symmetries studied in the language of commutant algebras.

\subsection{Hilbert Space Decomposition}
A fundamental property of finite-dimensional von Neumann algebras is that, given a pair of algebras $(\mc A,\mc C)$ that are each other's commutant, then the Hilbert space on which they act can be decomposed as~\cite{harlow2017, moudgalya2021hilbert}\footnote{This is also a consequence of the Wedderburn–Artin theorem~\cite{basicalgebra}.}
\begin{equation}
\mc H=\bigoplus_{\lambda}\left(\mc H_\lambda^{\mc A}\ot\mc H_\lambda^{\mc C}\right)\label{eq:fund-th}
\end{equation}
where the abstract spaces $\mc H_\lambda^{\mc A}$ (resp. $\mc H_\lambda^{\mc C}$) correspond to inequivalent irreducible representations of $\mc A$ (resp. $\mc C$). We define $\mc H_\lambda\defeq\mc H_\lambda^{\mc A}\ot\mc H_\lambda^{\mc C}$.
This decomposition means that for each $\lambda$ in the direct sum there is a tensored basis:
\begin{equation}
    \{\ket{\alpha}_\lambda\ot\ket{\gamma}_\lambda\}_{\substack{\alpha=1,...,D_\lambda\\\gamma=1,...,d_\lambda\,}}\label{eq:matbasis}
\end{equation}
such that operators in $\mc A$ (resp. $\mc C$) only act on the first (resp. second) factor in the product; in other words, elements $ K\in \mc A$ and $ Q\in \mc C$ of the algebras have the following matrix form:
\begin{equation}\label{eq:matrixrep}
   \begin{split}
   K=\bigoplus_\lambda(M_{\lambda}(K)\ot \1_{d_\lambda})\\ Q=\bigoplus_\lambda(\1_{D_\lambda}\ot N_{\lambda}(Q))
   \end{split}
\end{equation}
where $M_\lambda(\cdot)$ and $N_\lambda(\cdot)$ are $D_\lambda$- and $d_\lambda$-dimensional irreducible representations of $\mc A$ and $\mc C$ respectively. As a consequence elements $Z\in\mc Z$ of the center are
\begin{equation}\label{eq:matrixrepcenter}
    Z=\bigoplus_{\lambda}c_\lambda(Z)(\1_{D_\lambda}\ot \1_{d_\lambda}),\quad c_\lambda(Z)\in\mb C
\end{equation}
and in particular a linear basis for the center is given by the projectors $\{\Pi_\lambda\}$ onto the $\mc H_\lambda^{\mc A}\ot\mc H^{\mc C}_\lambda$ subspaces in Eq.~\eqref{eq:fund-th}.
The fact that an $r$-dimensional representation $R(\cdot)$ of such an algebra is irreducible means that any complex $r\times r$ matrix can be represented as $R(O)$ for a given element $O$ of the algebra.
By exploiting Eq. \eqref{eq:matrixrep}, it is then evident that the dimensions of the algebras, i.e., the number of linearly independent elements, are given by~\cite{moudgalya2021hilbert}
\begin{equation}
    \mrm{dim}(\mc A)=\sum_\lambda D_\lambda^2,
    \qquad
    \mrm{dim}(\mc C)=\sum_\lambda d_\lambda^2.
\end{equation}
Since the Hamiltonians $H(\{J_\alpha\})$ and the unitaries $U(\{J_\alpha\})$ of Eq.~(\ref{eq:bondhamiltonian}) belong to the bond algebra $\mc A$, this theorem tells us that all the Hamiltonians and unitaries in the family can simultaneously be block-diagonalized according to Eq.~\eqref{eq:matrixrep}.
Hence the dynamics described by these Hamiltonians or unitaries preserve a shared set of invariant subspaces (also sometimes referred to as \textit{Krylov subspaces}) of the form $\mc H_\lambda^{\mc A}\ot\mrm{span}\{\ket\gamma_\lambda\}$ for any $\ket\gamma_\lambda\in\mc H_\lambda^{\mc C}$.
If the commutant is generated by a conventional symmetry \textit{group}, the Krylov space decomposition corresponds to the decomposition into irreps of the group. 
For the $U(1)$ case of Eq.~\eqref{eq:CU1}, the index $\lambda$ is simply the eigenvalue of the global charge $Z_{\tot}$, and $\forall\lambda:d_\lambda=1$.\footnote{The irreps $\mc H_\lambda^{\mc C}$ of any Abelian commutant are always one-dimensional.}
In the more general case of a non-Abelian commutant the representations $\mc H_\lambda^{\mc C}$ can have $d_\lambda>1$, and some Krylov subspaces can therefore be degenerate\footnote{When $d_\lambda\defeq\mrm{dim}{\mc H_\lambda^{\mc C}}>1$, there is no unique way to decompose $\mc H_\lambda^{\mc A}\ot\mc H_\lambda^{\mc C}$ into $d_\lambda$ separate Krylov subspaces; these are said to be \textit{degenerate}.}.
For example, if we consider a conventional $SU(2)$ symmetry, the index $\lambda$ is the total spin of states the irrep, so that
\begin{equation}\label{eq:spincomm}
    {S_\mrm{tot}^2}_{|\lambda} = \lambda(\lambda+1)\cdot\1_\lambda \qquad d_\lambda=2\lambda+1.
\end{equation}
In general, $\lambda$ will always be related to the eigenvalues of the operators in the center $\mc Z$. If the Hilbert space contains multiple group irreps with the same value for the quantum numbers $\lambda$, this will correspond to having $D_\lambda>1$.
\subsection{Connectivity of the Hilbert space}\label{subsec:connectivity}
In addition to symmetry groups and their associated quantum number sectors, this framework is also able to capture a more diverse set of symmetries, such as the ones responsible for Hilbert space fragmentation (where $\mrm{dim}(\mc C)$ grows exponentially with the system size)~\cite{moudgalya2021hilbert} and exact quantum many-body scars (which correspond to one-dimensional ($D_\lambda=1$) Krylov subspaces)~\cite{moudgalya2022exhaustive}.
Indeed the structure of commutant algebras provides a general approach to studying the \textit{orbit} of states within a Hilbert space under the repeated action of the terms $\{H_\alpha\}$ (or any operator in the bond algebra).
For example, the states within a given Krylov subspace cannot evolve to states belonging to different subspaces, but they can evolve to \textit{any} state belonging to the same subspace\footnote{The action of the algebra is said to be \textit{transitive} within each Krylov subspace.}.
In other words, Krylov subspaces identify separate sets of states that can never be connected to each other through symmetric time evolution: the dynamics of a given initial state $\ket\psi$ can be determined by only considering the subspaces that are not orthogonal to $\ket\psi$, since its evolution will be restricted to the direct sum of these subspaces (the state $\ket\psi$ is said to ``overlap'' or to have ``non-zero weight'' on such subspaces).
This can be understood easily in the matrix notation of Eqs.~\eqref{eq:matbasis} and \eqref{eq:matrixrep}.
Suppose that for some $\lambda$ and $\ket\gamma_\lambda\in\mc H_\lambda^{\mc C}$, a state $\ket\psi$ is such that
\begin{equation}
    \big(\prescript{}{\lambda}{\bra{\alpha}}\ot\prescript{}{\lambda}{\bra{\gamma}}\big)\ket{\psi}=0, \quad\forall\ket\alpha_\lambda\in \mc H_\lambda^{\mc A}.
\end{equation}
Then due to the structure of the matrices in Eq.~\eqref{eq:matrixrep}, the state $\ket\psi$ can never be evolved to any state in the associated Krylov subspace $\mc H_\lambda^{\mc A}\ot \ket\gamma_\lambda$ through the action of the bond algebra $\mc A$. This property fully characterizes which Krylov subspaces can be accessed by the initial state $\ket\psi$.\footnote{Indeed by performing a Schmidt decomposition along the tensor product $\mc H_\lambda^{\mc A}\ot\mc H_\lambda^{\mc C}$ of the state $\Pi_\lambda\ket\psi=\sum_l b_l\ket{\alpha_l}_\lambda\ot\ket{\gamma_l}_\lambda$, we see that we can evolve $\ket\psi$ to $\ket\alpha_\lambda\ot\ket{\gamma_l}_\lambda$ by choosing $K\in\bond$ such that $M_\lambda(K)=\ket{\alpha}_\lambda\!\bra{\alpha_l}$.}

This is the main feature that renders commutant algebras suitable for the study of universality of operators, which can simply be formulated as a question of connectivity in the Hilbert space of operators.

\section{The Superoperator Algebra Approach to (Non-)Universality}\label{sec:superoperatoralgebra}
%
We now discuss the application of the commutant algebra framework to understand the non-universality of any given set of gates -- or equivalently, to the calculation of the so-called dynamical Lie algebras~\cite{d2007introduction}.
The main difference between the frameworks described here and those in the previous section is that instead of algebras of operators that act linearly on Hilbert space of states, here we will be interested in the algebras generated by superoperators that act linearly on the Hilbert space of operators.
For easy reference, in Tab.~\ref{tab:names} we have summarized the list of symbols that we use in this generalization.
\subsection{Dynamical Lie Algebras (DLAs)}\label{subsec:DLA}
In its most general form, the problem we wish to study can be stated as follows: given a finite set of hermitian operators $\gen =\{h_\alpha\}_{\alpha=1}^N$ acting on a finite-dimensional Hilbert space $\mc H$, we must find the set $\unit$ of all unitary operators generated from these operators,  i.e.,  any $U \in \unit$ of the form (see Fig.~\ref{fig:circuit})
\begin{equation}
    U=\prod_{k} u_{\alpha_k}(\theta_k),\qquad u_\alpha(\theta)\defeq\exp(i\theta h_\alpha),\label{eq:localgates}
\end{equation}
where $\theta_k\in\mb R$,  providing continuous control over the space of unitaries that can be generated from any operator in $\gen$. 
Due to the Baker-Campbell-Haussdorff formula, all unitary operators of this form can be written as
\begin{equation}
    U=\exp(iH),\quad H\in\dlie,\label{eq:DLAvsGroup}
\end{equation}
where $\dlie$ is the Lie algebra\footnote{Although $\dlie$ can be defined as a real Lie algebra (by replacing the generators $\{h_\alpha\}$ with their anti-hermitian counterpart $\{ih_\alpha\}$) it is useful to extend the scalar field to the complex numbers. This operation does not affect our considerations on universality, as long as one takes $H$ to be hermitian in Eq.~\eqref{eq:DLAvsGroup}.} generated by the set $\gen$, which is the vector space spanned by $\gen$ and nested commutators of operators in $\gen$ (i.e. $[h_{\alpha_1},h_{\alpha_2}]$, $[h_{\alpha_1},[h_{\alpha_2},h_{\alpha_3}]]$, etc.).
An important theorem in the theory of quantum control~\cite{d2007introduction} states that the converse holds: all unitaries of the form Eq.~\eqref{eq:DLAvsGroup} belong to $\unit$; in this context, the Lie algebra $\dlie$ is called the ``dynamical Lie algebra'' (DLA) of the set of generators $\gen$. 
By abusing terminology we will sometimes refer to $\gen$ as the gate set itself.
To study the universality of a set of quantum gates, one usually considers a many-body quantum system $\mc H=\mc H_\mrm{loc}^{\ot L}$, which possesses some notion of locality.
For example we may consider the local qudit degrees of freedom to be arranged as a chain; then to produce the quantum gates $u_\alpha(\theta)$ we will consider a spatially homogeneous set of $k$-local generators, i.e. operators that act non-trivially at most on $k$ consecutive qudits (and act as the identity everywhere else).
Circuits of this kind can for example produce unitary time evolutions generated by Hamiltonians composed of local interaction terms (e.g. through the Trotter decomposition).
Although this fundamental problem has been studied from many different points of view \cite{d2007introduction}, in this work we will focus on the recent results that show a link between symmetry and the non-universality of gates~\cite{marvian2020locality, hulse2021qudit, marvian2022rotationally, marvian2024abelian}.
These results imply that when $\{h_\alpha\}$ are chosen to be $k$-local for $k < L$ and symmetric under some types of on-site symmetric unitary Lie groups $G$, then the unitaries of the form of Eq.~(\ref{eq:DLAvsGroup}) do not even exhaust the complete set of \textit{symmetric} unitaries, and are therefore a non-universal gate set.
In order to provide a more general perspective on the origin of this non-universality for arbitrary gate sets, we will focus on the generating set $\gen$ directly, instead of considering all possible symmetric gates for a given group $G$.
In this general setting, non-universality is the situation where the space of generateable operators $\dlie$ is not equal to the space $\bond=\llangle\gen\rrangle$ of operators that are $\comm$-symmetric, i.e. that commute with all operators in $\comm$.
In other words
\begin{equation}
    \gen\mrm{\ is\ not\ universal} \iff \dlie \subsetneq \bond.
    \label{eq:def-non-univ}
\end{equation}
For the associated unitaries, non-universality implies that $\unit$ is a strict subgroup of the group $\unitt$ of all the symmetric unitaries.
In some cases, whenever the set of generated unitaries $\unit$ is \textit{compact}, it is not even possible to approximately obtain the missing unitaries;
this condition is always satisfied when the set $\gen$ consists only of generators $h_\alpha$ that have rational spectra, so that $u_\alpha(\theta)$ will constitute a compact $U(1)$ subgroup of the set of all unitaries.\footnote{See App.~\ref{app:compact} for more details on the question of compactness.}

Past works \cite{marvian2020locality,hulse2021qudit,marvian2022rotationally,marvian2024abelian,kazi2024permutationinvariant} always consider $\gen$ to be the set of \textit{all} $k$-local hermitian generators that are symmetric under some internal or spacetime symmetry group $G$ of the many-body system; in such systems $\comm$ is the associative algebra generated by the generators of the group $G$.
Here we work with general sets of operators $\gen$, which need not have any particular form or range, and while in our case $\comm$ could simply be generated by a group $G$, our statements will hold true also for more unconventional types of symmetries described by a commutant algebra $\comm$ which need not correspond to any group structure~\cite{moudgalya2021hilbert, moudgalya2022exhaustive}.
\begin{figure}[t]
\includegraphics[width=0.91\columnwidth]{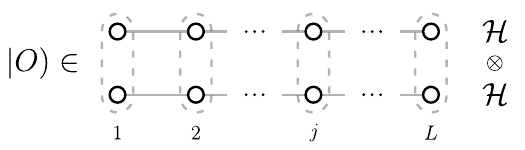}
\caption{The operator Hilbert space $\hend \defeq \mrm{End}(\mc H)$ can be interpreted as a ladder Hilbert space through the Liouvillian isomorphosm with $\mc H\ot\mc H$. \label{fig:ladder}}
\end{figure}

\subsection{Liouvillian Superoperators}
To express the nested commutators that are used to generate the DLA $\dlie$ more compactly, we introduce the adjoint map $\ad K$, that acts on any operator $O$ as
\begin{equation}
    \ad K \cdot O \defeq [ K , O ].
\label{eq:adjoint}
\end{equation}
In this notation, $\dlie$ is the vector space spanned by operators of the form $\left(\prod_k \ad{h_{\alpha_k}}\right)\cdot h_{\beta}$.
Objects of the form $\ad K$ can be represented as \textit{superoperators}, since they are linear operators that act on operators of a physical Hilbert space.
Such superoperators are often indicated in the Liouvillian notation, whereby operators are represented as states on a doubled Hilbert space through an  isomorphism:
\begin{equation}
    \begin{gathered}
        O=\sum_{\mu_1,\mu_2} O_{\mu_1\mu_2}\ketbra{\mu_1}{\mu_2} \\ \oket{O}\defeq\sum_{\mu_1,\mu_2} O_{\mu_1\mu_2} \ket{\mu_1}\!\ket{\mu_2},
    \end{gathered}\label{eq:liouvillian}
\end{equation}
where $\{\ket\mu\}$ is a given orthonormal basis for $\mc H$.
The vector space $\hend\defeq\mrm{End}(\mc H)$ of all operators $\oket{O}$ is a Hilbert space, and the inner product of two operators is defined as $\obraket{O_1}{O_2}=\mrm{tr}(O_1\+O_2)$.
In this notation the adjoint map of Eq.~(\ref{eq:adjoint}) can be written as:
\begin{equation}\label{eq:liouvillnot}
\begin{gathered}   
    \ad{K} = \ K\ot \1-\1\ot K^T,\\
    \ad{K}\oket O=\ \oket{[K,O]},
\end{gathered}
\end{equation}
since $K\ot\1\oket O=\oket{KO}$ and $\1\ot K^T\oket O=\oket{OK}$.
When taking locality into account, it is convenient to interpret the operator Hilbert space $\hend$ as a ladder~\cite{moudgalya2023symmetries}, where at each site $j$ one finds two copies of the local Hilbert space $\mc H_\mrm{loc}$ (see Fig.~\ref{fig:ladder}).
The adjoint operator $\ad{K}$ is the sum of an operator acting on the top leg (i.e. $K\ot\1$) and of an operator acting on the bottom leg (i.e. $-\1\ot K^T$).

\begin{table}[t!]
    \centering
    \renewcommand{\arraystretch}{1.5}
    \begin{tabular}{c c c}
         & \textbf{Operator} & \textbf{Superoperator} \\ \hline
         \textit{Hilbert space} & $\hend\defeq\mrm{End}(\mc H)$ & $-$ \\
         \textit{Generating set} & $\gen \defeq\{h_\alpha\}$ & $\{\ad{h_\alpha}\defeq[h_\alpha,\bullet]\}$ \\
         \textit{Bond algebra} & $\bond \defeq \llangle\{h_\alpha\}\rrangle$ & $\sbond\defeq\llangle\{\ad{h_\alpha}\}\rrangle$ \\
         \textit{Lie algebra} & $\dlie\subseteq\bond$ & $-$ \\
         \textit{Commutant} & $\comm\defeq\mrm{comm}(\bond)$ & $\scomm\defeq\mrm{comm}(\sbond)$ \\ \hline
    \end{tabular}
    \caption{Summary of the symbols introduced in this work. We essentially generalize the concepts of bond and commutant algebras for operators (consisting of symmetries) to super-bond and super-commutant algebras of superoperators (the latter consisting of superoperator symmetries).}
    \label{tab:names}
\end{table}

Working in this superoperator language is the crucial step which will allow us to use the strength of the commutant algebra formalism of Sec.~\ref{sec:bondcommutant} to study dynamical Lie algebras.
We call $\sbond$ the associative algebra generated by the adjoint action of the generators\footnote{We note that $\sbond$ can also be defined as the universal enveloping algebra of $\dlie$ in the representation $\ad{(\bullet)}$. Therefore if two sets of generators are such that $\gen\neq\gen'$ but $\dlie=\dliex{\gen'}$, all our considerations regarding these algebras will be identical (see Lemma \ref{lem1} in App.~\ref{sec:app-math})}
\begin{equation}
	\sbond \defeq \llangle\{\ad{h_\alpha}\}_{h_\alpha\in\gen}\rrangle,\label{eq:sbond-def}
\end{equation}
and we call $\scomm$ its commutant\footnote{These are not to be confused with the $\bond$ and $\comm$ algebras at operator level.}
\begin{equation}
	\scomm \defeq\mrm{comm}(\sbond)= \{\mc Q:\, [\ad{h_\alpha},\mc Q]\ \forall \alpha\},\label{eq:scomm-def}
\end{equation}
see Fig.~\ref{fig:venn-diagrams} for an illustration of how these sets relate to each other, and their action on operators in the original bond algebra $\mA_{\mc G}$.
A summary of the symbols introduced so far is shown in Table \ref{tab:names}.
It is easy to see that the Lie algebra generated by $\gen$ is obtained by acting with elements of $\sbond$ (which are compositions of adjoint maps in the original Hilbert space, now denoted as superoperators that act as operators on a doubled Hilbert space) on linear combinations of elements of $\gen$ (which are operators in the original Hilbert space, now denoted as states on a doubled Hilbert space), or in other words:
\begin{equation}\label{eq:DLAgeneration}
    \dlie=\{\mc K\,\oket H:\,\mc K\in\sbond,\ H\in\mrm{span}(\gen)\}.
\end{equation}
Due to the fact that these are sets of superoperators, we will call $\sbond$ and $\scomm$ the \textit{super-bond algebra} and the \textit{super-commutant} respectively\footnote{These sets can actually be shown to be {finite-dimensional} von Neumann algebras under \textit{two} independent product operations, cf.  App.~\ref{sec:app-manycopy}.}.
In the following section we directly show how this equation leads to a structural connection between superoperator symmetries and the DLA $\dlie$. A concrete example (i.e. the non-universality of local $U(1)$ gates) is later discussed in detail in Sec.~\ref{subsec:U1example}.
\subsection{Methodology}\label{sec:methodology}
\begin{figure*}[t]
\stackon[5pt]{\includegraphics[width=0.47\textwidth]{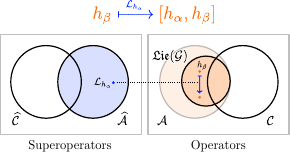}}{\textsf{(a)}}
\hspace{4em}
\stackon[5pt]{\includegraphics[width=0.3\textwidth]{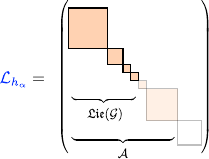}}{\textsf{(b)}}
\caption{Diagram of the operator and superoperator algebras under consideration. For simplicity, we have omitted most references to the generators $\gen$. (a) The superoperators in $\sbond$ are generated by the adjoint superoperators $\ad{h_\alpha}$, and their action on the generators $\{h_\beta\}$ spans the DLA $\dlie$.
The algebra generated by $\{h_\beta\}$ is $\bond$, and its commutant $\comm$ is the same as $\dlie$'s commutant; $\comm$ is the set of conventional symmetries of the set of gates. (b) By writing an element of $\sbond$ according to the general block-decomposition of equation \eqref{eq:matrixrep}, we can see that $\dlie$ and $\bond$ are a direct sum of Krylov subspaces. This is because: both algebras are invariant under the action of $\sbond$, and the action of $\sbond$ is irreducible on every Krylov subspace. \label{fig:venn-diagrams}}
\end{figure*}
Eq.~\eqref{eq:DLAgeneration} turns our original question regarding DLAs into a connectivity problem: we wish to ask how the operators in $\gen$ evolve under the ``dynamics" described by the action of superoperators in the super-bond algebra $\sbond$.
This involves determining the invariant subspaces (also referred to as Krylov subspaces) of the super-bond algebra, and determining the subspaces in which at least some operators in $\gen$ have non-zero overlap.
The direct sum of these subspaces is precisely $\dlie$, and hence the origin of any non-universality can be attributed to the properties of these invariant subspaces, which in turn are understood using the superoperator symmetries in the super-commutant $\scomm$.
This is completely analogous to what happens at the state level for quantum systems discussed in Sec.~\ref{subsec:connectivity}, where the orbit of states under the ``dynamics" described by the action of operators in the bond algebra $\bond$ is understood using its invariant subspaces and the symmetries in its commutant $\comm$.
We start by characterizing $\sbond$ and $\scomm$, and by decomposing the Hilbert space of operators $\hend$, according to the fundamental theorem \eqref{eq:fund-th}:
\begin{equation}
	\hend =\bigoplus_{{\widehat\lambda}}\left(\hend_{\widehat\lambda}^{\sbond}\ot\hend_{\widehat\lambda}^{\scomm}\right).\label{eq:hilbdecend}
\end{equation}
We then need to identify all the subspaces in which at least one of the generators $\oket{h_\alpha}$ in $\gen$ has a non-zero weight.
All the operators belonging to such subspaces can be generated by repeated actions of the Liouvillians $\{\mc L_{h_\alpha}\}$, and hence they span the dynamical Lie Algebra $\dlie$.  
To do so we define $\mc P_{\widehat\lambda}\in\scent$ to be the projection superoperator onto the subspace labelled by ${\widehat\lambda}$ and, within each $\widehat\lambda$, we identify the smallest vector space $\widehat V_{\widehat\lambda}\subseteq\hend_{\widehat\lambda}^{\scomm}$ such that
\begin{equation}\label{eq:project-to-tensor}
	\forall \oket{h_\alpha}\in\gen: \mc P_{\widehat\lambda}\oket{h_\alpha}\in \left(\hend_{\widehat\lambda}^{\sbond} \ot \widehat V_{\widehat\lambda}\right).
\end{equation}
If we choose a basis $\{\oket{v}_{\widehat\lambda}\}_v$ of $\widehat V_{\widehat\lambda}$, we can think of $\{\hend_{\widehat\lambda}^{\sbond}\ot\oket v_{\widehat\lambda}\}_{{\widehat\lambda},v}$ as the set of subspaces on which the generators $\gen$ have non-zero weight. Finally, thanks to the irreducilibity of the action of $\sbond$ on $\hend_{\widehat\lambda}^{\sbond}$, we find that $\dlie$ is the direct sum of these subspaces
\begin{equation}
    \dlie=\bigoplus\!\,_{\widehat\lambda}\left(\hend_{\widehat\lambda}^{\sbond} \ot \widehat V_{\widehat\lambda}\right).
\end{equation}
Note that in the decomposition of Eq.~(\ref{eq:hilbdecend}), $\hend$ contains both symmetric and non-symmetric operators, but for the calculation of the DLA we are only interested in computing the connectivity of the set of symmetric operators since the generators have no weight outside of this subspace.
The set of symmetric operators is simply $\bond$, which is invariant under the action of $\sbond$.
Hence we can always choose to restrict the scope of our Hilbert space decomposition in Eq.~\eqref{eq:hilbdecend} from $\hend$ to $\bond$ only.
\subsection{Numerical Methods}
The formulation of this problem in terms of commutant algebras also allows us to use numerical methods introduced in Ref.~\cite{moudgalya2023numerical} to compute the super-commutant algebra $\scomm$ or even the whole decomposition \eqref{eq:hilbdecend} of the Hilbert space for small systems. 
Since proving the complete structure of $\scomm$ can be challenging even for simple gate sets $\gen$, these numerical methods end up being a very useful tool for obtaining partial information on the problem.
A detailed discussion of them can be found in App.~\ref{sec:app-num}.
One method obtains an Matrix Product State (MPS)~\cite{perezgarcia2007matrix, cirac2021matrix} representation for $\scomm$ by using the fact that it can be expressed as the ground states of a frustration-free Hamiltonian; the time complexity of this method scales with the dimension of $\scomm$, and it is therefore useful when this grows algebraically as a function of the system size $L$.
Another method explicitly builds the block-decomposition \eqref{eq:matrixrep} for elements of $\sbond$, and its complexity scales exponentially with $L$, and is hence useful for small system sizes.
This method can however be specialized to the problem of calculating the DLA by restricting the space of interest from the total operator Hilbert space $\hend$ to the set of symmetric operators $\bond$, or even to the set of operators acting on a single symmetry sector $\mc H_\lambda$ of the physical Hilbert space. 
For gate sets which possess a large commutant, this step can greatly reduce the complexity of the calculations.
From the block-decomposition the DLA and its dimension can then be computed directly, by calculating the overlap between each block and the generators.

Since these techniques work with superoperator symmetries, their performance is not affected by the size of the DLA $\dlie$. This fact is especially useful for the MPS method, as it performs well when the super-commutant $\scomm$ is small -- a situation that typically corresponds to large DLAs. Since typically a gate set is expected to have a certain minimal amount of superoperator symmetries, this allows one to rapidly decide whether the given gate set has some special features or not (i.e. \textit{strong non-universality}, see Sec.~\ref{sec:nonuniversality}).
\subsection{Example: $U(1)$ symmetric circuits}\label{subsec:U1example}
We now illustrate the origin of non-universality in the commutant language, while using $U(1)$ symmetric circuits as a concrete example, which was first shown to be non-universal by Marvian~\cite{marvian2020locality}.
This serves as an introduction for the more general phenomenology of non-universality discussed in the subsequent sections.
Let us consider the case of $2$-local $U(1)$-symmetric gates on a chain of $L$ qubits, which can be generated using exponentials of the following set of generators:
\begin{equation}
    \gen=\gen_{U(1)} \defeq \{X_j X_{j+1} + Y_j Y_{j+1},\, Z_j Z_{j+1},\, Z_j\}_{j=1,...,L}
\label{eq:U1generators}
\end{equation}
where for simplicity we chose periodic boundary conditions $L+1=1$. From Eq.~\eqref{eq:liouvillnot}, it is clear that any operator of the form $Q_1\ot Q_2^T$ for $Q_{1,2}\in\comm= \commx{U(1)} = \llangle\{N_\mrm{tot}\}\rrangle$ [see Eq.~(\ref{eq:CU1}) for an expression of $\commx{U(1)}$] belongs to the super-commutant $\scommx{U(1)}$.
In addition, for any choice of $\gen$, the projector $\oketbra{\1}{\1}$ also belongs to $\scomm$ since
\begin{equation}
    \begin{gathered}
        \ad{K}\oket{\1}=\oket{[K,\1]}=0\\
        \implies[\ad{K},\oketbra{\1}{\1}]=\ad{K}\oketbra{\1}{\1}-\oketbra{\1}{\1}\ad{K}=0.
    \end{gathered}
    \label{eq:liouvannihilate}
\end{equation}
These operators generate the full super-commutant in the $U(1)$ case (see App.~\ref{sec:app-nplusoneproof} for a proof):
\begin{equation}
    \scommx{U(1)}=\llangle(\commx{U(1)}\ot\commx{U(1)}^T)\cup\{\oketbra{\1}{\1}\}\rrangle.
\label{eq:U1supercomm}
\end{equation}
Note that due to algebraic closure, this also means that operators of the form $\oketbra{Q_2}{Q_1}$ for $Q_{1,2} \in \commx{U(1)}$ also belong to $\scommx{U(1)}$, since $\oketbra{Q_2}{Q_1} = Q_2 \ot \1 \oketbra{\1}{\1}Q_1\ot\1$.
In the $U(1)$ case, since the projectors $\{\Pi_n\}_{n=0,...,L}$ onto the charge sectors $N_\mrm{tot}=n$ (as in Eq.~\eqref{eq:U1basis}) form a basis for the commutant $\commx{U(1)}$, a linear basis for the super-commutant $\scommx{U(1)}$ is given by
\begin{equation}
    \{\Pi_n\ot\Pi_m,\oketbra{\Pi_n}{\Pi_m}\}_{n,m=0,...,L}.
\end{equation}
Note that $\Pi_n\ot\Pi_n=\oketbra{\Pi_n}{\Pi_n}$ when $n=0$ or $n=L$, so we have $\mrm{dim}(\scommx{U(1)})=2L^2+4L$.

As we will discuss in the next section, this is the simplest possible structure for the super-commutant: the $\commx{U(1)}\ot\commx{U(1)}^T$ part is directly inherited from the $U(1)$ symmetry of the gates, while $\oketbra{\1}{\1}$ is the projector onto a one-dimensional Krylov subspace spanned by $\oket{\1}$ that is invariant under the action of the super-bond algebra generators $\ad{K}\in\sbondx{U(1)}$. 
Since one-dimensional Krylov subspaces in the physical Hilbert space are referred to as quantum many-body scars~\cite{moudgalya2022exhaustive}, we sometimes refer to projectors such as $\oketbra{\1}{\1}$ as \textit{scar projectors}.
Indeed, the operators $\oket{Q}$ for $Q\in\comm$ can be interpreted as quantum many-body scars in operator space, i.e., w.r.t. the decomposition of Eq.~(\ref{eq:hilbdecend}), which is consistent with the fact that symmetry operators can be interpreted as frustration-free ground states of local superoperators~\cite{moudgalya2023symmetries}.
As we will show, these scars are responsible for non-universality, while the $\commx{U(1)} \otimes \commx{U(1)}^T$ part of the super-commutant simply distinguishes symmetric operators from non-symmetric ones (hence identifying the subspaces spanned by $\bondx{U(1)}$).

Let us focus on the conventional part of the super-commutant first.
The Hilbert space of operators can be decomposed according to the $\commx{U(1)}\ot\commx{U(1)}^T$ quantum numbers as follows:
\begin{equation}
    \begin{gathered}
        \hend=\bigoplus_{n,m=0}^L \hend_{n,m},\\
        \hend_{n,m}=\{O\in\hend:\ N_\mrm{tot}O=nO,\ ON_\mrm{tot}=mO\}.
    \end{gathered}\label{eq:nmdec}
\end{equation}
In other words, operators in $\hend_{n,m}$ are eigenvectors of both he left and right action of the commutant $\commx{U(1)}$ on the operator space $\hend$. These operators map states with $m$ spin-ups to states with $n$ spin-ups.
The set of all $U(1)$-symmetric operators, which is also the bond algebra $\bondx{U(1)}$, are those whose actions preserve the charge, and hence are obtained by restricting the above sum to $n=m$.

In terms of irreps of the full $\sbondx{U(1)}\ot\scommx{U(1)}$ (given by the fundamental theorem \eqref{eq:fund-th}), each $\hend_{n,n}$ further splits into a scar $\{\oket{\Pi_n}\}$ and its orthogonal complement $\hend_n^*$.
This can be seen by a reasoning similar to Eq.~\eqref{eq:liouvannihilate}.
We therefore have \textit{one scar per symmetry sector} of the physical quantum system, which we find is a general feature of Abelian symmetries.
Since the Liouvillians $\{\ad{h_\alpha}\}$ for $h_\alpha \in \gen_{U(1)}$ annihilate the scars (see Eq.~(\ref{eq:liouvannihilate})), they are degenerate one-dimensional representations and span the subspace $\hend_\mrm{scar}=\mrm{span}(\{\oket{\Pi_n}\}_{n=0,...,L})$.
The decomposition of the bond algebra $\bondx{U(1)}$ (viewed as a Hilbert space) is then
\begin{equation}
    \bondx{U(1)} = \left(\bigoplus_{n=1}^{L-1}\hend_n^*\right)\oplus \hend_\mrm{scar}.
\label{eq:AU1decomp}
\end{equation}
Note that since $\dliex{\gen_{U(1)}}\subseteq\bondx{U(1)}$, we can restrict our attention only to these subspaces of symmetric operators, see Fig.~\ref{fig:venn-diagrams}b.
Furthermore, using the known dimensions of the $U(1)$ charge sectors on a spin-1/2 Hilbert space, we can deduce that:
\begin{enumerate}[(i)]
    \item The subspaces $\hend_n^*$ are the tensor product of a one-dimensional irrep of $\scommx{U(1)}$ and of a $(\binom{L}{n}^{2}-1)$-dimensional irrep of $\sbondx{U(1)}$. 
    \item The scar space $\hend_\mrm{scar} =\commx{U(1)}$ is the tensor product of an $(L+1)$-dimensional irrep of $\scommx{U(1)}$ and a one-dimensional irrep of $\sbondx{U(1)}$ (with zero eigenvalue). We say that there are $L+1$ degenerate scars.
\end{enumerate}
Having characterized the decomposition of $\bondx{U(1)}$, we now study the overlap of the generators of Eq.~(\ref{eq:U1generators}) with each subspace. Notice that the operator $Z_j\in\gen_{U(1)}$ has non-zero projection on all $\hend_{n}^*$ -- since it is traceless and non-zero within each $n$-particle sector $\mc H_n$ of the physical Hilbert space. 
This implies that the dynamical Lie algebra will contain all the operators in $\bigoplus_n\hend^*_n$ (this is proven more in general in Lemma \ref{lem3}).
Finally, to study the overlap of the generators with the one-dimensional degenerate subspaces in $\hend_\mrm{scar}$, it is convenient to use the following Pauli-string basis for the scar space:
\begin{equation}\label{eq:paulibasisu1}
    \begin{gathered}
    \hend_\mrm{scar}=\mrm{span}\left(\{\oket{\Sigma_n}_{n=0,...,L}\}\right)\\
    \Sigma_0=\1 \qquad\Sigma_{n\geq 1}\defeq\sum_{j_1<...<j_n}Z_{j_1}\cdot...\cdot Z_{j_n}
    \end{gathered}
\end{equation}
One can easily verify that the $2$-local generators of Eq.~(\ref{eq:U1generators}) only have a non-zero overlap\footnote{The generators also do not overlap with the identity operator $\oket{\Sigma_0}=\oket{\1}$, but since this is just the operator responsible for global phases, we can add it to our set of generators without loss of generality.} with $\oket{\Sigma_n}$ for $n\leq 2$, the set of realizable unitaries $\unit$ will have co-dimension $L-2$ in the set of symmetric unitaries $\unitt$.
More generally, \textit{any} choice of generators for the $U(1)$ bond algebra that contains at most $k$-local operators would have a non-zero overlap only with $\oket{\Sigma_n}$ for $n \leq k$, therefore the set of realizable unitaries would have co-dimension at least $L-k$.
The ``missing unitaries'' $\exp(i\theta \Sigma_n)$ for $n>k$ are operators that commute with all other operators in $\unitt$ and that simply give different relative phases to each $N_\mrm{tot}=m$ sector of the Hilbert space:
\begin{equation}
    \exp(i\theta \Sigma_n)=\prod_{m=0}^L \exp\left[i\theta\cdot\frac{\obraket{\Pi_m}{\Sigma_n}}{\binom{L}{m}}\cdot\Pi_m\right].
\label{eq:U1missing}
\end{equation}
This situation, in which all symmetric operations can be performed up to a charge-sector dependent phase, has been dubbed ``semi-universality'' in the literature \cite{kazi2024permutationinvariant,marvian2024abelian} and is discussed in greater generality in Sec.~\ref{subsec:semiuniversal}.
Note that these results can be straightforwardly generalized to arbitrary sets of generators $\gen$ that are $U(1)$ symmetric (i.e., $\comm = \commx{U(1)}$), such as non-local terms that have support on $k$ sites.
The co-dimension of the realizable unitaries in that case is also at least $L - k$, and the missing unitaries are again of the form Eq.~(\ref{eq:U1missing}). 
If the chosen set of generators produces a larger super-commutant $\scomm$ than the one in Eq.~\eqref{eq:U1supercomm} (e.g. the example in Sec.~\ref{sec:mgu1}), the decomposition of the super-bond algebra in Eq.~\eqref{eq:AU1decomp} will become finer, and the generated DLA will be missing larger subspaces beyond the ones given by scars $\oket{\Sigma_n}$, but this lower bound on the co-dimension will still apply.
Finally, if the chosen set of generators produces a group $\unit$ that is compact, then it is not even possible to approximate unitaries that cannot be generated exactly.
This is always the case when each generator $h_\alpha$ has rational spectrum, which is usually the physically relevant situation.
On the contrary, if $\unit$ is not compact, the closure of the group $\unit$ within the set of all unitaries might therefore include unitaries of the form $e^{i\theta Q}$ for some scars $\oket{Q}\in\commx{U(1)}$ that do not necessarily belong to the DLA $\dlie$, which hence may be approximated to arbitrary precision.
\section{Two Classes of Non-Universality}\label{sec:nonuniversality}
%
\begin{figure}[t]
\includegraphics[width=0.85\columnwidth]{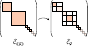}
\caption{The minimal super-commutant $\scommt$ describes the block decomposition associated to the conventional symmetries at operator level $\comm$. If $\scomm$ contains additional superoperator symmetries, they will be responsible for further splitting of the operator Hilbert space into smaller Krylov subspaces and/or for the formation of additional degeneracies between Krylov subspaces belonging to different symmetry sectors. \label{fig:blocks}}
\end{figure}
\subsection{The General Picture}
We now illustrate the general structure of the algebras corresponding to a general set of generators $\mathcal{G}$.
As we showed in the example of a $U(1)$ symmetric circuit, in general we can see that the super-commutant $\scomm$ is composed of the following two types of operators
\begin{equation}
    \forall Q_1, Q_2 \in \comm:\;\; 
    \text{(i)}\ Q_1 \ot Q_2^T \in \scomm,\;\;\text{(ii)}\ \oketbra{Q_2}{Q_1}\in\scomm,\label{eq:forsure}
\end{equation}
$\comm$ is the commutant (symmetries) of $\gen$, and we have assumed $Q_1$, $Q_2$ to be hermitian.
Stated in operator language, property (i) indicates that 
for any symmetric operator $K\in\bond$, the superoperator $\ad K$ commutes with left and right multiplication of symmetry operators (i.e. $Q_1 [K, \,\bullet\,] Q_2 =  [K, Q_1 \!\!\,\bullet\,\!\! Q_2]$).
Property (ii) follows from the fact that the symmetry operators commute with all symmetric operators (i.e., $\ad K\oket{Q}=\oket{[K,Q]} = 0$).
In the $U(1)$ example discussed in Sec.~\ref{subsec:U1example}, superoperators of these two forms generate the whole super-commutant algebra, but as we will show, in general more superoperators may be present: when they occur, these additional superoperator symmetries are responsible for the more dramatic obstructions to universality.
We define the \textit{minimal super-commutant} to be the algebra generated by the superoperators of Eq.~\eqref{eq:forsure}, which can be denoted as
\begin{equation}
	\scommt \defeq \llangle(\comm\ot \comm^T)\cup\{\oketbra{\1}{\1}\}\rrangle.
\label{eq:supercommtrivial}
\end{equation}
Note that including $\oketbra{\1}{\1}$ is sufficient to generate all operators of the form (ii) in \eqref{eq:forsure}, as also discussed in the $U(1)$ case below Eq.~\eqref{eq:U1supercomm}.
To it we can associate the \textit{maximal super-bond algebra}, generated by the adjoint action of \textit{all} symmetric operators -- not just the local ones:
\begin{equation}\label{eq:sbont}
	\sbondt \defeq \llangle\{\ad{K}\}_{K\in\bond}\rrangle = \mrm{comm}(\scommt).
\end{equation}
The second equality is proven in App.~\ref{sec:app-math} as Lemma \ref{lem2}; it states that the minimal super-commutant is the one associated to generic $\comm$-symmetric unitaries, and hence to the  symmetries of the unitaries (indeed it is defined as the set of all superoperator symmetries inherited from the original symmetries).

Given these definitions, a non-universal gate set (defined in Eq.~\eqref{eq:def-non-univ}) can lie in one of two main classes, which can be defined in the following way:
\begin{itemize}
    \item \textbf{\textit{Weak Non-Universality}}, where $\scomm=\scommt$;
    \item \textbf{\textit{Strong Non-Universality}}, where $\scomm\supsetneq\scommt$.
\end{itemize}
The first kind is simply due to vanishing overlaps between the generators $\mathcal G$ and some elements of the center $\cent$ (which we refer to as scar operators).
This situation appears to hold for many of symmetric gates we study, including in the $U(1)$ case illustrated earlier in Sec.~\ref{subsec:U1example}; we discuss this in Sec.~\ref{subsec:semiuniversal}.
The second provides richer possibilities, although it is restricted to a smaller number of gate sets, e.g. matchgate circuits, free-fermion Hamiltonians, or 2-local $SU(d)$-symmetric circuits on qudit chains \cite{hulse2021qudit}; we discuss this in Sec.~\ref{sec:con-super}. 
In the following sections our goal will be to provide a broad overview of weak and strong non-universality using many relevant examples.
In many cases, we leave the providing precise characterizations of the commutants $\scomm$ with rigorous proofs to future work, but we will point out the cases in which such a characterization exists.

These two possibilities also appear in Ref.~\cite{zimboras2015symmetry}, which studied the conditions under which a set of generators $\gen$ is able to fully simulate a larger set of generators $\gen\subseteq\gen'$. 
It is shown that this is possible if and only if the ``quadratic symmetries''~\cite{zeier2011symmetry} are equal\footnote{Quadratic symmetries are related to superoperator symmetries through the correspondences shown in Sec.~\ref{sec:manycopy}.} and the projections of the generators onto their centers are also equal.
Considering $\mG' = \lgen \mG \rgen$ in their analysis, the first condition fails in strongly non-universal systems, while only the second condition fails in weakly non-universal systems.
In App.~\ref{app:types}, we provide an overview of the relation of these types of universality proposed in Ref.~\cite{marvian2024abelian}, referred to as types I-IV.
In particular, it is possible to show that semi-universality as defined in Ref.~\cite{marvian2024abelian} follows from weak non-universality.
However, since these constraints can appear in combination with each other, for the remainder of this work we will stick to the two-fold classification into weak and strong non-universality.
\subsection{Weak Non-Universality}\label{subsec:semiuniversal}
As we discuss below, and prove in Lemma \ref{lem3}, the condition of weak non-universality leads to the situation in which the $\dlie$ differs from the algebra $\bond$ of \textit{all} symmetric operators only by a few central elements, i.e. operators $\oket Z\in\cent$ (cf. Eq.~\eqref{eq:matrixrepcenter}).
Such operators are general linear combinations of projectors $\oket{\Pi_\lambda}$ onto the irreps $\mc H_\lambda = \mc H_\lambda^{\bond}\ot\mc H_\lambda^{\comm}$ in the decomposition Eq.~\eqref{eq:fund-th}, and the unitaries obtained when these are exponentiated are simply the ones that give different relative phases to each symmetry sector $\mc H_\lambda$.
In this case the decomposition of $\bond$ is analogous to that of Eq.~(\ref{eq:AU1decomp}) in $U(1)$-symmetric circuits, where the only irrep responsible for non-universality is $\hend_{\text{scar}}$, which consists of operators from $\cent$.
We discuss this in the next section for general weak non-universal gate sets, which need not have $\cent = \comm$ as in the Abelian case.
\subsubsection{Counting Missing Dimensions for Weakly Non-Universal Systems}
We now wish to describe the physical and operator Hilbert space decompositions of a general -- possibly non-Abelian -- commutant $\comm$, for a weakly non-universal gate set $\gen$ (i.e. we assume $\scomm=\scommt$).
The physical Hilbert space $\mc H$ splits into $\bond\ot\comm$ irreps labeled by $\lambda$ (see Eq.~\eqref{eq:fund-th}).
Similarly to Eq.~(\ref{eq:nmdec}), the operator Hilbert space $\hend$ first splits into irreps labeled by the pair of quantum numbers $(\lambda,\lambda')$ under the action of $\comm\ot\comm^T$ (which is only part of the whole super-commutant $\scommt$ of Eq.~\eqref{eq:supercommtrivial}).\footnote{An operator $\oket{O}$ belongs to a given $(\lambda,\lambda')$ irrep if and only if for all $\ket\psi\in\mc H_\lambda$ it satisfies $O\ket\psi\in\mc H_{\lambda'}$ and for all $\ket\psi\in\mc H_\lambda^\perp$  it satisfies $O\ket\psi=0$.}
The operators in sectors $(\lambda, \lambda')$ for $\lambda \neq \lambda'$ are off-diagonal in the decomposition of Eq.~\eqref{eq:fund-th}, and thus cannot belong $\bond$, since they do not admit either of the forms of Eq.~(\ref{eq:matrixrep}).
Each $(\lambda,\lambda)$ sector is spanned by operators of the form $\oket{M_\lambda\ot N_\lambda}$.
Therefore under the action of $\sbond$, each of these sectors splits into $d_\lambda^2$ degenerate subspaces of dimension $D_\lambda^2$ (according to the notation of Eq.~\eqref{eq:matrixrep}).
Further, within each of these degenerate subspaces, there is a one-dimensional Krylov subspace (i.e. a scar) given by $\oket{\1_{D_\lambda}\ot N_\lambda}\in\comm$.
This is simply due to the fact that by definition all elements of $\comm$ commute with all generators in $\gen$, that is $\ad{h_\alpha}\oket{Q}=0$ for $\oket Q\in\comm$. Therefore $\comm\subseteq\hend$ constitutes a single irrep composed of one-dimensional Krylov subspaces in the decomposition of Eq.~\eqref{eq:hilbdecend}.
However, note that for the purposes of non-universality, we are interested only in the decomposition \textit{within} the space of symmetric operators $\bond$. 
The operators in the bond algebra $\bond$ within each $(\lambda, \lambda)$ sector are of the form $\oket{M_\lambda\ot\1_{d_\lambda}}$, and the only one-dimensional invariant subspace that is also in $\bond$ is the projector $\oket{\Pi_\lambda}=\oket{\1_{D_\lambda}\ot\1_{d_\lambda}}$.
These are simply the elements of the center $\cent\defeq\bond\cap\comm$ (cf. Eq.~(\ref{eq:matrixrepcenter})).
Hence the situation for more general gates, including non-Abelian ones, remains analogous to the ``\textit{one scar per symmetry sector}'' story from the $U(1)$ case (cf. Sec.~\ref{subsec:U1example}), if we define a symmetry sector to be any of the subspaces labeled by $\lambda$ in Eq.~\eqref{eq:fund-th}.
For a more detailed description of the operator Hilbert space decomposition when $\scomm=\scommt$ see App.~\ref{sec:app-scommt}.

Following the methods in Sec.~\ref{sec:methodology}, to find $\dlie$ one simply needs to study the overlaps between the generators $\gen$ and the invariant subspaces of $\hend$.
One can in general show that when $\scomm=\scommt$ the DLA $\dlie$ contains all traceless symmetric operators (i.e. all operators in $\bond$ up to central elements $\oket Z\in\cent$); this is proven in App.~\ref{sec:app-math} as Lemma \ref{lem3}.
Hence the co-dimension of the controllable manifold within the space of symmetric unitaries is the number of linearly independent elements of the center $\oket Z\in\cent$ which do not overlap with the generators $\gen$.\footnote{Note that technically speaking, the identity operator $\oket\1$ always belongs to the center, and never overlaps with traceless generators, but since it is responsible for global phases $u_\1(\theta)=e^{i\theta}$, we will never consider its absence in the $\dlie$ as a loss of universality.}
In other words, the vector space $\bond$ of symmetric operators splits up as
\begin{equation}\label{eq:semiunivformula}
    \bond=\dlie\oplus\{\oket Z\in\cent: \obraket{h_\alpha}{Z}=0,\ \forall h_\alpha\in\gen\}.
\end{equation}
where the second term is simply the orthogonal complement of $\dlie$ within the center $\cent$.
In terms of dimensions, we then have
\begin{multline}\label{eq:dimformula}
    \mrm{dim}(\bond)-\mrm{dim}(\dlie) \\= \mrm{dim}\{\oket Z\in\cent: \obraket{h_\alpha}{Z}=0,\ \forall h_\alpha\in\gen\}.
\end{multline}
Calculating the dimension of this space is equivalent to calculating, for an arbitrary basis $\{\oket{Z_l}\}$ of the center $\cent$, the overlap matrix $S$, which has elements $S_{l\alpha}\defeq \obraket{h_\alpha}{Z_l}$ and finding its rank:
\begin{equation}\label{eq:dimformulark}
    \mrm{dim}(\bond)-\mrm{dim}(\dlie)=\mrm{dim}(\cent)-\mrm{rk}(S).
\end{equation}
Notice that the rank of the overlap matrix $\mrm{rk}(S)$ is bounded from above by the number of linearly independent generators $\oket{h_\alpha}\in\gen$. 
The equality can also hold more generally, even when $\scomm$ is not necessarily equal $\scommt$, and this condition (i.e. any of Eqs.~(\ref{eq:semiunivformula}-\ref{eq:dimformulark})) is referred to in the literature as \textit{semi-universality}~\cite{marvian2024abelian}, where only the relative phases between symmetry sectors cannot be fully controlled.
{Hence weak non-universality is a subclass of semi-universality.
However, in general the R.H.S. of Eqs.~(\ref{eq:dimformula}-\ref{eq:dimformulark}) is still valid as a lower bound on the co-dimension of $\dlie$ in $\bond$, which also applies to strongly non-universal systems.}
Similar expressions (albeit originating from a different analysis) have appeared in Ref.~\cite{marvian2020locality, marvian2024abelian} to ultimately prove a lower bound on the number of degrees of freedom which cannot be controlled by local group-symmetric gates.
However, note that Eq.~\eqref{eq:dimformula} is completely general, and applies to \textit{arbitrary} sets of gates $\gen$.
In addition, the language of commutants (through Eq.~\eqref{eq:dimformula}) also enables us to compute these dimensions \textit{without} any knowledge of the representation theory of the symmetry algebra $\comm$, only the expression of the elements in the $\cent$.
Choosing an appropriate basis for the center $\cent$ can also greatly simplify the calculation of the overlaps in Eq.~\eqref{eq:dimformula}, and computing the co-dimension of the controllable manifold $\unit$.
For example, a Pauli-string basis has been exploited in this way in the $U(1)$ example of Sec.~\ref{subsec:U1example}, and we will show similar examples below.

Note that while Eq.~\eqref{eq:dimformula} is a statement about exact generation of unitaries, stronger statements can be made when $\unit$ is compact (which is always the case when each of the generators in $\gen$ has a rational spectrum) -- in such cases, it is not even possible to approximate the missing unitaries.
However, if $\dlie$ is not compact, it may be possible to approximate with arbitrary precision some unitaries outside of $\unit$, effectively reducing number of missing dimensions.

\subsubsection{$SU(2)$ Symmetric Systems}
To further illustrate our approach to weak non-universality, let us consider the case of $k$-local $SU(2)$-symmetric circuits on a qubit chain. This has been proven to be semi-universal in Ref.~\cite{marvian2022rotationally} using different methods, and the results of Ref.~\cite{li2024designslocalrandomquantum} also imply weak non-universality\footnote{These are stated to hold only for $L>8$, but for smaller system sizes weak non-universality can be checked numerically using the methods in App.~\ref{sec:app-num}.}.
For example, the $k=2$ case with periodic boundary conditions is generated by
\begin{equation}
    \gen=\gen_{SU(2)}\defeq\{\vec S_j\cdot\vec S_{j+1}\}_{j=1,...,L}.
\label{eq:SU2gens}
\end{equation}
To count the number of constraints on the relative phases, we can compute the dimension of the overlap between the center and the generators $\gen_{SU(2)}$ as in Eq.~\eqref{eq:dimformula}.
A simple approach is to consider a basis for center $\cent = \centx{SU(2)}$ such that each element only overlaps with generators of a given range (similar to the Pauli string basis of Eq.~\eqref{eq:paulibasisu1}).
Since $\gen_{SU(2)}$ generates the algebra of permutations of the $L$ sites on the lattice~\cite{moudgalya2022from}, $\centx{SU(2)}$ is simply the set of permutation-invariant operators which belong to the the permutation algebra.
By symmetrizing permutation operators that have support on $n$ qubits we 
obtain the following orthogonal basis for the center
\begin{equation}
    \cent=\mrm{span}(\{\oket{P_n}\}_{n=0,2,...,2\lfloor\frac{L}{2}\rfloor})
\end{equation}
where $\oket{P_n}$ is the operator
\begin{equation}
    P_n=\sum_{j_1\neq...\neq j_n}\left(\vec S_{j_1}\cdot \vec S_{j_2}\right)\cdot...\cdot\left(\vec S_{j_{n-1}}\cdot \vec S_{j_n}\right).
\end{equation}
When expanded in terms of Pauli strings, the operator $\oket{P_n}$ only consists of strings of length $n$, and therefore has zero overlap with any $k$-local operator with $k<n$.
We can therefore conclude through Eq.~\eqref{eq:dimformula} that for the set of \textit{all} $SU(2)$-symmetric gates that are at most $k$-local (which includes $\oket{P_n}$ for all $n\leq k$):
\begin{equation}
    \mrm{dim}({\bond})-\mrm{dim}({\dlie})=\left\lfloor\frac{L}{2}\right\rfloor-\left\lfloor\frac{k}{2}\right\rfloor.
\end{equation}
When instead considering a subset of all possible $SU(2)$-symmetric $k$-local gates, this result is still valid as a lower bound on the co-dimension of $\dlie$ in $\bond$.

\subsubsection{Systems with a Non-Group Symmetry}\label{subsubsec:nongroup}
We have so far demonstrated that this framework can be used to systematically understand non-universality in $k$-local group-symmetric circuits, which has also been demonstrated in previous works~\cite{marvian2020locality, marvian2024abelian} using alternate methods.
However, this framework is much more general, and it provides a unified description for studying non-universality for \textit{any} set of hermitian generators $\gen=\{h_\alpha\}$; these could be symmetric $k$-local gates, but in general they need not possess any particular internal symmetry or spatial structure.
This allows us to quantify the non-universality corresponding to \textit{any} weakly non-universal gate set $\gen$ by computing its center $\cent$, and simply applying the formula Eq.~\eqref{eq:dimformula}.
Note that particular subsets of symmetric gates can lead to more dramatic examples of non-universality, which we study in Sec.~\ref{sec:con-super}.
For example we can apply this framework to gates $\gen$ that have unconventional symmetries~\cite{moudgalya2021hilbert, moudgalya2022exhaustive}, where the commutant $\comm$ does not have a simple group structure.
For example the $t$-$J_z$ model~\cite{batista2000tJz, rakovszky2020statistical, moudgalya2021hilbert} is such an example, which exhibits Hilbert space fragmentation~\cite{moudgalya2021review}. 
It is a model with $\mb C^3$ local degrees of freedom denoted by the basis $\{\ket{\up}, \ket{0}, \ket{\downarrow}\}$, and the Hamiltonian consists of the terms
\begin{multline}
    \gen = \gen_{t\text{-}J_z} \defeq \{T_{j,j+1},Z_j Z_{j+1}\}_{j=1,...,L-1}\\\cup\{Z_j,Z_j^2\}_{j=1,...,L}\label{eq:tjzgates}
\end{multline}
where $Z_j \defn \ketbra{\uparrow}{\uparrow}_j -\ketbra{\downarrow}{\downarrow}_j$ and $T_{j,j+1} \defn (\ketbra{\uparrow 0}{0\uparrow}+\ketbra{\downarrow 0}{0\downarrow})_{j,j+1} + h.c.$.
These gates act on single or neighbouring degrees of freedom with open boundary conditions, and their action conserves the full \textit{pattern} of spins ($\ket{\uparrow}$ or $\ket{\downarrow}$) in one dimension~\cite{rakovszky2020statistical}.
This pattern conservation symmetry is seen in the commutant of this set of gates, which is Abelian and has  dimension that scales exponentially with the system size $\mrm{dim}(\comm)=\mrm{dim}(\cent)=2^{L+1}-1$~\cite{moudgalya2021hilbert}.
This directly implies that according to Eq.~\eqref{eq:dimformulark}, the co-dimension of $\dlie$ in $\bond$ will also need to grow exponentially with system size, since the number of generating gates only grows linearly.
According to numerics for small system sizes using the methods discussed in App.~\ref{sec:app-num}, we observe weak non-universality, and therefore we expect Eq.~\eqref{eq:dimformula} to apply exactly.
The exact calculation of the overlaps (performed in App.~\ref{sec:gates}) shows that the projection of $\gen_{t\text{-}J_z}$ onto the center $\cent=\comm$ has dimension $2L$, thus implying
\begin{equation}
    \mrm{dim}(\bond)-\mrm{dim}(\dlie)\geq 2^L-2L-1.
\end{equation}
This shows that in systems with Hilbert space fragmentation~\cite{sala2020fragmentation, khemani2019int,  moudgalya2019thermalization, yang2019hilbertspace, moudgalya2021review}, which have exponentially many symmetries~\cite{moudgalya2021hilbert}, there is not only a heavy constraint on the dynamics of states due to these symmetries, but in addition the dynamics of unitary operations is further heavily constrained due to locality.

\subsubsection{Translation Invariant Gates}
Non-universality arising from Eq.~\eqref{eq:dimformula} can also be found in systems that respect some spatial symmetries.
For example, any one-dimensional system with periodic boundary conditions (PBC) acted upon by a spatially homogeneous (possibly time-dependent) Hamiltonian can be described by a gate set $\gen$ which possesses translation symmetry.
Hence we have
\begin{equation}
    \llangle T\rrangle = \text{span}\{\mathds{1}, T, T^2, \cdots, T^{L-1}\} \subseteq\comm
\end{equation}
where $T$ is the translation operator that maps each spin $j$ to the next one $j+1$ (with PBC).
For simplicity, let us consider the case where $\comm=\llangle T\rrangle$, which can be guaranteed by choosing the following set of operators acting on a qubit chain
\begin{equation}
    \gen=\gen_T\defeq\left\{\sum_{j=1}^L S_j^\alpha,\ \sum_{j=1}^L S^\alpha_jS^\beta_{j+1}\right\}_{\alpha,\beta\in\{x,y,z\}}.
\label{eq:gates-transinv}\end{equation}
Since the interactions in the gate set have finite range, the number of generators will be $\mc O(1)$ and independent of system size $L$, while $\mrm{dim}(\comm)=\mrm{dim}(\cent)=L$.
According to Eq.~\eqref{eq:dimformulark}, this immediately implies non-universality, since $\mrm{rk}(S)\leq |\gen|$, where $|\gen|$ is the number of linearly independent generators in $\gen$, and therefore
\begin{equation}
    \mrm{dim}(\bond)-\mrm{dim}(\dlie)\geq L-|\gen|\sim \mc O(L).
\end{equation}
This illustrates the strength of the rank argument, as it can be applied with little modifications to \textit{any} gate set with translation symmetry\footnote{The counting would change if the spatial part of the symmetry is enlarged to a non-abelian group with a smaller center.} to obtain a lower bound on the co-dimension of the DLA. 
In the particular case of the gate set $\gen_T$, the overlap matrix $S$ of Eq.~\eqref{eq:dimformulark} can be seen to have rank $1$ (cf. App.~\ref{sec:gates}).
Assuming that -- as simulations suggest -- this system is weakly non-universal, then $\mrm{dim}(\bond)-\mrm{dim}(\dlie)=L-1$.
\subsection{Strong Non-Universality}\label{sec:con-super}
The other type of non-universality occurs due to the presence of non-trivial superoperator symmetries $\mc Q\in\scomm, \mc Q\notin\scommt$.
The existence of such symmetries indicates either that some Krylov subspaces associated to the symmetries in $\scommt$ must split into smaller Krylov subspaces which are invariant under the action of $\sbond$, or that some Krylov subspaces belonging to different symmetry sector must be degenerate under the action of $\sbond$ (see Fig. \ref{fig:blocks}).
The presence of such non-trivial conserved quantities $\mc Q$ in the super-commutant, implies that $\dlie\neq\bond$.
If these conserved superoperators act non-trivially on $\bond$ {(i.e., if the block decomposition of $\bond$ is modified w.r.t. weakly non-universal systems)}, then some non-central operators will necessarily be missing from $\dlie$, as we show in Lemma \ref{lem4} of App.~\ref{sec:app-math}.\footnote{As a consequence Eq.~\eqref{eq:dimformula} still applies as a lower bound on $\mrm{dim}(\bond)-\mrm{dim}(\dlie)$ in the case of strong non-universality.}
In other words, non-trivial superoperator symmetries always modify the decomposition of the operator Hilbert space $\hend$ and always imply non-universality, and if they modify the decomposition within the set of symmetric operators $\bond$, they also directly imply missing overlaps between the generators and some of the non-scar subspaces in the decomposition.
Such a superoperator is defined by the equations:
\begin{equation}\label{eq:comm-conditions}
    [\mc Q,\sbond]=0,\qquad [\mc Q,\sbondt]\neq 0,
\end{equation}
which means that $\mc Q$ commutes with the adjoint action of the gates in $\mc G$, but not with the adjoint action of general gates that have the same symmetries as $\mc G$.
In the setting of Ref.~\cite{marvian2020locality}, for a symmetry group $G$ this would mean that we are looking for superoperators that only commute with the adjoint action of \textit{local} $G$-symmetric gates but not with that of \textit{global} ones.
In the following, we provide a few simple examples of gate sets that display strong non-universality.
\subsubsection{Decoupled Qubits}\label{sec:zx}
To illustrate the framework, we consider an almost trivial example of strong non-universality obtained by on-site (1-local) gate sets
\begin{equation}
    \gen=\gen_\mrm{XZ} \defeq \{X_j,Z_j\}_{j=1,...,L},
\end{equation}
which acts universally on each qubit of the chain, but cannot couple any pair of qubits.
The symmetry (commutant algebra) of this set of gates is trivial $\comm=\llangle\{\1\}\rrangle$, but it is clear that the set of realizable unitaries $\unit$ only contains those that factorize along the local degrees of freedom $U=\prod_{j=1}^L U_j$.
Interestingly, this non-universality can alternately be understood at the superoperator level as originating from additional strictly-local superoperator conserved quantities beyond the ones from Eq.~\eqref{eq:forsure}.
The super-commutant reads
\begin{equation}
    \scomm = \llangle \{...(\1\ot\1)_{j-1}\ot\oketbra{\1_j}{\1_j}\ot (\1\ot\1)_{j+1}...\}_{j=1,...,L}\rrangle.\label{eq:decouplingscomm}
\end{equation}
These conserved superoperators, corresponding to the on-site one-dimensional invariant subspaces $\oket{\1_j}$, simply emerge because the generators act independently on each qubit.
In the block decomposition of Eq.~\eqref{eq:hilbdecend}, the full operator Hilbert space then splits into $2^L$ non-degenerate Krylov subspaces, each one containing operators that only couple a given subset of qubits; the generators only overlap with the subspaces corresponding to on-site operators, which only act on one qubit, resulting in non-universality.

\subsubsection{Decoupling due to Frozen Configurations}
Superoperoperator symmetries also appear in more general situations where dynamical decoupling occurs.
If one considers for example the generators for the PXP model~\cite{turner2017quantum}:
\begin{equation}
    \gen =\gen_\mrm{PXP}\defeq \{P_jX_{j+1}P_{j+2}\}_{j=1,...,L-2},
\end{equation}
where $P=\ketbra{\downarrow}{\downarrow}$, one can see that a configuration of two neighboring spins such as $\uparrow\uparrow$ always remains frozen under the action of these terms.
For example, in the following configuration of $L = 8$ spins
\begin{equation}
    \ket{\psi(t=0)}=\ket{\downarrow\downarrow\downarrow\uparrow\uparrow\downarrow\downarrow\downarrow},
\end{equation}
the 2nd and 7th spins along the chain will evolve independently from each other, the other spins remain frozen under the action of the PXP terms.
This fact is described for example by the non-trivial superoperator conserved quantity:
\begin{equation}\label{eq:gluedPXPsymm}
    \begin{gathered}
        \mc Q= I_1I_2I_3U_4U_5S_6S_7S_8,\ \text{where:}\\
        I=\1\ot\1,\ S=\oketbra{\1}{\1},\ \text{and}\ U=\ketbra{\uparrow}{\uparrow}\ot \ketbra{\uparrow}{\uparrow}.
    \end{gathered}
\end{equation}
While the superoperators in Eq.~\eqref{eq:decouplingscomm} describe a system where all qubits are completely decoupled to each other, this describes a system where the left and right subsystems are decoupled as long as they are separated by a ``frozen'' configuration of the form $\uparrow\uparrow$ configuration.
The presence of this superoperator in the super-commutant $\scomm$ leads for example to the absence of the operators such as $P_1X_2P_3P_6X_7P_8\in\bond$ from the DLA $\dlie$.
Non-trivial conserved superoperator are a general feature of dynamical decoupling, and they have been identified also in other models where decoupling occurs \cite{kovács2024operatorspacefragmentationperturbed}. 
If a given configuration $\ket{f}$ is frozen under the dynamics of a given gate set $\gen$, and if the Hamiltonian terms $h_\alpha\in\gen$ are unable to couple degrees of freedom at the left and right of a frozen configuration (for example due to locality of the interactions), then superoperators of form similar to Eq.~\eqref{eq:gluedPXPsymm} will appear in the super-commutant $\scomm$, with the string $U_j...U_{j+k}$ being replaced by $\ketbra{f}{f}\ot\ketbra{f}{f}$.
\subsubsection{\texorpdfstring{$\mathbb Z_2$}{} Matchgate Circuits}\label{sec:mgz2}
An important example of a superoperator conserved quantity can be found in matchgate circuits \cite{valiant2002quantum, terhal2002classical, jozsa2008matchgate}, where the generators are of the form
\begin{equation}\label{eq:matchgate-gens}
    \gen=\gen_\mrm{MG} \defeq \{X_jX_{j+1}\}_{j=1,...,L-1}\cup \{Z_j\}_{j=1,...,L},
\end{equation}
where we have assumed open boundary conditions (OBC) for simplicity.
These terms have a quadratic free-fermion form after a Jordan-Wigner transformation.
That is, we define the Majorana fermions
\begin{equation}
\gamma_{2j-1}=Z_1...Z_{j-1} X_{j},\quad \gamma_{2j}=Z_1...Z_{j-1} Y_{j},
\label{eq:majoranadefn}
\end{equation}
with anticommutation relations $\{\gamma_k,\gamma_l\}=2\delta_{kl}$.
We define the Majorana strings $\oket{a}=\oket{\prod_{k=1}^{2L}\gamma_k^{a_k}}$ where $a=(a_1,a_2,...,a_{2L})$ with $a_k\in\{0,1\}$ and $|a|\defeq\sum_{k=1}^{2L}a_k$ is referred to as the length of the string (i.e. the \textit{degree} of the operator).
The set of all such products of operators is an orthogonal basis of the operator space.
In terms of these Majoranas, the generators of Eq.~\eqref{eq:matchgate-gens} take the form of quadratic fermion operators
\begin{gather}
    Z_j=-i\gamma_{2j-1}\gamma_{2j}\quad X_j X_{j+1}=-i\gamma_{2j}\gamma_{2j+1} \nn \\
    \implies \gen_{\text{MG}} = \{-i \gamma_j \gamma_{j+1}\}_{j = 1, \cdots, 2L-1} 
    \label{eq:majoranagen}
\end{gather}
These generators commute with the $\mb Z_2$ parity operator $P=(-i)^L\prod_{k=1}^{2L} \gamma_k$, which can easily be verified to generate the commutant $\comm$.
The associated super-commutant $\scomm$ contains operators of the form of Eq.~\eqref{eq:forsure}, which are inherited from the symmetry.
However, beyond these superoperators, the adjoint action of the generators preserves an unrelated $U(1)$ superoperator symmetry: the number of Majorana fermions, defined as
\begin{equation}
    \mc N_\gamma = \sum_{n=0}^{2L}n\sum_{|a|=n}\oketbra{a}{a}\label{eq:majoranaN}
\end{equation}
Note that this or related conserved quantity has been pointed out in earlier works~\cite{bravyi2004lagrangian, prosen2008third, spee2018mode, bao2021symmetry}.
This additional symmetry splits the parity subspaces according to the number of Majorana fermions, making sure that $\dlie$ is equal to the space of quadratic operators (i.e. $\mc N_\gamma=2$).

The $U(1)$ generator $\mc N_\gamma$ does not commute with the $\mb Z_2\ot\mb Z_2$ superoperator symmetry that is derived from the physical $\mb Z_2$ parity, thus generating a non-Abelian super-commutant $\scomm$ with dimension\footnote{The absence of additional superoperator symmetries can be proven for example {by mapping them to the ground state of an effective Hamiltonian on four copies of the Hilbert space [see discussion around Eq.~(\ref{eq:fourcopyavg})] that turns out to be related to} the spin-$\frac{1}{2}$ Heisenberg model \cite{PhysRevB.112.064301}{, whose ground states are well-known}.}
\begin{equation}
    \mrm{dim}(\scomm)=4L+2.
\end{equation}
The operator Hilbert space $\hend$ splits into irreps (cf. Eq.~\eqref{eq:hilbdecend})
\begin{equation}
    \hend = \left(\bigoplus_{n=0}^{L-1}\hend_{n,2L-n}\right)\oplus \hend_{L^+}\oplus\hend_{L^-}.
\label{eq:MGdecomp}
\end{equation}
Here we have
\begin{align}
    \hend_{n,2L-n} &= \mrm{span}(\{\oket{a}:\,|a|=n\ \mrm{or}\ |a|=2L-n\}),\nn \\
    \hend_{L^\pm} &= \mrm{span}(\{(\Pi_\pm\ot\1)\oket{a}:\,|a|=L\}),
\end{align}
and we define $\Pi_\pm = \frac{\1\pm P}{2}$ to be the projectors onto the parity sectors.
The $\hend_{n,2L-n}$ irreps correspond to two degenerate Krylov subspaces (i.e. $d_{\widehat\lambda}=2$ in the language of Eq.~\eqref{eq:matrixrep}) consisting of strings of length $n$ and $2L-n$, while $\hend_{L^\pm}$ are simple Krylov subspaces with $d_{\widehat\lambda}=1$ consisting of strings of length $L$ (since the action of $P$ preserves the length of the string); the space of strings of length $L$ is block diagonalized into parity sectors, which are $\hend_{L^+}$ and $\hend_{L^-}$.
The bond algebra $\bond$ of symmetric operators is given by all $\hend_{n,2L-n}$ for even $n$, and $\hend_{L^\pm}$ if $L$ is even.
However, since the generators in Eq.~(\ref{eq:majoranagen}) have length $2$, they have complete weight within one of the two degenerate subspaces of $\hend_{2,2L-2}$, the one composed of all operators of degree $2$, which is also the DLA $\dlie$.
This leads to dramatic non-universality within the space of symmetric operators.\footnote{\label{ft:matchgate2}However, note that for $L=2,$ matchgates only differs from the general $\mathbb Z_2$ case (seen in App. \ref{sec:gates}) through the absence of the central element $Z_1Z_2=P$ from the Lie algebra. Thus the $L=2$ system is semi-universal, while also being strongly non-universal (due to the double degeneracy found in the subspace $\hend_{1,3}$).}

A related example is obtained by considering periodic boundary conditions (PBC) for the generators $\gen_\mrm{MG}$ by adding the term $X_LX_1=-iP\gamma_1\gamma_{2L}$ to Eq.~(\ref{eq:matchgate-gens}), which in terms of Majoranas of Eq.~(\ref{eq:majoranadefn}) is a string of length $2L-2$.
Adding this single term modifies the super-commutant and block-decomposition of the operator Hilbert space since the Majorana number superoperator of Eq.~(\ref{eq:majoranaN}) is no longer conserved under the action of the corresponding Liouvillian $\mL_{X_L X_1}$.
If $|a|$ is even, $\ad{X_LX_1}\oket a=iP\ad{\gamma_1\gamma_{2L}}\oket a$, hence mapping Majorana strings of length $n$ to those of length $2L-n$.
In terms of Eq.~(\ref{eq:MGdecomp}), for even $n$, $\hend_{n,2L-n}$ is fully connected under the action of this term, and becomes a single Krylov subspace with $d_\lambda=1$. 
The DLA $\dlie$ will correspond to the space $\hend_{2,2L-2}$, spanned by Majorana strings of length $2$ and $2L-2$.
The gate set therefore remains strongly non-universal, with a dynamical Lie algebra of twice the size compared to the OBC case.
\subsubsection{\texorpdfstring{$U(1)$}{} Matchgate Circuits}\label{sec:mgu1}
The case of $U(1)$-conserving matchgate circuits, associated to particle number conserving free-fermion Hamiltonians, is a simple extension of the previous case.
A set of generators for this type of circuits is
\begin{multline}\label{eq:matchgate-gens-u1}
    \gen=\gen_{\mrm{MG},U(1)} \defeq \{X_jX_{j+1}+Y_jY_{j+1}\}_{j=1,...,L-1}\\ \cup \{Z_j\}_{j=1,...,L}.
\end{multline}
These can be expressed in terms of physical fermionic creation and annihilation operators defined as:
\begin{equation}
\begin{gathered}
    c_j = \frac{\gamma_{2j-1}+i\gamma_{2j}}{2},\qquad c_j\+ = \frac{\gamma_{2j-1}-i\gamma_{2j}}{2},\\
    Z_j = 2c_{j}\+c_{j}-\1, \quad X_jX_{j+1}+Y_jY_{j+1}=c\+_{j}c_{j+1}+h.c.
\end{gathered}
\end{equation}
These generators commute with the $U(1)$ number operator $N_\mrm{tot}=\sum_j c\+_jc_j$, which generates the commutant $\comm$~\cite{moudgalya2022from}.
Due to the presence of this symmetry, as well as the superoperator symmetry introduced in the previous example in Eq.~\eqref{eq:majoranaN}, the number of $c$ and $c\+$ operators in a string are preserved independently under the adjoint action of the generators (see App.~\ref{sec:gates} for the full block decomposition of the operator space).
As expected, this leads to $\dlie$ being equal to the space of number-conserving quadratic operators spanned by $\oket{c\+_jc_{j'}}$ (except for the number operator $N_\mrm{tot}$) giving rise to non-universality.\footnote{As for the $\mathbb Z_2$ matchgate case, this gate set is semi-universal and strongly non-universal for $L=2$.}

\subsubsection{Matchgate-like Subsystems}
The examples discussed so far provide an understanding of the non-universality of some classic types of circuits in terms of superoperator conserved quantities.
Nevertheless, it is important to note that our discussion of matchgate systems also naturally applies -- with little adjustments -- to many other examples of strong non-universality, other than the ones associated with the sets of gates in Eqs.~\eqref{eq:matchgate-gens} and \eqref{eq:matchgate-gens-u1}.
There are many examples of systems whose dynamics are not free, but that within some symmetry sectors behave exactly like free fermions.
In all such cases, the part of the super-commutant $\scomm$ responsible for the block decomposition of operators acting on these subspaces will take one of the forms shown above, with the presence of an additional $U(1)$ generator $\mc N_\gamma$ at the superoperator level.
Many systems with Hilbert space fragmentation can often possess integrable subspaces which evolve under free-fermion dynamics~\cite{Batista_2000, rakovszky2020statistical, moudgalya2019thermalization}.
For example, the $t$-$J_z$ model of \eqref{eq:tjzgates} discussed in Sec.~\ref{subsubsec:nongroup} above can be restricted in such a way that any given subspace maps onto spinless free fermions.
For example, consider
\begin{multline}
    \gen = \gen_{t\text{-}J_z,\text{MG}} \defeq \{T_{j,j+1}, Z_{j}^2Z_{j+1}^2-Z_{j}Z_{j+1}\}_{j=1,...,L-1}\\\cup\{Z_j^2\}_{j=1,...,L}
\end{multline}
and the subspace $\mc H_{\sigma}$, with $\sigma\in\{\uparrow,\downarrow\}$, spanned by all product states where all spins are either $0$ or $\sigma$, or equivalently, obtained by repeatedly acting with the terms $T_{j,j+1}$ on states of the form 
\begin{equation}
|\underset{N}{\underbrace{\sigma\sigma
\cdots \sigma}} \ \underset{L-N}{\underbrace{00
\cdots 0}} \ \rangle.
\end{equation}
If we map each product state in $\mc H_{\sigma}$ to a state in a spin-$\frac{1}{2}$ chain through the mapping $\ket{0}\mapsto|\tilde\downarrow\rangle$ and $\ket{\downarrow/\uparrow}\mapsto|\tilde\uparrow\rangle$ (due to the fixed spin pattern, this mapping is one-to-one between the two spaces), then we find that on this subspace
\begin{equation}
\begin{gathered}
    T_{j,j+1}=\tilde X_j\tilde X_{j+1}+\tilde Y_j\tilde Y_{j+1},\quad Z_j^2 = \frac{1}{2}\left(\tilde Z_j+\1\right),\\
    Z_{j}^2Z_{j+1}^2-Z_{j}Z_{j+1} = 0,
\end{gathered}
\end{equation}
where $\tilde X_j$, $\tilde Y_j$, $\tilde Z_j$ are Pauli matrices on the spin-$\frac{1}{2}$ degrees of freedom.
Therefore this gate set acts on each $\mc H_{\sigma}$ symmetry sector exactly like the $U(1)$-conserving matchgates of Eq.~\eqref{eq:matchgate-gens-u1}, with $N$ being the total particle number, thus inheriting the conserved superoperators that characterize its strong non-universality.
Numerous mappings of the same flavor have also been found in subspaces of certain dipole-conserving models~\cite{moudgalya2019thermalization, rakovszky2020statistical}, and we expect the same type of strong non-universality in such systems.
In many cases, slight modifications to the analysis performed above may need to be applied when studying free-fermion sectors of such circuits.
For example, to perform the appropriate mapping, the extent of the physical Hilbert space $\mc H$ of the matchgate system may need to be restricted to just a few symmetry sectors: the free-fermion subspace of the system of interest could correspond to only one of the parity sectors of the $\mb Z_2$ matchgates, or to only some of the particle-number sectors of the $U(1)$ matchgates.
Additionally, the free-fermion sector of the system of interest may be composed of many degenerate copies of the same subspace ($d_{\widehat\lambda}>1$ in Eq.~\eqref{eq:fund-th}), in which case the physical Hilbert space of matchgates would need to be enlarged by tensoring it with additional degrees of freedom for the mapping to be precise.

Both of these phenomena appear in the example of $SU(d)$ symmetric circuits for $d > 2$, where non-trivial superoperator symmetries have also been found when considering 2-local gates~\cite{hulse2021qudit}.\footnote{There this has been shown in a different language, but it is related to ours through the correspondences shown in Sec.~\ref{sec:manycopy}.}, although a systematic understanding of non-universality in this case is still lacking.
The main signature, which persists for arbitrarily large system sizes, is that some of the symmetry sectors in the $SU(d)$ decomposition of the Hilbert space can be mapped to number-conserving free fermions; this part of the analysis therefore reduces to the matchgate example from the previous section.
This sector maps onto free fermions with a restriction on the total particle number $N<d$, and the additional superoperators associated to free-fermion dynamics only appearing in $SU(d)$ circuits for $d>2$, where states with more than one particle can exist.\footnote{For $d=2$, a mapping can be made between some subspaces and a free-fermion system with total particle number restricted to $0\leq N\leq 1$, but such a system is only weakly non-universal.}
Furthermore, the $SU(d)$ irreps associated to free fermions are degenerate in the Hilbert space decomposition of Eq.~\eqref{eq:fund-th}, with each $N$-particle sector of the free-fermion subspace appearing with multiplicity $d_{\lambda_N}=\frac{d!}{N!(d-N
-1)!}$.

\section{Physical Implications}\label{sec:physicalimplications}
We now ask if the physical consequences of the non-universality.
One naive way to physically detect non-universality in a set of gates is to find a quantity that is conserved under $\unit$ unitary evolution {under the restricted set of symmetric gates}, but is not conserved under general $\unitt$ unitary evolution {under arbitrary set of symmetric gates}.
However, it is impossible to find such a quantity since $\mrm{comm}(\gen)=\mrm{comm}(\lgen \mathcal{G} \rgen)=\comm$.
Nevertheless, as we show below, conserved quantities of this type can be found on multiple copies of the system.
In this section we discuss some observables whose late-time physics is controlled by objects known in past literature as ``quadratic symmetries'' \cite{zimboras2015symmetry}, that have a direct connection to strong non-universality and the super-commutant.
We also discuss some multi-copy generalizations which can instead hold information on weak non-universality.
\subsection{Many-copy Interpretation and Quadratic Symmetries}\label{sec:manycopy}
The study of non-universality in this work is based on the interpretation through Eq. \eqref{eq:liouvillian} of operators as states of two copies of the original system.
Operator evolution translates to
\begin{equation}\label{eq:twocopy-evo}
    O(t)=UOU\+\implies \oket{O(t)}=U\ot U^*\oket{O},
\end{equation}
which can be interpreted as a Schr\"{o}dinger time-evolution on a two-copy Hilbert space, with the first copy evolving forwards and the second copy evolving backwards.
For $U=e^{-iHt}$, we have $\oket{O(t)}=e^{-i \ad{H} t}\oket{O}$, where $\ad{H} = H \otimes \mathds{1} - \mathds{1} \otimes H^T$ (cf. Eq.~\eqref{eq:liouvillnot}), and operators in the commutant $ Q\in \comm$ are invariant under the Heisenberg evolution Eq.~\eqref{eq:twocopy-evo} for any unitary $U$ generated by the generators $\gen$.
This is of course not the only way of interpreting a two-copy Hilbert space: the forward time evolution on both the copies of a two-copy state $\ket\psi_2\in \mc H\ot\mc H$ is
\begin{equation}
    \ket{\psi(t)}_2= U\ot U \ket\psi_2, \label{eq:twocopy}
\end{equation}
where for $U = e^{-iH t}$, the associated two-copy Hamiltonian has the form $\bad{H}=H\ot\1+\1\ot H$.
While at this level  Eqs.~(\ref{eq:twocopy-evo}) and (\ref{eq:twocopy}) do not seem to have much in common, for a given set of local Hamiltonian terms $\gen$, the commutants of superoperators of type $\{\ad{h_\alpha}\}$ (which we are referring to as \textit{superoperator symmetries}) and those of type $\{\bad{h_\alpha}\}$ (which were referred to as \textit{quadratic symmetries} in earlier literature~\cite{zimboras2015symmetry}) are actually in a one-to-one correspondence through partial transposition.
To show this, it is useful to generalize the Liouvillian notation of Eq.~\eqref{eq:liouvillian} by introducing a graphical notation for states in many-copy Hilbert spaces $\mc H^{\ot n}$ that will also be useful for the next sections.
We start by representing each copy of the Hilbert space as a dot, so that
\begin{equation}
    \tket{\vdotc{0}{$\mu_1$}\vdotc{0.4}{$\mu_2$}\vdotc{0.8}{$\cdots$}\vdotc{1.2}{$\mu_n$}}_{\!n}\defeq \ket{\mu_1}\ot\ket{\mu_2}\ot\cdots\ot\ket{\mu_n}
\end{equation}
where $\{\ket{\mu}\}$ is an orthonormal basis for $\mc H$.
For our purposes, we are restricted to cases where the total number of copies is even, with half of them corresponding to forward time-evolution, while the other half correspond to backward time-evolution.
We then introduce the following notation to represent operators
\begin{equation}
    \tket{\varcdimercr{0.4}{0}{$O$}{white}}_{\!2}\defeq \sum_{\mu_1,\mu_2} O_{\mu_1\mu_2}\tket{\vdotc{0}{$\mu_1$}\vdotc{0.4}{$\mu_2$}}_{\!2},
\end{equation}
where the arrow indicates the fact that the operator $O=\sum_{\mu_1,\mu_2} O_{\mu_1\mu_2}\ketbra{\mu_1}{\mu_2}$ is a map from $\tp{\vdotc{0}{$2$}}$ to $\tp{\vdotc{0}{$1$}}$.
Note that this is identical to the definition of $\oket O$ in Eq.~\eqref{eq:liouvillian}. For the sake of brevity we will sometimes simply indicate the number of copies as a lower index to the ket, as in Eq.~\eqref{eq:twocopy}, so that $\oket O$ becomes $\ket{O}_2$.
More generally, we will use a similar notation to denote states on higher number of copies of the Hilbert space, such as superoperators.
For example we will indicate a generic superoperator symmetry as $\ket{\mc Q}_4\in\scomm$, while the two types of superoperator symmetries appearing in Eq.~\eqref{eq:forsure} in this notation have the form:\footnote{Here we replaced $Q_1$ with $Q_1\+$ to consider the general case in which $Q_1\in\comm$ is not necessarily hermitian.}
\begin{equation}
\begin{aligned}
    Q_1\ot Q_2^T &\mapsto \tket{\varcdimercr{0.4}{0}{$Q_1$}{white}\varcdimerc{0.8}{1.2}{$Q_2$}{white}}_{\!4}\\
    \oketbra{Q_2}{Q_1\+} &\mapsto \tket{\varcdimercr{0.8}{0}{$Q_2$}{white}\varcdimerc{0.4}{1.2}{$Q_1$}{white}}_{\!4}.\label{eq:dimer-trivial-comm}
\end{aligned}
\end{equation}

We can now straightforwardly show the one-to-one map between the commutant of $\{\ad{h_\alpha}\}$ and that of $\{\bad{h_\alpha}\}$.
From Eq.~\eqref{eq:twocopy-evo} we obtain that the operators $\ket Q_2$ that satisfy $U\ot U^* \ket{Q}_2 = \ket{Q}_2$ are exactly the ones such that $[H,Q]=0$; hence replacing $U$ with superoperators $e^{-i\ad{h_\alpha}t}$ and $e^{-i\bad{h_\alpha}t}$ we get the equivalences
\begin{equation}\label{eq:unitary-patterns}
    \begin{gathered}
        [\ad{h_\alpha},\mc Q]=0\iff (U\ot U^*\ot U^*\ot U)|\mc Q\rangle_4=|\mc Q\rangle_4,\\
        [\bad{h_\alpha},\bar{\mc Q}]=0\iff (U\ot U\ot U^*\ot U^*)|{\bar{\mc Q}}\rangle_4=|{\bar{\mc Q}}\rangle_4.
    \end{gathered}
\end{equation}
This shows that superoperators in the super-commutant $\mc Q\in\scomm$ are in one to one correspondence with operators $\bar{\mc Q}$ in the commutant of a two-copy system through partial transposition of the second and fourth copies of the Hilbert space in the tensor product $\tp{\vdotc{0}{$2$}}\leftrightarrow\tp{\vdotc{0}{$4$}}$.
This allows us to establish a direct connection between the strong forms of non-universality, and physical phenomena, which we will explore in the following sections.
Note however that while the commutants of the superoperators $\{\ad{h_\alpha}\}$ and $\{\bad{h_\alpha}\}$ have the same dimension, they do not need to have the same algebraic structure, since the mapping between them does not preserve matrix multiplication.
Therefore all the results discussed in earlier sections, which connect the block structure of superoperator conserved quantities to the Lie algebra $\dlie$ are only true when considering the actual super-commutant $\scomm$, derived from $\{\ad{h_\alpha}\}$, while such connections are not so simple for quadratic conserved quantities, derived from $\{\bad{h_\alpha}\}$.
This fact also entails a few interesting consequences on the general algebraic structure of the super-commutants $\scomm$ which are explored in App.~\ref{sec:app-manycopy}, and allows for improvements in the performance of the numerical methods studied in App.~\ref{sec:app-num}.
In summary, conserved superoperators responsible for strong non-universality can be thought of as conserved quantities appearing in two-copy systems, which then provide dynamical signatures when studying the late-time behaviour of appropriate physical quantities such as higher point correlation functions and Rényi entropies.
In the case of circuits that produce a compact subgroup $\unit$, this is equivalent to the statement that strong $\scomm\supsetneq\scommt$ non-universality is equivalent to the absence of $2$-designs~\cite{li2024designslocalrandomquantum}, which is a standard result in the literature.

\begin{figure*}[t]\vspace{0pt}
\stackon[5pt]{\includegraphics[width=.45\textwidth]{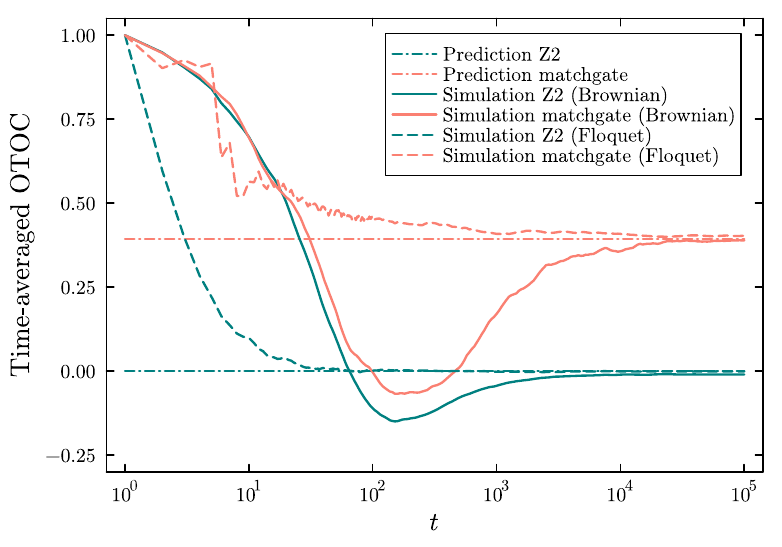}}{\textsf{(a)}}
\hspace{2em}
\stackon[5pt]{\includegraphics[width=.45\textwidth]{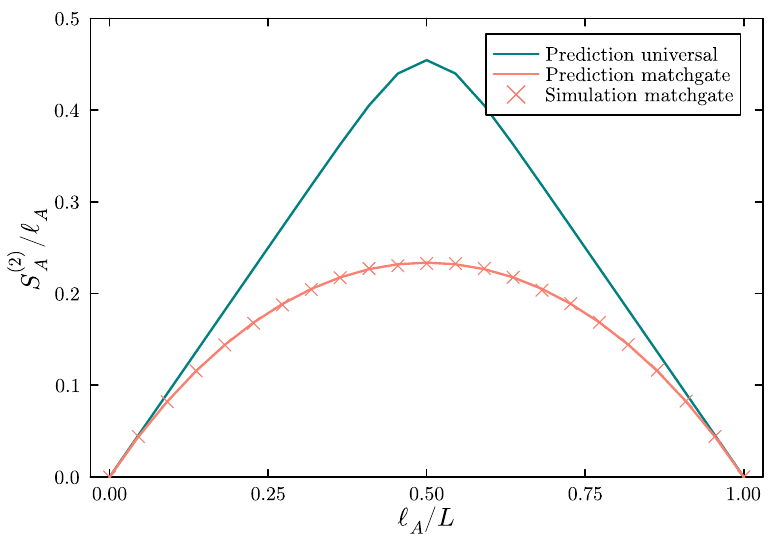}}{\textsf{(b)}}
\caption{Plots of simulations and predictions for the asymptotic values of physical observables associated to strong universality breaking. (a) Time-averaged OTOC of $Z_j$ for a $\mb Z_2$-preserving and matchgate Brownian circuit, and for a $\mb Z_2$-preserving and matchgate Floquet system ($L=6$). (b) Average value of the second Rényi entropies for a universal Brownian circuit and a matchgate Brownian circuit ($L=22$).}\label{fig:phys-mg}
\end{figure*}
\subsection{Higher Point Correlation Functions}
In this section, we will show the impact of superoperator conserved quantities on the higher point correlation functions. 
For the purposes of illustration, we focus on the Out-of-Time Ordered Correlators (OTOCs), which
are quantities that are often used to study operator spreading and chaos in quantum many-body systems~\cite{cotler2017chaos}.
For a pair of hermitian observables $A(t)$ and $B(t)$, the OTOC is defined as
\begin{equation}
\begin{aligned}
        C_{AB}(t)\defeq &\,\langle A(t)B(0)A(t)B(0)\rangle\\
    =&\,\langle UA U\+ B U A U\+ B \rangle
\end{aligned}
\end{equation}
where $\langle \cdot\rangle$ is a thermal average (we will consider the infinite temperature limit $\langle \cdot\rangle\rightarrow\frac{1}{\dim(\mc H)}\tr(\cdot)$).
When averaged over an \textit{ensemble} of unitaries, this quantity contains information about the second moment of the unitary distribution, which is evident from its expression as a four-copy time evolution operator with two forward and two backward time-evolutions:
\begin{equation}
    C_{AB}(t)= \prescript{}{4\!}{\tbra{\varcdimercr{0.4}{0}{$B$}{white}\varcdimerc{0.8}{1.2}{$B$}{white}}}
    \frac{U\ot U^*\ot U^*\ot U}{\dim(\mc H)}
    \tket{\varcdimercr{0.8}{0}{$A$}{white}\varcdimerc{0.4}{1.2}{$A$}{white}}_{\!4}.\label{eq:otoc-many-copy}
\end{equation}
Note that in order to make a clear connection to non-universality we chose to order the unitaries as in the superoperator interpretation from Eq.~\eqref{eq:unitary-patterns}.
The dynamics of the OTOC $C_{AB}(t)$ are in general very complex, but their average long-time behaviour is expected to be determined by the superoperator symmetries $\ket{\mc Q}_4\in\scomm$.
This is in analogy to the fact that the late-time behavior of two-point functions $C_A(t) \defn \langle{A(t)A(0)}\rangle$ is controlled by the physical symmetries $\ket Q_2\in\comm$, a result known as the Mazur bound~\cite{mazurbound1969, dhar2020revisiting, moudgalya2021hilbert,moudgalya2023symmetries}.
In particular, this can be written as a matrix element
\begin{equation}
C_A(t) =\prescript{}{2\!}{\tbra{\varcdimercr{0.4}{0}{$A$}{white}}}\frac{U\ot U^*}{\dim(\mc H)}\tket{\varcdimercr{0.4}{0}{$A$}{white}}_{\!2},
\end{equation}
and we expect that the late-time behavior is given by
\begin{equation}
    \overline{C_A}(\infty) \sim \sum_{Q}\prescript{}{2\!}{\tbra{\varcdimercr{0.4}{0}{$A$}{white}}}\frac{\ket{Q}_2\!\bra{Q}}{\dim(\mc H)}\tket{\varcdimercr{0.4}{0}{$A$}{white}}_{\!2}
\label{eq:reviewmazur}
\end{equation}
where $\overline{C_A}(\infty) = \lim_{t \rightarrow \infty}\frac{1}{T}\int_0^\infty C_A(t)\ \dd t$ is the time-average of $C_A(t)$, and $\{\oket Q\}$ is an orthonormal basis for the symmetries in the commutant $\comm$.
Note that the Mazur bound actually says that the R.H.S. of Eq.~(\ref{eq:reviewmazur}) is a lower-bound for the L.H.S., but in practice one finds in numerical simulations that this is quite close to saturation for general symmetric evolution.
In addition, Eq.~(\ref{eq:reviewmazur}) can be shown to be exact equality when $U$ is chosen from an ensemble of random Brownian circuits~\cite{moudgalya2023symmetries} generated by $\mG$ (reviewed in App.~\ref{sec:brown}), and $\overline{C_A(\infty)}$ is interpreted as averaging over the ensemble of Brownian circuits.
In particular, ensemble averaging over the class of Brownian unitary evolutions gives us~\cite{lashkari2013towards,bauer2017stochastic,  sunderhauf2019quantum, xu2019locality,  ogunnaike2023unifying, moudgalya2023symmetries, vardhan2024entanglement}
\begin{equation}
    \overline{U \otimes U^\ast} = e^{-P_2 t},
\label{eq:twocopyavg}
\end{equation}
where $P_2 \defn \sum_\alpha \mL_{h_\alpha}^2$ is a superoperator (that acts on two copies of the Hilbert space operator) whose ground states are the symmetries in the commutant $\comm$ ($\ad{h_\alpha}$ is the adjoint map defined in Eq.~\eqref{eq:adjoint}).
We expect that similar results generalize to OTOCs, and in particular that their late-time value depends on the superoperator symmetries of the system. 
Similar to Eq.~(\ref{eq:reviewmazur}), we will be interested in studying the long-time average of the OTOC $\overline{C_{AB}}(\infty)\defeq \lim_{T\rightarrow\infty}\frac{1}{T}\int_0^T C_{AB}(t)\,\dd t$,
and we expect that 
\begin{equation}
    \overline{C_{AB}}(\infty) \sim \sum_{\mc Q}\prescript{}{4\!}{\tbra{\varcdimercr{0.4}{0}{$B$}{white}\varcdimerc{0.8}{1.2}{$B$}{white}}}\frac{\ket{\mc Q}_4\!\bra{\mc Q}}{\mrm{dim}(\mc H)}
    \tket{\varcdimercr{0.8}{0}{$A$}{white}\varcdimerc{0.4}{1.2}{$A$}{white}}_{\!4}\label{eq:otoc-projector}
\end{equation}
where $\{\ket{\mc Q}\}$ forms an orthonormal basis of $\scomm$.
As shown in App.~\ref{sec:brown}, this is precise for the case of Brownian circuits with generators $\gen$: similar to the Eq.~(\ref{eq:twocopyavg}), the for the four-copy case we obtain
\begin{equation}
    \overline{U \otimes U^\ast \otimes U^\ast \otimes U} = e^{-P_4 t},
\label{eq:fourcopyavg}
\end{equation}
where $P_4 \defn \sum_\alpha{(\mL_{h_\alpha} \otimes \mathds{1} - \mathds{1} \otimes \mL_{h_\alpha}^T)^2}$ and the ground states of $P_4$ are the superoperator symmetries of $U$, which lie in $\scomm$.
This allows us to conjecture a direct dynamical consequence of the strong non-universality of Eq.~\eqref{eq:comm-conditions}: the saturation value for the time-averaged OTOC of systems with the given set of generators $\gen$ has a different saturation value than the one of a generic symmetric system (with more general generators chosen from $\gent$).
For example, in the matchgate example of Eq.~\eqref{eq:matchgate-gens} we show in App.~\ref{sec:gates} that
\begin{equation}
    \overline{C_{Z_jZ_j}}(\infty) = 1-8\cdot \frac{L-1}{2L^2-L}\label{eq:how-to-otoc-mg}
\end{equation}
while for a generic $\mb Z_2$ symmetric Brownian circuit the OTOC saturates to zero.
Numerical simulations of time-averaged OTOCs for the two types of Brownian circuits, as well as Floquet circuits of the same form, are shown in Fig.~\ref{fig:phys-mg}.
Similar results for any other pair of operators can naturally be extracted from Eq.~\eqref{eq:otoc-projector}.
Note that this diagnostic does not apply to systems with weak non-universality, since the superoperator symmetries for the generators $\gen$ are the same as those for $\gent$ in that case by definition.

Note that the Brownian circuit formalism from App.~\ref{sec:brown} used for the derivation of Eq.~\eqref{eq:otoc-projector} can be extended further to higher point correlation functions of the form
\begin{equation}
    C_{{\{A_i B_i\}}}(\infty) \defn \bigg\langle\prod_{i=1}^k [A_i(t)B_i(0)]\bigg\rangle.
\label{eq:highercopycorr}
\end{equation}
In this general case, the late-time behavior of such correlation functions is expected to depend on 
$k$-copy conserved quantities of the form $\ket{\mc Q}_{2k}$, in direct analogy to Eq.~(\ref{eq:otoc-projector}).
One set of higher copy conserved quantities can be derived from the physical symmetries and the superoperator symmetries, hence systems exhibiting strong non-universality (i.e., with non-trivial superoperator symmetries) would show signatures in these higher point correlation functions.
There could in principle exist non-trivial higher copy conserved quantities, that affect the behavior of higher point correlation functions, and we defer an exploration of such cases to future work.
\subsection{R\'enyi Entropies}
Signatures of non-universality also occur in the entanglement entropy, which is a standard quantity that can be studied to reveal information about complex dynamics in quantum systems.
Given a bipartition of the Hilbert space $\mc H=\mc H_A\otimes\mc H_{\bar A}$,
if $\rho_A=\tr_{\bar A}(\ketbra{\psi}{\psi})$ is the reduced density matrix for the subsystem $A$, the $k$-th Rényi entropies are defined as\footnote{The von Neumann entropy $S_{A}=-\tr(\rho_A\log\rho_A)$ is the $k \rightarrow 1$ limit of $S^{(k)}_A$.}
\begin{equation}
    S^{(k)}_{A}\defeq\frac{1}{1-k}\log(\tr(\rho_A^k)).\label{eq:def-renyi}
\end{equation}
R\'{e}nyi entropies are well-defined for any $k\in(1,\infty)$, but for integer $k>1$ the expression in Eq.~\eqref{eq:def-renyi} takes a precise meaning in terms of matrix multiplication, and the ``replica trick'' can be used to interpret it in terms of a many-copy system~\cite{calabrese2004entanglement}.
To do so we start by considering the following identities:
\begin{equation}
    \begin{aligned}
        \prescript{}{2k\!}{\bigg \langle}
        \tp{\varcdimercr{0.4}{0}{$\1$}{white}\alabel{0.8}{$...$}\varcdimercr{1.6}{1.2}{$\1$}{white}}
        \,\bigg |\,
        \tp{\varcdimercr{0.4}{0}{$O_1$}{white}\alabel{0.8}{$...$}\varcdimercr{1.6}{1.2}{$O_k$}{white}}
        \bigg \rangle_{\!2k}
        &= \prod_{j=1}^k\tr(O_j), \\
        \prescript{}{2k\!}{\bigg \langle}
        \tp{\varcdimerca{0.4}{0.7}{$\1$}{white}\varcdimercb{0.9}{1.2}{}{white}\alabel{0.8}{$...$}\varcdimercr{1.6}{0.0}{$\1$}{white}}
        \,\bigg |\,
        \tp{\varcdimercr{0.4}{0}{$O_1$}{white}\alabel{0.8}{$...$}\varcdimercr{1.6}{1.2}{$O_k$}{white}}
        \bigg \rangle_{\!2k}
        &= \tr(O_1... O_k).
    \end{aligned}
\end{equation}
We can then define a \textit{domain wall configuration} as the following state defined on $(\mc H_A)^{\ot 2k}\ot(\mc H_{\bar A})^{\ot 2k}$
\begin{equation}\label{eq:domainwalls}
    \ket{A\!:\!\bar A}_{2k}\defeq
    \tket{\varcdimercr{0.4}{0}{$\1$}{white}\alabel{0.8}{$...$}\varcdimercr{1.6}{1.2}{$\1$}{white}}_{\!\!\bar A}\ot
    \tket{\varcdimerca{0.4}{0.7}{$\1$}{white}\varcdimercb{0.9}{1.2}{}{white}\alabel{0.8}{$...$}\varcdimercr{1.6}{0.0}{$\1$}{white}}_{\!\! A}
\end{equation}
so that the $k$-th \textit{purity} of a state $\ket{\psi}$ can be expressed as
\begin{equation}
    \tr(\rho_A^k)=\prescript{}{2k}{\bra{A\!:\!\bar A}}\left(\ket{\psi}^{\ot 2k}\right).
\end{equation}
A typical state in the Hilbert space has entanglement close to the maximal possible value for any bipartition.
This classical result in information is usually stated in terms of the average entanglement of states in the Hilbert space as a function of bipartition size, known as the Page curve \cite{page1993}.
The idea of the Page curve has been extended to various other settings, such as for the average von Neumann entropy of Gaussian states \cite{PhysRevB.103.L241118, PhysRevB.104.214306,PRXQuantum.3.030201, Yu_2023} or constrained systems~\cite{morampudi2020universal}.

Here we show that some of these Page curves can be understood in terms of superoperator conserved quantities, which, in systems with strong non-universality, can lead to deviations from the expected behavior.
Similar to the higher point correlation functions discussed in the previous section, this can be concretely studied in the setting of Brownian circuits generated by $\gen$. Focusing on $k = 2$, we can write the purity as
\begin{equation}
     \tr(\rho_A^2(t))=\prescript{}{4}{\bra{A\!:\!\bar A}}U\ot U^*\ot U\ot U^* \ket{\psi}^{\ot 4},
     \label{eq:renyiexpr}
\end{equation}
and similar to Eq.~(\ref{eq:fourcopyavg}) we obtain
\begin{equation}
    \overline{U \otimes U^\ast \otimes U \otimes U^\ast} = e^{-\widetilde{P}_4 t},
\end{equation}
where $\widetilde{P}_4 \defn \sum_\alpha \left(\ad{h_\alpha}\ot\1+\1\ot\ad{h_\alpha}\right)^2$ is related to $P_4$ in Eq.~(\ref{eq:fourcopyavg}) 
through a partial transposition $\tp{\vdotc{0}{$3$}}\leftrightarrow\tp{\vdotc{0}{$4$}}$.
By going to long times we obtain an expression similar to Eq.~\eqref{eq:otoc-projector}
\begin{equation}\label{eq:renyisuperoperatorformula}
    \overline{\tr(\rho_A^2(\infty))} =\sum_{\widetilde{\mc Q}}\prescript{}{4}{\langle}{A\!:\!\bar A}\ket*{\widetilde{\mc Q}}_{4}\!\!\bra*{\widetilde{\mc Q}}\!\psi\rangle^{\ot 4}
\end{equation}
where $\{|{\widetilde{\mc Q}}\rangle_4\}$ forms an orthonormal basis of the ground state space of $\widetilde{P}_4$, which is related to the superoperator symmetries through the same partial transposition shown above.
Finally, by exploiting the convexity properties of the logarithm we find the following lower bound to the asymptotic value of the entropy\footnote{In many cases numerical simulations show close agreement between this lower bound and the asymptotic value for the entanglement, see also \cite{swann2023spacetime}.
Analytical predictions on the closeness of this lower bound would require the study of the variance of the purity, which would require an analysis of a higher number of replicas.}:
\begin{equation}
    \overline{S_A^{(2)}(\infty)} \geq -\log(\overline{\tr(\rho_A^2(\infty))})
    \label{eq:how-to-page}
\end{equation}
Note that there is usually a non-trivial dependence on the initial-state, as it selects the weight of each superoperator $|{\widetilde{\mc Q}}\rangle_4$ in the sum in Eq.~(\ref{eq:renyisuperoperatorformula}).
Intuitively, the choice of the initial state determines the portion of the Hilbert space that will be explored by the trajectory; superoperators $\mc Q\in\scommt$ emerging from the conventional symmetries (such as Eq.~\eqref{eq:dimer-trivial-comm}) simply constrain the state to evolve within a given symmetry sector, while the superoperators associated to strong non-universality are responsible for less trivial constraints.

We now use the R.H.S. of Eq.~(\ref{eq:how-to-page}) to obtain the page curves for a couple of relevant cases. Despite in general being only an inequality, simulations tend to agree very well with the predicted lower bound, showing that the fluctuations of the purity around its mean value are not too strong.
As a quick check, we can apply Eq.~\eqref{eq:how-to-page} to the case of a universal gate set, where $\scomm=\llangle\oketbra{\1}{\1}\rrangle$, and we obtain
\begin{equation}
    \begin{aligned}
        \overline{\tr(\rho_A^2(\infty))} &=\frac{(2^{\ell}-2^{-\ell})+(2^{L-\ell}-2^{\ell-L})}{2^L-2^{-L}}\label{eq:how-to-page-uni}\\
        &= \frac{1}{2^\ell}+\frac{1}{2^{L-\ell}}+o(2^{-L})
    \end{aligned}
\end{equation}
where $\ell$ is the size of the subsystem $A$. In this case there is no initial state dependence since the overlap of $\oketbra{\1}{\1}$ with $\ket\psi^{\ot 4}$ is equal to $\braket{\psi}{\psi}^2$. 
A similar result has been obtained in Ref.~\cite{Liu_2021}, and this yields the expected form of the Page curve (shown in Fig.~\ref{fig:phys-mg}b).
Conventional symmetries such as $\mb Z_2$ or $U(1)$ group symmetries tend to correct the Page result by effectively reducing the size of the Hilbert space to the size of the symmetry sectors that overlap with the initial state, and explicit computations can be done using similar methods.
For the purposes of non-universality, it is then clear that the non-trivial superoperator symmetries responsible for strong non-universality contribute in Eq.~(\ref{eq:renyisuperoperatorformula}), and can change the nature of the Page curve.
For example in App.~\ref{sec:gates} we compute the prediction of Eq.~(\ref{eq:how-to-page}) in the matchgate example of Eq.~\eqref{eq:matchgate-gens} for a Gaussian initial state.
These are a set of states that is preserved by free-fermionic evolution, and that can for example be obtained by evolving the vacuum state $\ket{\downarrow...\downarrow}$. By performing a random matchgate time evolution of such a state, we obtain the following asymptotic value for the purity 
\begin{equation}
    \overline{\tr(\rho_A^2(\infty))} =\frac{1}{2^{L-\ell}}\sum_{k=0}^{L-\ell}\frac{\binom{L}{k}\binom{2(L-\ell)}{2k}}{\binom{2L}{2k}}.\label{eq:how-to-page-mg}
\end{equation}
This yields a Page curve (shown in Fig.~\ref{fig:phys-mg}) that qualitatively shows the same features as analogous results for the average von Neumann entropy for Gaussian states derived using other methods in Refs.~\cite{PhysRevB.103.L241118, PhysRevB.104.214306,PRXQuantum.3.030201, Yu_2023}.
Similar to correlation functions such as Eq.~(\ref{eq:highercopycorr}), we expect that the saturation value of the higher Rényi entropies depend on the higher copy conserved quantities.
Since these higher copy conserved quantities can be derived from the physical and superoperator symmetries, they too would contain signatures of strong non-universality.

\subsection{Weak Non-Universality and $k$-Designs}
In the previous sections we have shown examples of signatures of non-universality which can appear in quantities at the so-called 2-copy level, i.e., those whose expressions involve $(U \otimes U^\ast)^{\otimes 2}$ in some order, e.g., see Eqs.~(\ref{eq:otoc-projector}) or (\ref{eq:renyiexpr}).
Due to the connection (Eq.~\eqref{eq:unitary-patterns}) between the super-commutant $\scomm$ and the conserved superoperators appearing in Eqs.~\eqref{eq:otoc-projector} and \eqref{eq:renyisuperoperatorformula}, these quantities are only suitable for detecting strong non-universality: any discrepancy from the behaviour expected from the minimal super-commutant $\scommt$ indicates the presence of additional conserved superoperators.
Vice versa, the absence of such signatures at the 2-copy level indicates that the system being investigated is weakly non-universal.
It is natural to wonder about possible signatures of weak non-universality at the $k$-copy level for $k>2$.
In general, the statistical properties of $k$ copies of a quantum system are expressed through the notion of $k$-designs~\cite{harrow2009random, brandao2016local, hunterjones2019unitary}: a set of gates is said to be a $k$-design of another set of gates, if the average values of $k$-copy observables (i.e., whose expressions necessarily involve unitaries such as $(U \otimes U^\ast)^{\otimes k}$ or its permutations) under random evolution\footnote{For the case of circuits generated by the unitaries $u_\alpha(\theta)$ of Eq.~\eqref{eq:localgates} one might for example consider Brownian trajectories as in App.~\ref{sec:brown} or the Haar measure on the generated group $\unit$.} of the two systems are indistinguishable.
These observables may for example be higher point correlation functions or Rényi entropies discussed in Sec.~\ref{sec:physicalimplications}, and the unitary evolution can be Brownian circuits with some set of gates (see App.~\ref{sec:brown}).
Recently there has been interest in understanding whether local circuits which are symmetric under the action of some group $G$ are a $k$-design of globally symmetric time evolutions \cite{hearth2023unitary,li2024designslocalrandomquantum}.
As a standard result in the theory of $k$-designs, it can be shown \cite{li2024designslocalrandomquantum} that for such systems the presence of $k$-copy conserved quantities is equivalent to the local circuit not being a $k$-design of the global one.\footnote{This result relies on the compactness of the group of unitaries $\unit$ generated by the local terms, and in general can be applied in our case by ensuring that the generators $h_\alpha\in\gen$ have rational spectrum (see App.~\ref{app:compact})}
As a consequence, weakly non-universal gate sets are always 2-designs, while strongly non-universal gate sets are not.

It is easy to see that if a set of gates is non-universal, then it cannot be a $k$-design for arbitrary large values of $k$, since the latter would imply that all symmetric time-evolution can be produced.
This should also apply to weakly non-universal systems, but since the discrepancy between the DLA $\dlie$ and the set of symmetric operators $\bond$ only consists in the absence of a few scar-like operators, which do not play an important role in the dynamics of the system, one might expect that these are $k$-designs for a large range of values of $k$.
Indeed, it has been shown that 2-local $U(1)$-conserving qubit circuits are $k$-designs for at least $k\sim \mc O(L)$ \cite{hearth2023unitary} and 4-local $SU(d)$-conserving qudit circuits are $k$-designs for at least $k\sim \mc O(L^2)$ \cite{li2024designslocalrandomquantum}.
Therefore signatures of weak non-universality, at least in these cases, only appear for large values of $k$.
In fact, very recently, Refs.~\cite{hulse2024unitary, mitsuhashi2024unitary, mitsuhashi2024characterization} have obtained precise values of $k$ for which semi-universal systems with certain symmetries become $k$-designs.
\section{Conclusions}\label{sec:conclusions}
In this work, we studied the question of \textit{universality} of the unitary operators of the form $\{e^{i \theta h_\alpha}\}$ generated from a given set of terms $\mG = \{h_\alpha\}$, which are usually chosen to be $k$-local on a lattice.
While it is clear that these cannot generate unitaries that have different symmetries than that of $\mG$, even the space of all unitaries with the same symmetries cannot in general be generated. Additional restrictions imposed from the fact that the starting set of unitaries is $k$-local have been shown in earlier literature for some kinds of symmetric unitaries by Marvian and collaborators~\cite{marvian2020locality, hulse2021qudit, marvian2022rotationally, marvian2024abelian}.
In this work we consider instead a completely general set of gates $\mG$, obtain the precise conditions under which they can exhibit different types of non-universality, and understand the full structure of the manifold of unitaries that can be generated.
The key technique that enables this generality is to study the dynamical Lie algebra $\dlie$ associated with the unitary evolution ($\dlie$ is generated by taking nested commutators of operators in $\mG$ and their linear combinations), and phrase it as a problem of connectivity of operators $\mG$ to the rest of the operator Hilbert space under the adjoint action of the commutators in $\mG$.
Just as the connectivity of states in the physical Hilbert space under the action of some terms $\mG$ can be understood in terms of the symmetries of $\mG$, connectivity of operators in the operator Hilbert space under the adjoint action $[h_\alpha,\bullet]$ of $\mG$ can be understood in terms of superoperator symmetries of the adjoint superoperators.
The structure of these superoperator symmetries can be systematically derived using the framework of commutant algebras, which has been successful in understanding a variety of symmetries and their associated block decompositions in the physical Hilbert space~\cite{moudgalya2021hilbert, moudgalya2022from, moudgalya2023numerical} -- here we apply this machinery to understand the block decompositions of the operator Hilbert space. 
From this understanding, we directly obtain concrete criteria for universality in terms of the superoperator conserved quantities.
This allows us to distinguish two classes of systems that exhibit non-universality.
First, there is the phenomenon of weak non-universality, where all the superoperator symmetries of the unitary evolution are derived from the physical symmetries of the system, but the number of realizable unitaries is nevertheless suppressed.
This {implies the condition known as semi-universality, which} has been the focus of many previous works on non-universality~\cite{marvian2020locality,hulse2021qudit,  marvian2022rotationally, marvian2024abelian}, and occurs for local symmetric gates $\mG$ for many physical symmetries.
Second, there is strong non-universality, where the superoperator symmetries are completely distinct from the physical symmetries of the system, which leads to novel kinds of block decompositions that constrain the connectivity of local operators. 
These occur mainly when gates $\mG$ have a special structure, such as in matchgate circuits or systems that have subspaces that map onto matchgate circuits. 
We also discussed physical implications of non-universality, but signatures in simple measures such as out-of-time-ordered correlation functions and the second Rényi entropies only manifest in systems with strong non-universality.
These two classes also connect to conditions for the simulability of a particular gate set from another gate set that were formulated in terms of ``quadratic symmetries"~\cite{zimboras2015symmetry}, which we have shown to be closely related to the superoperator symmetries we obtain.

In future work, it would be interesting to better understand some other aspects of {weakly non-universal} systems in this framework.
For example, the connection between semi-universality~\cite{marvian2024abelian} and weak non-universality is not fully clear. 
While the latter implies the former (see Sec.~\ref{subsec:semiuniversal}), the opposite is not true, as evident from the example of 2-site matchgate systems (cf. footnote \ref{ft:matchgate2}).
It would be interesting to understand whether such examples can exist for larger system sizes.
We also believe that the ideas used in the proof of weak non-universality for the $U(1)$-symmetric case in App.~\ref{sec:app-nplusoneproof} could be generalized to other gate sets, and might provide a useful criterion for proving the absence of non-trivial superoperator symmetries through simple numerical checks.
Moreover, while strong non-universality corresponds to having non-trivial superoperator conserved quantities, weak non-universality should correspond to having non-trivial conserved quantities at the level of $k$-copies of the systems.
This is due to the fact that they do not form $k$-designs for a large enough $k$ due to non-universality~\cite{hearth2023unitary, li2024designslocalrandomquantum}, and an exhaustively characterization of such $k$-copy conserved quantities and their implications for physical quantities is highly desirable. 
On this note, we should point out two recent sets of results in the literature that make progress in this direction.
Ref.~\cite{zhukas2024observation} proposed quantities that detect the degree of semi-universality in $U(1)$-conserving systems, and very recently Refs.~\cite{hulse2024unitary,mitsuhashi2024unitary, mitsuhashi2024characterization} used different techniques to derive criteria for when a given set of semi-universal gates can be a $k$-design. 
It would be interesting to understand and connect these results in the language we develop in this work.
Finally, there are many natural extensions of this problem that would be interesting to understand in this framework.
First, it would be interesting to understand the general effect of adding ancilla degrees of freedom on non-universality, particularly due to results of Ref.~\cite{marvian2020locality, marvian2022rotationally, marvian2024abelian} that show that universality for circuits with some kinds of symmetries can be ``recovered" by the addition of a few ancilla.
Then, it would be interesting to explore similar questions for more general quantum channels, for which we expect that many of the methods developed in this work might readily generalize.
Finally, being able to characterize the constraints that arise in circuits composed of gates with a discrete structure, such as Clifford~\cite{bittel2025completetheorycliffordcommutant} or dual unitary circuits~\cite{piroli2020dual}, would also be very interesting.
On a different note, the method we employ here can also be viewed as an alternate method to study dynamical Lie algebras and obtain their properties such as their dimension, which is a problem of much current interest~\cite{Larocca2022diagnosingbarren,Fontana_2024,Ragone_2024}.
This might also aid in understanding and potentially classifying more general forms of algebras, extending some recent results~\cite{wiersema2023classification, aguilar2024full}. 
This is also related to the question on the kinds of superoperator symmetries that can naturally appear in quantum many-body systems with locality, analogous to the Majorana number conservation in matchgate circuits.
It would be interesting to employ numerical methods such as those discussed in App.~\ref{sec:app-num} to systematically search for such systems.
The MPS method is particularly promising in this regard, as it can rapidly identify -- in 1d -- systems that possess an unusual super-commutant $\scomm$, which might lead to unusual kinds of dynamical Lie algebras.

\textit{Note Added --} During the preparation of this work, Refs.~\cite{hulse2024unitary, mitsuhashi2024unitary, mitsuhashi2024characterization} appeared, which explicitly showed examples of $k$-copy conserved quantities associated to semi-universality.
While Ref.~\cite{mitsuhashi2024characterization} also discusses circuits with general sets of gates, it assumes semi-universality and focuses on the analysis of $k$-designs, whereas here we are interested in explicitly characterizing the extent of non-universality in general.
Our results agree wherever there is some overlap.

\textbf{Data and materials availability:} Data analysis and simulation codes are available on Zenodo~\cite{zenodo}.

\phantom{a}\\ 

\phantom{a}\\

\section*{Acknowledgements}
We particularly thank Lesik Motrunich for enlightining discussions in the early stages of this work. 
We also thank David Gross, Pavel Kos, Barbara Kraus, Tibor Rakovszky, Tomohiro Soejima, David Stephen, Shreya Vardhan for useful discussions. 
We acknowledge support from the Munich Quantum Valley, which is supported by the Bavarian state government with funds from the Hightech Agenda Bayern Plus, and the Munich Center for Quantum Science and Technology (MCQST), supported by the Deutsche Forschungsgemeinschaft (DFG, German Research Foundation) under Germany’s Excellence Strategy--EXC--2111--390814868.
S.M. thanks Lesik Motrunich for collaboration on \cite{moudgalya2022from, moudgalya2023numerical}, and Shreya Vardhan for collaboration on \cite{vardhan2024entanglement}.
\bibliography{newrefs, refs}

@misc{zenodo,
  author       = {Lastres, Marco and Pollmann, Frank and Moudgalya, Sanjay},
  title        = {{Non-Universality from Conserved Superoperators in Unitary Circuits}},
  month        = 12,
  year         = 2025,
  publisher    = {Zenodo},
  doi          = {10.5281/zenodo.17991810},
  url          = {https://zenodo.org/records/17991811}
}

@article{kazi2024permutationinvariant,
   title={On the universality of ${S}_n$-equivariant $k$-body gates},
   volume={26},
   ISSN={1367-2630},
   url={http://dx.doi.org/10.1088/1367-2630/ad4819},
   DOI={10.1088/1367-2630/ad4819},
   number={5},
   journal={New Journal of Physics},
   publisher={IOP Publishing},
   author={Kazi, Sujay and Larocca, Martín and Cerezo, M},
   year={2024},
   month=may, pages={053030}
}

@article{xu2019locality,
  title = {Locality, Quantum Fluctuations, and Scrambling},
  author = {Xu, Shenglong and Swingle, Brian},
  journal = {Phys. Rev. X},
  volume = {9},
  issue = {3},
  pages = {031048},
  numpages = {21},
  year = {2019},
  month = {Sep},
  publisher = {American Physical Society},
  doi = {10.1103/PhysRevX.9.031048},
  url = {https://link.aps.org/doi/10.1103/PhysRevX.9.031048}
}

@Article{lashkari2013towards,
author={Lashkari, Nima
and Stanford, Douglas
and Hastings, Matthew
and Osborne, Tobias
and Hayden, Patrick},
title={Towards the fast scrambling conjecture},
journal={Journal of High Energy Physics},
year={2013},
month={Apr},
day={03},
volume={2013},
number={4},
pages={22},
issn={1029-8479},
doi={10.1007/JHEP04(2013)022},
url={https://doi.org/10.1007/JHEP04(2013)022}
}

@article{rakovszky2018diffusive,
  title = {Diffusive Hydrodynamics of Out-of-Time-Ordered Correlators with Charge Conservation},
  author = {Rakovszky, Tibor and Pollmann, Frank and von Keyserlingk, C. W.},
  journal = {Phys. Rev. X},
  volume = {8},
  issue = {3},
  pages = {031058},
  numpages = {28},
  year = {2018},
  month = {Sep},
  publisher = {American Physical Society},
  doi = {10.1103/PhysRevX.8.031058},
  url = {https://link.aps.org/doi/10.1103/PhysRevX.8.031058}
}

@article{aguilar2024full,
  title={Full classification of Pauli Lie algebras},
  author={Aguilar, Gerard and Cichy, Simon and Eisert, Jens and Bittel, Lennart},
  journal={arXiv preprint arXiv:2408.00081},
  year={2024}
}

@article{zimboras2015symmetry,
  title = {Symmetry criteria for quantum simulability of effective interactions},
  author = {Zimbor\'as, Zolt\'an and Zeier, Robert and Schulte-Herbr\"uggen, Thomas and Burgarth, Daniel},
  journal = {Phys. Rev. A},
  volume = {92},
  issue = {4},
  pages = {042309},
  numpages = {8},
  year = {2015},
  month = {Oct},
  publisher = {American Physical Society},
  doi = {10.1103/PhysRevA.92.042309},
  url = {https://link.aps.org/doi/10.1103/PhysRevA.92.042309}
}

@Article{zeier2011symmetry,
author={Zeier, Robert
and Schulte-Herbr{\"u}ggen, Thomas},
title={Symmetry principles in quantum systems theory},
journal={Journal of Mathematical Physics},
year={2011},
month={Nov},
day={23},
volume={52},
number={11},
pages={113510},
issn={0022-2488},
doi={10.1063/1.3657939},
url={https://doi.org/10.1063/1.3657939}
}

@article{chen2010local,
  title = {Local unitary transformation, long-range quantum entanglement, wave function renormalization, and topological order},
  author = {Chen, Xie and Gu, Zheng-Cheng and Wen, Xiao-Gang},
  journal = {Phys. Rev. B},
  volume = {82},
  issue = {15},
  pages = {155138},
  numpages = {28},
  year = {2010},
  month = {Oct},
  publisher = {American Physical Society},
  doi = {10.1103/PhysRevB.82.155138},
  url = {https://link.aps.org/doi/10.1103/PhysRevB.82.155138}
}

@article{huang2015quantum,
  title = {Quantum circuit complexity of one-dimensional topological phases},
  author = {Huang, Yichen and Chen, Xie},
  journal = {Phys. Rev. B},
  volume = {91},
  issue = {19},
  pages = {195143},
  numpages = {10},
  year = {2015},
  month = {May},
  publisher = {American Physical Society},
  doi = {10.1103/PhysRevB.91.195143},
  url = {https://link.aps.org/doi/10.1103/PhysRevB.91.195143}
}

@article{zhou2020diffusive,
  title = {Diffusive scaling of R\'enyi entanglement entropy},
  author = {Zhou, Tianci and Ludwig, Andreas W. W.},
  journal = {Phys. Rev. Res.},
  volume = {2},
  issue = {3},
  pages = {033020},
  numpages = {9},
  year = {2020},
  month = {Jul},
  publisher = {American Physical Society},
  doi = {10.1103/PhysRevResearch.2.033020},
  url = {https://link.aps.org/doi/10.1103/PhysRevResearch.2.033020}
}

@article{huang2020dynamics,
doi = {10.1088/2633-1357/abd1e2},
url = {https://dx.doi.org/10.1088/2633-1357/abd1e2},
year = {2020},
month = {dec},
publisher = {IOP Publishing},
volume = {1},
number = {3},
pages = {035205},
author = {Yichen Huang},
title = {Dynamics of Rényi entanglement entropy in diffusive qudit systems},
journal = {IOP SciNotes}}

@article{rakovszky2019subballistic,
  title = {Sub-ballistic Growth of R\'enyi Entropies due to Diffusion},
  author = {Rakovszky, Tibor and Pollmann, Frank and von Keyserlingk, C. W.},
  journal = {Phys. Rev. Lett.},
  volume = {122},
  issue = {25},
  pages = {250602},
  numpages = {6},
  year = {2019},
  month = {Jun},
  publisher = {American Physical Society},
  doi = {10.1103/PhysRevLett.122.250602},
  url = {https://link.aps.org/doi/10.1103/PhysRevLett.122.250602}
}

@article{khemani2018operator,
  title = {Operator Spreading and the Emergence of Dissipative Hydrodynamics under Unitary Evolution with Conservation Laws},
  author = {Khemani, Vedika and Vishwanath, Ashvin and Huse, David A.},
  journal = {Phys. Rev. X},
  volume = {8},
  issue = {3},
  pages = {031057},
  numpages = {25},
  year = {2018},
  month = {Sep},
  publisher = {American Physical Society},
  doi = {10.1103/PhysRevX.8.031057},
  url = {https://link.aps.org/doi/10.1103/PhysRevX.8.031057}
}

@article{friedman2019spectral,
  title = {Spectral Statistics and Many-Body Quantum Chaos with Conserved Charge},
  author = {Friedman, Aaron J. and Chan, Amos and De Luca, Andrea and Chalker, J. T.},
  journal = {Phys. Rev. Lett.},
  volume = {123},
  issue = {21},
  pages = {210603},
  numpages = {6},
  year = {2019},
  month = {Nov},
  publisher = {American Physical Society},
  doi = {10.1103/PhysRevLett.123.210603},
  url = {https://link.aps.org/doi/10.1103/PhysRevLett.123.210603}
}

@article{calabrese2004entanglement,
doi = {10.1088/1742-5468/2004/06/P06002},
url = {https://dx.doi.org/10.1088/1742-5468/2004/06/P06002},
year = {2004},
month = {jun},
publisher = {},
volume = {2004},
number = {06},
pages = {P06002},
author = {Pasquale Calabrese and  John Cardy},
title = {Entanglement entropy and quantum field theory},
journal = {Journal of Statistical Mechanics: Theory and Experiment}}

@ARTICLE{wiersema2023classification,
  title={Classification of dynamical Lie algebras of 2-local spin systems on linear, circular and fully connected topologies},
  author={Roeland Wiersema and Efekan K{\"o}kc{\"u} and Alexander F Kemper and Bojko N Bakalov},
  journal={Npj Quantum Information},
  year={2023},
  volume={10},
  url={https://api.semanticscholar.org/CorpusID:261697035}
}

@ARTICLE{zhukas2024observation,
       author = {{Zhukas}, Liudmila A. and {Wang}, Qingfeng and {Katz}, Or and {Monroe}, Christopher and {Marvian}, Iman},
        title = "{Observation of the Symmetry-Protected Signature of 3-body Interactions}",
      journal = {arXiv e-prints},
         year = 2024,
        month = aug,
archivePrefix = {arXiv},
       eprint = {2408.10475},
 primaryClass = {quant-ph},
       adsurl = {https://ui.adsabs.harvard.edu/abs/2024arXiv240810475Z}
}

@ARTICLE{mitsuhashi2024unitary,
       author = {{Mitsuhashi}, Yosuke and {Suzuki}, Ryotaro and {Soejima}, Tomohiro and {Yoshioka}, Nobuyuki},
        title = "{Unitary Designs of Symmetric Local Random Circuits}",
      journal = {arXiv e-prints},
         year = 2024,
        month = aug,
archivePrefix = {arXiv},
       eprint = {2408.13472},
 primaryClass = {quant-ph},
       adsurl = {https://ui.adsabs.harvard.edu/abs/2024arXiv240813472M}
}

@ARTICLE{mitsuhashi2024characterization,
       author = {{Mitsuhashi}, Yosuke and {Suzuki}, Ryotaro and {Soejima}, Tomohiro and {Yoshioka}, Nobuyuki},
        title = "{Characterization of Randomness in Quantum Circuits of Continuous Gate Sets}",
      journal = {arXiv e-prints},
         year = 2024,
        month = aug,
archivePrefix = {arXiv},
       eprint = {2408.13475},
 primaryClass = {quant-ph},
       adsurl = {https://ui.adsabs.harvard.edu/abs/2024arXiv240813475M}
}

@ARTICLE{hulse2024unitary,
       author = {{Hulse}, Austin and {Liu}, Hanqing and {Marvian}, Iman},
        title = "{Unitary Designs from Random Symmetric Quantum Circuits}",
      journal = {arXiv e-prints},
         year = 2024,
        month = aug,
archivePrefix = {arXiv},
       eprint = {2408.14463},
 primaryClass = {quant-ph},
       adsurl = {https://ui.adsabs.harvard.edu/abs/2024arXiv240814463H}
}

@article{morampudi2020universal,
  title = {Universal Entanglement of Typical States in Constrained Systems},
  author = {Morampudi, S. C. and Chandran, A. and Laumann, C. R.},
  journal = {Phys. Rev. Lett.},
  volume = {124},
  issue = {5},
  pages = {050602},
  numpages = {6},
  year = {2020},
  month = {Feb},
  publisher = {American Physical Society},
  doi = {10.1103/PhysRevLett.124.050602},
  url = {https://link.aps.org/doi/10.1103/PhysRevLett.124.050602}
}

@misc{moudgalya2023symmetries,
   title={Symmetries as Ground States of Local Superoperators: Hydrodynamic Implications},
   volume={5},
   ISSN={2691-3399},
   url={http://dx.doi.org/10.1103/PRXQuantum.5.040330},
   DOI={10.1103/prxquantum.5.040330},
   number={4},
   journal={PRX Quantum},
   publisher={American Physical Society (APS)},
   author={Moudgalya, Sanjay and Motrunich, Olexei I.},
   year={2024},
   month=nov
}

@article{dhar2020revisiting,
title = {{Revisiting the Mazur bound and the Suzuki equality}},
journal = {Chaos, Solitons \& Fractals},
volume = {144},
pages = {110618},
year = {2021},
issn = {0960-0779},
doi = {https://doi.org/10.1016/j.chaos.2020.110618},
url = {https://www.sciencedirect.com/science/article/pii/S0960077920310092},
author = {Abhishek Dhar and Aritra Kundu and Keiji Saito}}

@article{mazurbound1969,
	title        = {Non-ergodicity of phase functions in certain systems},
	author       = {P. Mazur},
	year         = 1969,
	journal      = {Physica},
	volume       = 43,
	number       = 4,
	pages        = {533--545},
	doi          = {https://doi.org/10.1016/0031-8914(69)90185-2},
	issn         = {0031-8914},
	url          = {http://www.sciencedirect.com/science/article/pii/0031891469901852}
}

@inbook{moudgalya2019thermalization,
author = {Sanjay   Moudgalya  and  Abhinav   Prem  and  Rahul   Nandkishore  and  Nicolas   Regnault  and  B. Andrei   Bernevig },
title = {{Thermalization and Its Absence within Krylov Subspaces of a Constrained Hamiltonian}},
booktitle = {Memorial Volume for Shoucheng Zhang},
chapter = {7},
pages = {147-209},
doi = {10.1142/9789811231711_0009},
URL = {https://www.worldscientific.com/doi/abs/10.1142/9789811231711_0009},
}

@article{turner2017quantum,
	title        = {Weak ergodicity breaking from quantum many-body scars},
	author       = {Turner, CJ and Michailidis, AA and Abanin, DA and Serbyn, M and Papic, Z},
	year         = 2018,
	journal      = {Nature Physics},
	publisher    = {Springer Nature},
	volume       = 14,
	number       = 7,
	pages        = {745--749},
	doi          = {10.1038/s41567-018-0137-5},
	url          = {https://www.nature.com/articles/s41567-018-0137-5}
}

@article{perezgarcia2007matrix,
	title        = {Matrix Product State Representations},
	author       = {Perez-Garcia, D. and Verstraete, F. and Wolf, M. M. and Cirac, J. I.},
	year         = 2007,
	month        = jul,
	journal      = {Quantum Info. Comput.},
	publisher    = {Rinton Press, Incorporated},
	address      = {Paramus, NJ},
	volume       = 7,
	number       = 5,
	pages        = {401–430},
	issn         = {1533-7146},
	issue_date   = {July 2007},
	numpages     = 30
}

@article{zanardi2001virtual,
	title        = {Virtual Quantum Subsystems},
	author       = {Zanardi, Paolo},
	year         = 2001,
	month        = {Jul},
	journal      = {Phys. Rev. Lett.},
	publisher    = {American Physical Society},
	volume       = 87,
	pages        = {077901},
	doi          = {10.1103/PhysRevLett.87.077901},
	url          = {https://link.aps.org/doi/10.1103/PhysRevLett.87.077901},
	issue        = 7,
	numpages     = 4
}

@article{bartlett2007reference,
	title        = {Reference frames, superselection rules, and quantum information},
	author       = {Bartlett, Stephen D. and Rudolph, Terry and Spekkens, Robert W.},
	year         = 2007,
	month        = {Apr},
	journal      = {Rev. Mod. Phys.},
	publisher    = {American Physical Society},
	volume       = 79,
	pages        = {555--609},
	doi          = {10.1103/RevModPhys.79.555},
	url          = {https://link.aps.org/doi/10.1103/RevModPhys.79.555},
	issue        = 2,
	numpages     = {0}
}

@article{lidar1998decoherence,
	title        = {Decoherence-Free Subspaces for Quantum Computation},
	author       = {Lidar, D. A. and Chuang, I. L. and Whaley, K. B.},
	year         = 1998,
	month        = {Sep},
	journal      = {Phys. Rev. Lett.},
	publisher    = {American Physical Society},
	volume       = 81,
	pages        = {2594--2597},
	doi          = {10.1103/PhysRevLett.81.2594},
	url          = {https://link.aps.org/doi/10.1103/PhysRevLett.81.2594},
	issue        = 12,
	numpages     = {0}
}

@inbook{lidar2003decoherencereview,
	title        = {Decoherence-Free Subspaces and Subsystems},
	author       = {Lidar, Daniel A. and Birgitta Whaley, K.},
	year         = 2003,
	booktitle    = {Irreversible Quantum Dynamics},
	publisher    = {Springer Berlin Heidelberg},
	address      = {Berlin, Heidelberg},
	pages        = {83--120},
	doi          = {10.1007/3-540-44874-8_5},
	url          = {https://doi.org/10.1007/3-540-44874-8_5},
	editor       = {Benatti, Fabio and Floreanini, Roberto}
}

@Article{haferkamp2022linear,
author={Haferkamp, Jonas
and Faist, Philippe
and Kothakonda, Naga B. T.
and Eisert, Jens
and Yunger Halpern, Nicole},
title={Linear growth of quantum circuit complexity},
journal={Nature Physics},
year={2022},
month={May},
day={01},
volume={18},
number={5},
pages={528-532},
issn={1745-2481},
doi={10.1038/s41567-022-01539-6},
url={https://doi.org/10.1038/s41567-022-01539-6}
}

@article{lent1993quantum,
doi = {10.1088/0957-4484/4/1/004},
url = {https://dx.doi.org/10.1088/0957-4484/4/1/004},
year = {1993},
month = {jan},
publisher = {},
volume = {4},
number = {1},
pages = {49},
author = {C S Lent and  P D Tougaw and  W Porod and  G H Bernstein},
title = {Quantum cellular automata},
journal = {Nanotechnology},
}

@article{farrelly2020review,
  doi = {10.22331/q-2020-11-30-368},
  url = {https://doi.org/10.22331/q-2020-11-30-368},
  title = {A review of {Q}uantum {C}ellular {A}utomata},
  author = {Farrelly, Terry},
  journal = {{Quantum}},
  issn = {2521-327X},
  publisher = {{Verein zur F{\"{o}}rderung des Open Access Publizierens in den Quantenwissenschaften}},
  volume = {4},
  pages = {368},
  month = nov,
  year = {2020}
}

@Article{arrighi2019overview,
author={Arrighi, P.},
title={An overview of quantum cellular automata},
journal={Natural Computing},
year={2019},
month={Dec},
day={01},
volume={18},
number={4},
pages={885-899},
issn={1572-9796},
doi={10.1007/s11047-019-09762-6},
url={https://doi.org/10.1007/s11047-019-09762-6}
}

@Article{bulchandani2021smooth,
author={Bulchandani, Vir B.
and Sondhi, S. L.},
title={How smooth is quantum complexity?},
journal={Journal of High Energy Physics},
year={2021},
month={Oct},
day={28},
volume={2021},
number={10},
pages={230},
issn={1029-8479},
doi={10.1007/JHEP10(2021)230},
url={https://doi.org/10.1007/JHEP10(2021)230}
}

@article{nielsen2005geometric,
       author = {{Nielsen}, Michael A.},
        title = "{A geometric approach to quantum circuit lower bounds}",
      journal = {arXiv e-prints},
         year = 2005,
        month = feb,
          eid = {quant-ph/0502070},
          doi = {10.48550/arXiv.quant-ph/0502070},
archivePrefix = {arXiv},
 primaryClass = {quant-ph},
       adsurl = {https://ui.adsabs.harvard.edu/abs/2005quant.ph..2070N}
}

@article{brown2018second,
  title = {Second law of quantum complexity},
  author = {Brown, Adam R. and Susskind, Leonard},
  journal = {Phys. Rev. D},
  volume = {97},
  issue = {8},
  pages = {086015},
  numpages = {29},
  year = {2018},
  month = {Apr},
  publisher = {American Physical Society},
  doi = {10.1103/PhysRevD.97.086015},
  url = {https://link.aps.org/doi/10.1103/PhysRevD.97.086015}
}

@Article{piroli2022random,
author={Piroli, Lorenzo},
title={Random circuits have no shortcuts},
journal={Nature Physics},
year={2022},
month={May},
day={01},
volume={18},
number={5},
pages={482-483},
issn={1745-2481},
doi={10.1038/s41567-022-01559-2},
url={https://doi.org/10.1038/s41567-022-01559-2}
}

@article{moudgalya2021spectral,
  title = {Spectral statistics in constrained many-body quantum chaotic systems},
  author = {Moudgalya, Sanjay and Prem, Abhinav and Huse, David A. and Chan, Amos},
  journal = {Phys. Rev. Res.},
  volume = {3},
  issue = {2},
  pages = {023176},
  numpages = {27},
  year = {2021},
  month = {Jun},
  publisher = {American Physical Society},
  doi = {10.1103/PhysRevResearch.3.023176},
  url = {https://link.aps.org/doi/10.1103/PhysRevResearch.3.023176}
}

@article{batista2000tJz,
	title        = {Quantum Phase Diagram of the $\mathit{t}\ensuremath{-}{J}_{z}$ Chain Model},
	author       = {Batista, C. D. and Ortiz, G.},
	year         = 2000,
	month        = {Nov},
	journal      = {Phys. Rev. Lett.},
	publisher    = {American Physical Society},
	volume       = 85,
	pages        = {4755--4758},
	doi          = {10.1103/PhysRevLett.85.4755},
	url          = {https://link.aps.org/doi/10.1103/PhysRevLett.85.4755},
	issue        = 22,
	numpages     = {0}
}

@Article{serbyn2020review,
author={Serbyn, Maksym
and Abanin, Dmitry A.
and Papi{\'{c}}, Zlatko},
title={Quantum many-body scars and weak breaking of ergodicity},
journal={Nature Physics},
year={2021},
month={Jun},
day={01},
volume={17},
number={6},
pages={675-685},
issn={1745-2481},
doi={10.1038/s41567-021-01230-2},
url={https://doi.org/10.1038/s41567-021-01230-2}
}

@article{yang2019hilbertspace,
	title        = {Hilbert-Space Fragmentation from Strict Confinement},
	author       = {Yang, Zhi-Cheng and Liu, Fangli and Gorshkov, Alexey V. and Iadecola, Thomas},
	year         = 2020,
	month        = {May},
	journal      = {Phys. Rev. Lett.},
	publisher    = {American Physical Society},
	volume       = 124,
	pages        = 207602,
	doi          = {10.1103/PhysRevLett.124.207602},
	url          = {https://link.aps.org/doi/10.1103/PhysRevLett.124.207602},
	issue        = 20,
	numpages     = 6
}

@article{sala2020fragmentation,
	title        = {Ergodicity Breaking Arising from Hilbert Space Fragmentation in Dipole-Conserving Hamiltonians},
	author       = {Sala, Pablo and Rakovszky, Tibor and Verresen, Ruben and Knap, Michael and Pollmann, Frank},
	year         = 2020,
	month        = {Feb},
	journal      = {Phys. Rev. X},
	publisher    = {American Physical Society},
	volume       = 10,
	pages        = {011047},
	doi          = {10.1103/PhysRevX.10.011047},
	url          = {https://link.aps.org/doi/10.1103/PhysRevX.10.011047},
	issue        = 1,
	numpages     = 19
}

@article{rakovszky2020statistical,
	title        = {Statistical localization: From strong fragmentation to strong edge modes},
	author       = {Rakovszky, Tibor and Sala, Pablo and Verresen, Ruben and Knap, Michael and Pollmann, Frank},
	year         = 2020,
	month        = {Mar},
	journal      = {Phys. Rev. B},
	publisher    = {American Physical Society},
	volume       = 101,
	pages        = 125126,
	doi          = {10.1103/PhysRevB.101.125126},
	url          = {https://link.aps.org/doi/10.1103/PhysRevB.101.125126},
	issue        = 12,
	numpages     = 23
}

@article{landsman1998lecture,
	title        = {{Lecture notes on $C^*$-algebras, Hilbert $C^*$-modules, and quantum mechanics}},
	author       = {Landsman, Nicolas P},
	year         = 1998,
	journal      = {arXiv preprint math-ph/9807030}
}

@article{moudgalya2021hilbert,
  title = {Hilbert Space Fragmentation and Commutant Algebras},
  author = {Moudgalya, Sanjay and Motrunich, Olexei I.},
  journal = {Phys. Rev. X},
  volume = {12},
  issue = {1},
  pages = {011050},
  numpages = {44},
  year = {2022},
  month = {Mar},
  publisher = {American Physical Society},
  doi = {10.1103/PhysRevX.12.011050},
  url = {https://link.aps.org/doi/10.1103/PhysRevX.12.011050}
}

@article{harlow2017,
	title        = {The Ryu--Takayanagi Formula from Quantum Error Correction},
	author       = {Harlow, Daniel},
	year         = 2017,
	month        = {Sep},
	day          = {01},
	journal      = {Communications in Mathematical Physics},
	volume       = 354,
	number       = 3,
	pages        = {865--912},
	doi          = {10.1007/s00220-017-2904-z},
	issn         = {1432-0916},
	url          = {https://doi.org/10.1007/s00220-017-2904-z}
}

@book{fulton2013representation,
	title        = {Representation theory: a first course},
	author       = {Fulton, William and Harris, Joe},
	year         = 2013,
	publisher    = {Springer Science \& Business Media},
	volume       = 129
}

@article{moudgalya2022from,
title = {{From symmetries to commutant algebras in standard Hamiltonians}},
journal = {Annals of Physics},
volume = {455},
pages = {169384},
year = {2023},
issn = {0003-4916},
doi = {https://doi.org/10.1016/j.aop.2023.169384},
url = {https://www.sciencedirect.com/science/article/pii/S0003491623001707},
author = {Sanjay Moudgalya and Olexei I. Motrunich}}

@article{moudgalya2022exhaustive,
   title={Exhaustive Characterization of Quantum Many-Body Scars Using Commutant Algebras},
   volume={14},
   ISSN={2160-3308},
   url={http://dx.doi.org/10.1103/PhysRevX.14.041069},
   DOI={10.1103/physrevx.14.041069},
   number={4},
   journal={Physical Review X},
   publisher={American Physical Society (APS)},
   author={Moudgalya, Sanjay and Motrunich, Olexei I.},
   year={2024},
   month=dec
}

@Article{cotler2017chaos,
author={Cotler, Jordan
and Hunter-Jones, Nicholas
and Liu, Junyu
and Yoshida, Beni},
title={Chaos, complexity, and random matrices},
journal={Journal of High Energy Physics},
year={2017},
month={Nov},
day={09},
volume={2017},
number={11},
pages={48},
issn={1029-8479},
doi={10.1007/JHEP11(2017)048},
url={https://doi.org/10.1007/JHEP11(2017)048}
}

@article{spee2018mode,
  title = {Mode entanglement of Gaussian fermionic states},
  author = {Spee, C. and Schwaiger, K. and Giedke, G. and Kraus, B.},
  journal = {Phys. Rev. A},
  volume = {97},
  issue = {4},
  pages = {042325},
  numpages = {20},
  year = {2018},
  month = {Apr},
  publisher = {American Physical Society},
  doi = {10.1103/PhysRevA.97.042325},
  url = {https://link.aps.org/doi/10.1103/PhysRevA.97.042325}
}

@article{bravyi2004lagrangian,
  title={Lagrangian representation for fermionic linear optics},
  author={Sergey Bravyi},
  journal={Quantum Inf. Comput.},
  year={2004},
  volume={5},
  pages={216-238},
  url={https://api.semanticscholar.org/CorpusID:43803405}
}

@article{bao2021symmetry,
title = {Symmetry enriched phases of quantum circuits},
journal = {Annals of Physics},
volume = {435},
pages = {168618},
year = {2021},
note = {Special issue on Philip W. Anderson},
issn = {0003-4916},
doi = {https://doi.org/10.1016/j.aop.2021.168618},
url = {https://www.sciencedirect.com/science/article/pii/S0003491621002244},
author = {Yimu Bao and Soonwon Choi and Ehud Altman}}

@article{prosen2008third,
doi = {10.1088/1367-2630/10/4/043026},
url = {https://dx.doi.org/10.1088/1367-2630/10/4/043026},
year = {2008},
month = {apr},
publisher = {},
volume = {10},
number = {4},
pages = {043026},
author = {Tomaž Prosen},
title = {Third quantization: a general method to solve master equations for quadratic open Fermi systems},
journal = {New Journal of Physics}
}

@ARTICLE{jozsa2008matchgate,
   title={Matchgates and classical simulation of quantum circuits},
   volume={464},
   ISSN={1471-2946},
   url={http://dx.doi.org/10.1098/rspa.2008.0189},
   DOI={10.1098/rspa.2008.0189},
   number={2100},
   journal={Proceedings of the Royal Society A: Mathematical, Physical and Engineering Sciences},
   publisher={The Royal Society},
   author={Jozsa, Richard and Miyake, Akimasa},
   year={2008},
   month=jul, pages={3089–3106}
}

@article{terhal2002classical,
  title = {Classical simulation of noninteracting-fermion quantum circuits},
  author = {Terhal, Barbara M. and DiVincenzo, David P.},
  journal = {Phys. Rev. A},
  volume = {65},
  issue = {3},
  pages = {032325},
  numpages = {10},
  year = {2002},
  month = {Mar},
  publisher = {American Physical Society},
  doi = {10.1103/PhysRevA.65.032325},
  url = {https://link.aps.org/doi/10.1103/PhysRevA.65.032325}
}

@article{valiant2002quantum,
author = {Valiant, Leslie G.},
title = {Quantum Circuits That Can Be Simulated Classically in Polynomial Time},
journal = {SIAM Journal on Computing},
volume = {31},
number = {4},
pages = {1229-1254},
year = {2002},
doi = {10.1137/S0097539700377025},
URL = {https://doi.org/10.1137/S0097539700377025},
eprint = {https://doi.org/10.1137/S0097539700377025}}

@Article{brandao2016local,
author={Brand{\~a}o, Fernando G. S. L.
and Harrow, Aram W.
and Horodecki, Micha{\l}},
title={Local Random Quantum Circuits are Approximate Polynomial-Designs},
journal={Communications in Mathematical Physics},
year={2016},
month={Sep},
day={01},
volume={346},
number={2},
pages={397-434},
issn={1432-0916},
doi={10.1007/s00220-016-2706-8},
url={https://doi.org/10.1007/s00220-016-2706-8}
}

@Article{harrow2009random,
author={Harrow, Aram W.
and Low, Richard A.},
title={Random Quantum Circuits are Approximate 2-designs},
journal={Communications in Mathematical Physics},
year={2009},
month={Oct},
day={01},
volume={291},
number={1},
pages={257-302},
issn={1432-0916},
doi={10.1007/s00220-009-0873-6},
url={https://doi.org/10.1007/s00220-009-0873-6}
}

@ARTICLE{hunterjones2019unitary,
       author = {{Hunter-Jones}, Nicholas},
        title = "{Unitary designs from statistical mechanics in random quantum circuits}",
      journal = {arXiv e-prints},
         year = 2019,
        month = may,
archivePrefix = {arXiv},
       eprint = {1905.12053},
 primaryClass = {quant-ph},
       adsurl = {https://ui.adsabs.harvard.edu/abs/2019arXiv190512053H}
}

@article{piroli2020dual,
  title = {Exact dynamics in dual-unitary quantum circuits},
  author = {Piroli, Lorenzo and Bertini, Bruno and Cirac, J. Ignacio and Prosen, Toma\ifmmode \check{z}\else \v{z}\fi{}},
  journal = {Phys. Rev. B},
  volume = {101},
  issue = {9},
  pages = {094304},
  numpages = {16},
  year = {2020},
  month = {Mar},
  publisher = {American Physical Society},
  doi = {10.1103/PhysRevB.101.094304},
  url = {https://link.aps.org/doi/10.1103/PhysRevB.101.094304}
}

@ARTICLE{vardhan2024entanglement,
       author = {{Vardhan}, Shreya and {Moudgalya}, Sanjay},
        title = "{Entanglement dynamics from universal low-lying modes}",
      journal = {arXiv e-prints},
         year = 2024,
        month = jul,
archivePrefix = {arXiv},
       eprint = {2407.16763},
 primaryClass = {cond-mat.stat-mech},
       adsurl = {https://ui.adsabs.harvard.edu/abs/2024arXiv240716763V}
}

@article{moudgalya2021review,
	title        = {Quantum many-body scars and Hilbert space fragmentation: a review of exact results},
	author       = {Sanjay Moudgalya and B Andrei Bernevig and Nicolas Regnault},
	year         = 2022,
	month        = {jul},
	journal      = {Reports on Progress in Physics},
	publisher    = {{IOP} Publishing},
	volume       = 85,
	number       = 8,
	pages        = {086501},
	doi          = {10.1088/1361-6633/ac73a0},
	url          = {https://doi.org/10.1088/1361-6633/ac73a0}
}

@misc{papic2021review,
author="Papi{\'{c}}, Zlatko",
editor="Bayat, Abolfazl
and Bose, Sougato
and Johannesson, Henrik",
title="Weak Ergodicity Breaking Through the Lens of Quantum Entanglement",
bookTitle="Entanglement in Spin Chains: From Theory to Quantum Technology Applications",
year="2022",
publisher="Springer International Publishing",
address="Cham",
pages="341--395",
isbn="978-3-031-03998-0",
doi="10.1007/978-3-031-03998-0_13",
url="https://doi.org/10.1007/978-3-031-03998-0_13"
}

@article{chandran2022review,
author = {Chandran, Anushya and Iadecola, Thomas and Khemani, Vedika and Moessner, Roderich},
title = {{Quantum Many-Body Scars: A Quasiparticle Perspective}},
journal = {Annual Review of Condensed Matter Physics},
volume = {14},
number = {1},
pages = {443-469},
year = {2023},
doi = {10.1146/annurev-conmatphys-031620-101617},
URL = {https://doi.org/10.1146/annurev-conmatphys-031620-101617}
}

@Article{bauer2017stochastic,
	title={{Stochastic dissipative quantum spin chains (I) : Quantum fluctuating  discrete hydrodynamics}},
	author={Michel Bauer and Denis Bernard and Tony Jin},
	journal={SciPost Phys.},
	volume={3},
	pages={033},
	year={2017},
	publisher={SciPost},
	doi={10.21468/SciPostPhys.3.5.033},
	url={https://scipost.org/10.21468/SciPostPhys.3.5.033},
}

@article{page1993,
	title        = {Average entropy of a subsystem},
	author       = {Page, Don N.},
	year         = 1993,
	month        = {Aug},
	journal      = {Physical Review Letters},
	publisher    = {American Physical Society},
	volume       = 71,
	pages        = {1291--1294},
	doi          = {10.1103/PhysRevLett.71.1291},
	url          = {https://link.aps.org/doi/10.1103/PhysRevLett.71.1291},
	issue        = 9,
	numpages     = {0}
}

@inbook{lidar2014dfs,
	title        = {Review of Decoherence-Free Subspaces, Noiseless Subsystems, and Dynamical Decoupling},
	author       = {Lidar, Daniel A.},
	year         = 2014,
	booktitle    = {Quantum Information and Computation for Chemistry},
	publisher    = {John Wiley \& Sons, Ltd},
	pages        = {295--354},
	doi          = {https://doi.org/10.1002/9781118742631.ch11},
	isbn         = 9781118742631,
	url          = {https://onlinelibrary.wiley.com/doi/abs/10.1002/9781118742631.ch11},
	chapter      = {}
}

@article{poulin2005stabilizer,
	title        = {Stabilizer Formalism for Operator Quantum Error Correction},
	author       = {Poulin, David},
	year         = 2005,
	month        = {Dec},
	journal      = {Phys. Rev. Lett.},
	publisher    = {American Physical Society},
	volume       = 95,
	pages        = 230504,
	doi          = {10.1103/PhysRevLett.95.230504},
	url          = {https://link.aps.org/doi/10.1103/PhysRevLett.95.230504},
	issue        = 23,
	numpages     = 4
}

@article{moudgalya2023numerical,
  title = {Numerical methods for detecting symmetries and commutant algebras},
  author = {Moudgalya, Sanjay and Motrunich, Olexei I.},
  journal = {Phys. Rev. B},
  volume = {107},
  issue = {22},
  pages = {224312},
  numpages = {19},
  year = {2023},
  month = {Jun},
  publisher = {American Physical Society},
  doi = {10.1103/PhysRevB.107.224312},
  url = {https://link.aps.org/doi/10.1103/PhysRevB.107.224312}
}

@article{lootens2021MPO,
	title        = {{Matrix product operator symmetries and intertwiners in string-nets with  domain walls}},
	author       = {Laurens Lootens and Jürgen Fuchs and Jutho Haegeman and Christoph Schweigert and Frank Verstraete},
	year         = 2021,
	journal      = {SciPost Phys.},
	publisher    = {SciPost},
	volume       = 10,
	pages        = 53,
	doi          = {10.21468/SciPostPhys.10.3.053},
	url          = {https://scipost.org/10.21468/SciPostPhys.10.3.053},
	issue        = 3
}

@article{cirac2021matrix,
	title        = {Matrix product states and projected entangled pair states: Concepts, symmetries, theorems},
	author       = {Cirac, J. Ignacio and P\'erez-Garc\'{\i}a, David and Schuch, Norbert and Verstraete, Frank},
	year         = 2021,
	month        = {Dec},
	journal      = {Rev. Mod. Phys.},
	publisher    = {American Physical Society},
	volume       = 93,
	pages        = {045003},
	doi          = {10.1103/RevModPhys.93.045003},
	url          = {https://link.aps.org/doi/10.1103/RevModPhys.93.045003},
	issue        = 4,
	numpages     = 65
}

@article{khemani2019int,
	title        = {Signatures of integrability in the dynamics of Rydberg-blockaded chains},
	author       = {Khemani, Vedika and Laumann, Chris R. and Chandran, Anushya},
	year         = 2019,
	month        = {Apr},
	journal      = {Physical Review B},
	publisher    = {American Physical Society},
	volume       = 99,
	pages        = 161101,
	doi          = {10.1103/PhysRevB.99.161101},
	url          = {https://link.aps.org/doi/10.1103/PhysRevB.99.161101},
	issue        = 16,
	numpages     = 6
}

@book{d2007introduction,
  title={Introduction to Quantum Control and Dynamics},
  author={D'Alessandro, D.},
  isbn={9781584888833},
  series={Chapman \& Hall/CRC Applied Mathematics \& Nonlinear Science},
  year={2007},
  publisher={CRC Press},
  edition={1},
  pages={81-82}
}

@misc{marvian2024abelian,
  title = {Theory of quantum circuits with Abelian symmetries},
  author = {Marvian, Iman},
  journal = {Phys. Rev. Res.},
  volume = {6},
  issue = {4},
  pages = {043292},
  numpages = {17},
  year = {2024},
  month = {Dec},
  publisher = {American Physical Society},
  doi = {10.1103/PhysRevResearch.6.043292},
  url = {https://link.aps.org/doi/10.1103/PhysRevResearch.6.043292}
}

@article{marvian2022rotationally,
  title = {Rotationally Invariant Circuits: Universality with the Exchange Interaction and Two Ancilla Qubits},
  author = {Marvian, Iman and Liu, Hanqing and Hulse, Austin},
  journal = {Phys. Rev. Lett.},
  volume = {132},
  issue = {13},
  pages = {130201},
  numpages = {7},
  year = {2024},
  month = {Mar},
  publisher = {American Physical Society},
  doi = {10.1103/PhysRevLett.132.130201},
  url = {https://link.aps.org/doi/10.1103/PhysRevLett.132.130201}
}

@Article{marvian2020locality,
author={Marvian, Iman},
title={Restrictions on realizable unitary operations imposed by symmetry and locality},
journal={Nature Physics},
year={2022},
month={Mar},
day={01},
volume={18},
number={3},
pages={283-289},
issn={1745-2481},
doi={10.1038/s41567-021-01464-0},
url={https://doi.org/10.1038/s41567-021-01464-0}
}

@ARTICLE{hearth2023unitary,
       author = {{Hearth}, Sumner N. and {Flynn}, Michael O. and {Chandran}, Anushya and {Laumann}, Chris R.},
        title = "{Unitary k-designs from random number-conserving quantum circuits}",
      journal = {arXiv e-prints},
         year = 2023,
        month = jun,
archivePrefix = {arXiv},
       eprint = {2306.01035},
 primaryClass = {cond-mat.stat-mech},
       adsurl = {https://ui.adsabs.harvard.edu/abs/2023arXiv230601035H}
}

@article{PhysRevB.103.L241118,
  title = {Page curve for fermionic Gaussian states},
  author = {Bianchi, Eugenio and Hackl, Lucas and Kieburg, Mario},
  journal = {Phys. Rev. B},
  volume = {103},
  issue = {24},
  pages = {L241118},
  numpages = {7},
  year = {2021},
  month = {Jun},
  publisher = {American Physical Society},
  doi = {10.1103/PhysRevB.103.L241118},
  url = {https://link.aps.org/doi/10.1103/PhysRevB.103.L241118}
}

@article{PhysRevB.104.214306,
  title = {Eigenstate capacity and Page curve in fermionic Gaussian states},
  author = {Bhattacharjee, Budhaditya and Nandy, Pratik and Pathak, Tanay},
  journal = {Phys. Rev. B},
  volume = {104},
  issue = {21},
  pages = {214306},
  numpages = {10},
  year = {2021},
  month = {Dec},
  publisher = {American Physical Society},
  doi = {10.1103/PhysRevB.104.214306},
  url = {https://link.aps.org/doi/10.1103/PhysRevB.104.214306}
}

@article{onishchik1990lie,
  title={Lie Groups and Algebraic Groups},
  author={Onishchik, Arkadij L and Vinberg, Ernest B},
  journal={Springer Series in Soviet Mathematics},
  year={1990},
  publisher={Springer Berlin Heidelberg}
}

@article{PRXQuantum.3.030201,
  title = {Volume-Law Entanglement Entropy of Typical Pure Quantum States},
  author = {Bianchi, Eugenio and Hackl, Lucas and Kieburg, Mario and Rigol, Marcos and Vidmar, Lev},
  journal = {PRX Quantum},
  volume = {3},
  issue = {3},
  pages = {030201},
  numpages = {77},
  year = {2022},
  month = {Jul},
  publisher = {American Physical Society},
  doi = {10.1103/PRXQuantum.3.030201},
  url = {https://link.aps.org/doi/10.1103/PRXQuantum.3.030201}
}

@article{Yu_2023,
   title={Free-fermion Page curve: Canonical typicality and dynamical emergence},
   volume={5},
   ISSN={2643-1564},
   url={http://dx.doi.org/10.1103/PhysRevResearch.5.013044},
   DOI={10.1103/physrevresearch.5.013044},
   number={1},
   journal={Physical Review Research},
   publisher={American Physical Society (APS)},
   author={Yu, Xie-Hang and Gong, Zongping and Cirac, J. Ignacio},
   year={2023},
   month=jan }

@article{Liu_2021,
   title={Entanglement Entropies of Equilibrated Pure States in Quantum Many-Body Systems and Gravity},
   volume={2},
   ISSN={2691-3399},
   url={http://dx.doi.org/10.1103/PRXQuantum.2.010344},
   DOI={10.1103/prxquantum.2.010344},
   number={1},
   journal={PRX Quantum},
   publisher={American Physical Society (APS)},
   author={Liu, Hong and Vardhan, Shreya},
   year={2021},
   month=mar }

@misc{li2024designslocalrandomquantum,
  title = {Designs from Local Random Quantum Circuits with $\mathrm{SU}(d)$ Symmetry},
  author = {Li, Zimu and Zheng, Han and Liu, Junyu and Jiang, Liang and Liu, Zi-Wen},
  journal = {PRX Quantum},
  volume = {5},
  issue = {4},
  pages = {040349},
  numpages = {55},
  year = {2024},
  month = {Dec},
  publisher = {American Physical Society},
  doi = {10.1103/PRXQuantum.5.040349},
  url = {https://link.aps.org/doi/10.1103/PRXQuantum.5.040349}
}

@article{Larocca2022diagnosingbarren,
  doi = {10.22331/q-2022-09-29-824},
  url = {https://doi.org/10.22331/q-2022-09-29-824},
  title = {Diagnosing {B}arren {P}lateaus with {T}ools from {Q}uantum {O}ptimal {C}ontrol},
  author = {Larocca, Martin and Czarnik, Piotr and Sharma, Kunal and Muraleedharan, Gopikrishnan and Coles, Patrick J. and Cerezo, M.},
  journal = {{Quantum}},
  issn = {2521-327X},
  publisher = {{Verein zur F{\"{o}}rderung des Open Access Publizierens in den Quantenwissenschaften}},
  volume = {6},
  pages = {824},
  month = sep,
  year = {2022}
}

@article{Fontana_2024,
   title={Characterizing barren plateaus in quantum ansätze with the adjoint representation},
   volume={15},
   ISSN={2041-1723},
   url={http://dx.doi.org/10.1038/s41467-024-49910-w},
   DOI={10.1038/s41467-024-49910-w},
   number={1},
   journal={Nature Communications},
   publisher={Springer Science and Business Media LLC},
   author={Fontana, Enrico and Herman, Dylan and Chakrabarti, Shouvanik and Kumar, Niraj and Yalovetzky, Romina and Heredge, Jamie and Sureshbabu, Shree Hari and Pistoia, Marco},
   year={2024},
   month=aug }

@article{Ragone_2024,
   title={A Lie algebraic theory of barren plateaus for deep parameterized quantum circuits},
   volume={15},
   ISSN={2041-1723},
   url={http://dx.doi.org/10.1038/s41467-024-49909-3},
   DOI={10.1038/s41467-024-49909-3},
   number={1},
   journal={Nature Communications},
   publisher={Springer Science and Business Media LLC},
   author={Ragone, Michael and Bakalov, Bojko N. and Sauvage, Frédéric and Kemper, Alexander F. and Ortiz Marrero, Carlos and Larocca, Martín and Cerezo, M.},
   year={2024},
   month=aug }

@book{basicalgebra,
	title        = {Basic Algebra: Groups, Rings and Fields},
	author       = {Cohn, P. M.},
	year         = 2003,
	publisher    = {Springer London}
}

@article{PhysRevB.112.064301,
  title = {Spacetime picture for entanglement generation in noisy fermion chains},
  author = {Swann, Tobias and Bernard, Denis and Nahum, Adam},
  journal = {Phys. Rev. B},
  volume = {112},
  issue = {6},
  pages = {064301},
  numpages = {23},
  year = {2025},
  month = {Aug},
  publisher = {American Physical Society},
  doi = {10.1103/PhysRevB.112.064301},
  url = {https://link.aps.org/doi/10.1103/PhysRevB.112.064301}
}

@misc{kovács2024operatorspacefragmentationperturbed,
      title={Operator space fragmentation in perturbed Floquet-Clifford circuits}, 
      author={Marcell D. Kovács and Christopher J. Turner and Lluis Masanes and Arijeet Pal},
      year={2024},
      eprint={2408.01545},
      archivePrefix={arXiv},
      primaryClass={quant-ph},
      url={https://arxiv.org/abs/2408.01545}, 
}

@misc{bittel2025completetheorycliffordcommutant,
      title={A complete theory of the Clifford commutant}, 
      author={Lennart Bittel and Jens Eisert and Lorenzo Leone and Antonio A. Mele and Salvatore F. E. Oliviero},
      year={2025},
      eprint={2504.12263},
      archivePrefix={arXiv},
      primaryClass={quant-ph},
      url={https://arxiv.org/abs/2504.12263}, 
}

@misc{hulse2021qudit,
      title={Qudit circuits with $\mathrm{SU}(d)$ symmetry: Locality imposes additional conservation laws}, 
      author={Austin Hulse and Hanqing Liu and Iman Marvian},
      year={2021},
      eprint={2105.12877},
      archivePrefix={arXiv},
      primaryClass={quant-ph}
}

@Article{sunderhauf2019quantum,
author={S{\"u}nderhauf, Christoph
and Piroli, Lorenzo
and Qi, Xiao-Liang
and Schuch, Norbert
and Cirac, J. Ignacio},
title={Quantum chaos in the Brownian SYK model with large finite N : OTOCs and tripartite information},
journal={Journal of High Energy Physics},
year={2019},
month={Nov},
day={07},
volume={2019},
number={11},
pages={38},
issn={1029-8479},
doi={10.1007/JHEP11(2019)038},
url={https://doi.org/10.1007/JHEP11(2019)038}
}

@article{ogunnaike2023unifying,
  title = {Unifying Emergent Hydrodynamics and Lindbladian Low-Energy Spectra across Symmetries, Constraints, and Long-Range Interactions},
  author = {Ogunnaike, Olumakinde and Feldmeier, Johannes and Lee, Jong Yeon},
  journal = {Phys. Rev. Lett.},
  volume = {131},
  issue = {22},
  pages = {220403},
  numpages = {7},
  year = {2023},
  month = {Nov},
  publisher = {American Physical Society},
  doi = {10.1103/PhysRevLett.131.220403},
  url = {https://link.aps.org/doi/10.1103/PhysRevLett.131.220403}
}

@article{Batista_2000,
   title={Quantum Phase Diagram of th $t$-$J_z$ Chain Model},
   volume={85},
   ISSN={1079-7114},
   url={http://dx.doi.org/10.1103/PhysRevLett.85.4755},
   DOI={10.1103/physrevlett.85.4755},
   number={22},
   journal={Physical Review Letters},
   publisher={American Physical Society (APS)},
   author={Batista, C. D. and Ortiz, G.},
   year={2000},
   month=nov, pages={4755–4758} }

@ARTICLE{swann2023spacetime,
       author = {{Swann}, Tobias and {Bernard}, Denis and {Nahum}, Adam},
        title = "{Spacetime picture for entanglement generation in noisy fermion chains}",
      journal = {arXiv e-prints},
         year = 2023,
        month = feb,
archivePrefix = {arXiv},
       eprint = {2302.12212},
 primaryClass = {cond-mat.stat-mech},
       adsurl = {https://ui.adsabs.harvard.edu/abs/2023arXiv230212212S},
}

\newpage

\onecolumngrid
\appendix

\section{Numerical Methods for the Super-Commutant}\label{sec:app-num}
An immediate consequence of the formulation of this problem in terms of commutant algebras, it is possible to use known numerical methods to compute the super-commutant algebra $\scomm$ or even the whole decomposition of Eq.~\eqref{eq:hilbdecend} of the Hilbert space, based on the techniques of Ref.~\cite{moudgalya2023numerical}.
We discuss two algorithms in detail.
The first explicitly builds the matrix block-decomposition of Eq.~\eqref{eq:matrixrep} for elements of $\sbond$, while the second uses an MPS representation to find an exact basis of $\scomm$ as the ground states of a one-dimensional frustration-free Hamiltonian.
Here we simply outline the main ideas and limitations of the two techniques.
More details on the techniques themselves can be found in Ref.~\cite{moudgalya2023numerical}.
\subsection{Matrix Diagonalization method}
Consider a superoperator $\mc K\in\sbond$, which can therefore be obtained as a linear combination of products of $\ad{h_\alpha}$ superoperators, for $h_\alpha\in\gen$.
Due to the decomposition in Eq.~\eqref{eq:matrixrep}, there is a basis of the operator Hilbert space $\hend$ such that for any such $\mc K$, we have
\begin{equation}
    \mc K=\bigoplus_{\widehat\lambda}(\mc M_{\widehat\lambda}(\mc K)\ot\1_{d_{\widehat\lambda}}).
\end{equation}
The matrices $\mc M_{\widehat\lambda}(\mc K)$ are some $D_{\widehat\lambda}\times D_{\widehat\lambda}$ complex matrices, and they are in general distinct for distinct values of $\widehat\lambda$.
Therefore we can expect that for a ``random choice'' of $\mc K$, all eigenvalues of all the matrices $\mc M_{\widehat\lambda}(\mc K)$ will be distinct, even across different values of the index $\widehat\lambda$.
Under this assumption, if one diagonalizes $\mc K=\mc W\mc D\mc W\+$, then by putting \textit{any}\footnote{A random operator $\mc K'\in\sbond$ can be sampled by performing a random linear combination of the generators $\{\ad{h_\alpha}\}$ and their products.} $\mc K'\in\sbond$ into this eigenbasis, the resulting matrix $\mc W\+\mc K'\mc W$ will be block diagonal along the irreps labeled by $\widehat\lambda$, therefore allowing them to be identified~\cite{moudgalya2023numerical}.
Within each block labelled by $\widehat\lambda$, the values of $D_{\widehat\lambda}$ and $d_{\widehat\lambda}$ can simply be obtained by counting the degeneracies in the spectrum: $D_{\widehat\lambda}$ is just the number of distinct eigenvalues in the block, while $d_{\widehat\lambda}$ is equal the degeneracy of the each eigenvalue.
A basis in which $\mc P_{\widehat\lambda}\mc K'\mc P_{\widehat\lambda}=\mc M_{\widehat\lambda}(\mc K')\ot\1_{d_{\widehat\lambda}}$ (where $\mc P_{\widehat\lambda}$ is the projector superoperator onto the $\widehat\lambda$-irrep, as in Eq.~\eqref{eq:project-to-tensor}) can also be found numerically with some additional work as discussed in \cite{moudgalya2023numerical}, hence providing an explicit construction of the full super-commutant $\scomm$.
A numerical speedup can be obtained by exploiting the relationship between the commutant $\comm$ and the super-commutant $\scomm$.
If the physical Hilbert space is decomposed as in Eq.~\eqref{eq:fund-th}, then we have
\begin{gather}
    h_\alpha=\bigoplus_\lambda M_\lambda(h_\alpha)\ot\1_{d_\lambda},\;\;{\mathds{1} = \bigoplus_\lambda \mathds{1}_{D_{\lambda}} \otimes \mathds{1}_{d_\lambda}} \nn \\
    \implies \ad{h_\alpha}{\defn h_\alpha \otimes \mathds{1} - \mathds{1} \otimes h_\alpha^T} = \bigoplus_{\lambda,\lambda'} (M_\lambda(h_\alpha)\ot\1_{D_{\lambda'}}-\1_{D_{\lambda}}\ot M_{\lambda'}(h_\alpha)^T)\ot\1_{d_\lambda\cdot d_{\lambda'}}.
\end{gather}
The indices $(\lambda,\lambda')$ are exactly the ones discussed in Sec.~\ref{subsec:semiuniversal} and App.~\ref{sec:app-scommt}, and this decomposition of the superoperator $\ad{h_\alpha}$ is equivalent to the presence of superoperators of type (i) in Eq.~\eqref{eq:forsure} in the super-commutant.
Using this decomposition, we can first identified the $\lambda$ blocks by applying the aforementioned procedure for the $h_\alpha$s, and then for each $(\lambda,\lambda')$ block apply the procedure again to square matrices of size $D_\lambda D_{\lambda'}$, which can be much smaller than the original $\mrm{dim}(\mc H)^2$, especially for non-abelian commutants, where one can have $d_\lambda>1$\footnote{From Eq.~\eqref{eq:fund-th} we see that $\mrm{dim}(\mc H)=\sum_\lambda D_\lambda d_\lambda$.}. The computational complexity of diagonalizing each block is $\mc O(D_\lambda^3 D_{\lambda'}^3)$ instead of the original $\mc O(\mrm{dim}(\mc H)^6)$ for the full system.
In the most common cases of group-like symmetry, this speed-up is usually polynomial with system size, with a computation time that scales like $\mc O(D_\lambda^6)$ for the subspace with the largest dimension $D_\lambda$.
Possible degeneracies arising between different $(\lambda,\lambda')$ sectors can be identified by looking for degenerate eigenvalues across the different sectors.
If we only wish to study the block-decomposition within the set of symmetric operators $\bond$ (which is the only relevant one for universality since the DLA $\dlie$ is fully contained in this subspace), we can further reduce the computational costs by analyzing the blocks where $\lambda=\lambda'$, which are precisely the ones acting on symmetric operators.
Once the block-decomposition is completed, one can construct the DLA without any significant computational costs by studying the overlaps between the blocks and the generators $\{\oket{h_\alpha}\}$, as described in Sec.~\ref{sec:methodology}.
One can do slightly better if one interested in finding the dimension of the super-commutant $\dim(\scomm)$, which can be a good indicator for the existence of an interesting DLA.
A slight improvement in the numerical efficiency may be provided by using $\bad{h_\alpha}\defeq h_\alpha\ot\1+\1\ot h_\alpha$ instead of the usual $\ad{h_\alpha}\defeq h_\alpha \otimes \mathds{1} - \mathds{1} \otimes h_\alpha^T$.
As discussed through the many-copy interpretation (see Sec.~\ref{sec:manycopy} and App.~\ref{sec:app-manycopy}), the commutants of these two families of superoperators are related through partial transposition, and have therefore the same dimension.
While in general the commutant of $\{\ad{h_\alpha}\}$ (i.e. the super-commutant) always contains at least the identity superoperator $\1\ot\1$ and the scar projector $\oketbra{\1}{\1}$, the commutant of $\{\bad{h_\alpha}\}$ instead always contains the identity superoperator $\1\ot\1$ and the swap superoperator $S$, which acts as $S(\ket a\ot\ket b)=\ket b\ot\ket a$ on the two-copy Hilbert space.
This $\mb Z_2$ swap symmetry of the $\bad{h_\alpha}$ superoperators can be used to further reduce the cost of performing the simultaneous block-diagonalization, since it splits each matrix into two blocks of approximately the same dimension.
Note however that while the commutants of these two families of superoperators have the same dimensions, they do not necessarily have the same algebraic structure, and therefore more detailed information about the numbers $\{(D_\lambda,d_\lambda)\}$ and about the DLA is directly accessible only from the original adjoint superoperators $\{\ad{h_\alpha}\}$, and not from $\{\bad{h_\alpha}\}$.


\begin{figure*}[t]\vspace{0pt}
\includegraphics[width=.9\textwidth]{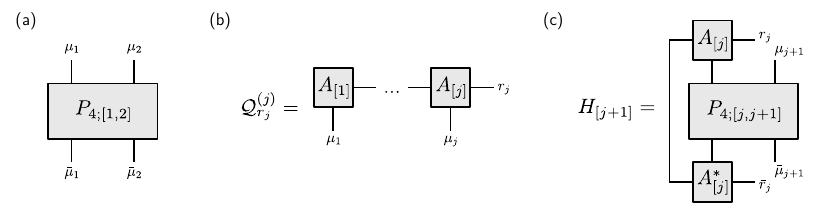}
\caption{Tensor network diagrams for the MPS method for finding the super-commutant. (a) Tensor representation of the first local Hamiltonian term $P_{4;[1,2]}$. (b) Canonical MPS representation of an orthonormal basis for the super-commutant of a subsystem of length $L=j$, as in Eq.~\eqref{eq:mps-basisscomm}. (c) Tensor diagram for the Hamiltonian whose ground state has to be computed at each step, as in Eq.~\eqref{eq:mps-hamiltonian}.}\label{fig:mps}
\end{figure*}

\subsection{Matrix Product State method}\label{app:mpsmethod}
A different method based on Matrix Product States (MPS) can be adopted if the generators $\gen$ act locally on a one-dimensional qudit chain $\mc H=(\mb C^d)^{\ot L}$. This method is based on the idea that symmetries can be understood as ground states of frustration-free Hamiltonians on a doubled Hilbert space (see App.~\ref{sec:brown})~\cite{moudgalya2023numerical, moudgalya2023symmetries}.
For simplicity of notation, we will assume that each generator is $2$-local, but the ideas generalize to longer-range but strictly local generators~\cite{moudgalya2023numerical}.
We start with the following gates with OBC:
\begin{equation}
    \gen=\{\1_d^{\ot (j-1)}\ot h^{(\alpha)}_{j,j+1}\ot \1_d^{\ot (L-j-1)}\}^{\alpha=1,...,N_j}_{j=1,...,L-1}.
\end{equation}
Note that with this definition, the operators $h^{(\alpha)}_{j,j+1}$ only act on the two adjacent qudits at positions $j$ and $j+1$.
As explained in Sec.~\ref{sec:physicalimplications} and App.~\ref{sec:brown} (in particular Eq.~\eqref{eq:supercommutant-hamiltonian}), the elements $\mc Q\in\scomm$ are exactly the ground states of the four-copy Hamiltonian:
\begin{equation}
    P_4=\sum_{j=1}^{L-1} \1_d^{\ot (j-1)}\ot P_{4;[j,j+1]}\ot \1_d^{\ot (L-j-1)}\qquad P_{4;[j,j+1]}\defeq\sum_{\alpha=1}^{N_j}\left(\ad{h^{(\alpha)}_{j,j+1}}\ot\1-\1\ot\ad{h^{(\alpha)}_{j,j+1}}^T\right)^2.
\end{equation}
Due to the positivity of each term in the sum, the ground states of $P_4$ are also ground states of each individual term, hence resulting in a ``frustration-free'' ground state problem \cite{moudgalya2023numerical}, which can be solved exactly and efficiently through the use of an MPS representation of the superoperators $\mc Q\in\scomm$ to be found.

We start by finding an orthonormal basis $\{\mc Q_{r_2}^{(2)}\}_{r_2=1,...,\chi_2}$ for the ground state space of $P_{4;[1,2]}$.
The number $\chi_2$ is equal to the dimension of the super-commutant for the subsystem of length $L=2$ composed of the first two qudits.
By interpreting $P_{4;[1,2]}$ as a $d^4\times d^4\times d^4\times d^4$ tensor with indices $(\mu_1,\mu_2,\bar\mu_1,\bar\mu_2)$ (see Fig.~\ref{fig:mps}a) we can write this basis as a canonical MPS by performing a singular value decomposition between the index $\mu_1$ and the indices $(\mu_2,r_2)$:
\begin{equation}
    (\mc Q_{r_2}^{(2)})^{\mu_1,\mu_2} = \sum_{r_1=1}^{\chi_1} (A_{[1]})^{\mu_1,r_1} (A_{[2]})^{r_1,\mu_2,r_2}.
\end{equation}
Note that here $(r_1, r_2)$ and $(\mu_1, \mu_2)$ are interpreted as the auxiliary and physical indices of a (non-translation-invariant) MPS.
At each step, we build an orthonormal basis $\{\mc Q_{r_{j+1}}^{(j+1)}\}$ for the super-commutant for the subsystem of length $L=j+1$ composed of the qudits at position $1,...,j+1$.
This is done iteratively by finding a tensor $(A_{[j+1]})^{r_j,\mu_{j+1},r_{j+1}}$ such that (see Fig.~\ref{fig:mps}b)
\begin{equation}\label{eq:mps-basisscomm}
    (\mc Q_{r_{j+1}}^{(j+1)})^{\mu_1,\mu_2,...,\mu_{j+1}} = \sum_{r_1=1}^{\chi_1}\sum_{r_2=1}^{\chi_2}...\sum_{r_{j+1}=1}^{\chi_{j+1}} (A_{[1]})^{\mu_1,r_1} (A_{[2]})^{r_1,\mu_2,r_2}...(A_{[j+1]})^{r_j,\mu_{j+1},r_{j+1}}.
\end{equation}
The equation to find $A_{[j+1]}$ only depends on $A_{[j]}$ and $P_{4;[j,j+1]}$~\cite{moudgalya2023numerical}: if we join together the indices $(r_j,\mu_{j+1})$, the tensor $(A_{[j+1]})^{(r_j,\mu_{j+1}),r_{j+1}}$ is any orthonormal basis (indexed by $r_{j+1}$) for the null space of the matrix (see Fig.~\ref{fig:mps}c)
\begin{equation}\label{eq:mps-hamiltonian}
    (H_{[j+1]})^{(r_{j},\mu_{j+1}),(\bar r_{j},\bar \mu_{j+1})} = \sum_{r_{j-1}=1}^{\chi_{j-1}} (A_{[j]})^{r_{j-1},\mu_j,r_j}(A_{[j]}^*)^{r_{j-1},\bar\mu_j,\bar r_j}(P_{4;[j,j+1]})^{\mu_{j},\mu_{j+1},\bar\mu_{j},\bar\mu_{j+1}}.
\end{equation}
The square matrix $H_{[j+1]}$ has dimension $\chi_j d^4$, leading to a computational complexity of $\mc O(\chi_j^3 d^{12})$ for diagonalization in general.
Therefore this method ultimately scales with the dimension of the super-commutant: it will be efficient if $\mrm{dim}(\scomm)$ remains small or grows slowly as a function of the system size $L$, making it a valuable tool in these cases.
Due to the steep $\mc O(d^{12})$ power law, this method is not suitable for systems where the dimension of the local degrees of freedom is not small.
Some optimizations to this algorithm are possible by exploiting some on-site symmetries of $P_4$ in the MPS representation, which might lead to block-diagonalizations of the tensors as in standard implementations of MPS~\cite{cirac2021matrix}.
For example, the Hamiltonian $P_4$ in general possesses an on-site $\mb Z_2\times\mb Z_2$ symmetry, associated to swapping the replicas $\tp{\vdotc{0}{$1$}}\leftrightarrow\tp{\vdotc{0}{$4$}}$ and $\tp{\vdotc{0}{$2$}}\leftrightarrow\tp{\vdotc{0}{$3$}}$ among the four replicas of the original Hilbert space involved in $P_4$.
Furthermore, if the operator-level commutant $\comm$ is non-trivial, then $P_4$ will possess an independent $\comm$ symmetry on each of the four replicas.
\section{Relation to the ``Types of Non-Universality'' in Earlier Works}\label{app:types}
In light of our two-fold classification of the types of universality, it is useful to compare it with the classification of the possible obstructions to universality introduced in Ref.~\cite{marvian2024abelian}, numbered from I to IV.
In our language, this classification can be restated in the following way:
\begin{enumerate}[Type I]
    \item constraints limit the controllability of relative phases between symmetry sectors -- these appear to be the most common, and are the only ones present in the case of weak non-universality. They correspond to vanishing overlaps between scars (i.e. elements of the center $\cent$) and generators $\gen$ (see Sec.~\ref{subsec:semiuniversal} for more details).
    \item constraints arise whenever the set of $G$-symmetric gates $\gen$ that we start with possesses a larger symmetry algebra $\comm$ than the symmetry group $G$. Hence, the $G$-symmetry blocks in the physical Hilbert space $\mc H$ become reducible under the action of the DLA.
    From the point of view presented in Sec.~\ref{sec:superoperatoralgebra}, this is not natural to consider, since we are mostly interested in studying the gate set $\gen$ (and derive its commutant $\comm$ later) rather than study the symmetry group $G$. 
    \item constraints limit the set of unitaries that can be generated \textit{within} each symmetry sector. At the algebra level this means that the set of realizable $M_\lambda(K)$ in Eq.~\eqref{eq:matrixrep} does not contain all $D_\lambda\times D_\lambda$ matrices.\footnote{To be precise, only the absence of a \textit{traceless} matrix would constitute a type III constraint. The absence of the identity matrix $\1_{D_\lambda}$, which corresponds to $\Pi_\lambda$ through Eq.~\eqref{eq:matrixrepcenter}, is instead associated with type I constraints.} 
    This happens when in the $\scommt$ block decomposition of $\hend$, some of the sectors split further into blocks, due to the presence of additional projectors in $\scomm$ beyond the ones of Eq.~\eqref{eq:forsure}.
    \item constraints corresponds to the correlated action of the DLA on different symmetry sectors: referring to Eq.~\eqref{eq:matrixrep}, this means that for some values of $\lambda$, the matrices $M_\lambda(K)$ are not independent one from the other;\footnote{If the interdependence is a linear constraint on the traces of the $M_\lambda(K)$ matrices, then this is a type I constraint instead.} this happens when some irreps in $\hend$ belonging to different $\comm\ot\comm^T$ symmetry sectors become degenerate under the action of $\scomm$ due to the presence of additional superoperators beyond the ones of Eq.~\eqref{eq:forsure}.
\end{enumerate}
Note that although they can often appear together, these constraints are all independent from each other.\footnote{For example the system studied in Sec.~\ref{sec:zx} has type III constraints but no type IV, while vice versa matchgate circuits studied in Secs.~\ref{sec:mgu1} and \ref{sec:mgz2} have type IV constraints but no type III for $L<4$.}
In our language, the rather technical type III and IV constraints both simply correspond to the presence of additional superoperators in the super-commutant, beyond the ones described in Eq.~\eqref{eq:forsure} above, and their presence is therefore associated with strong non-universality.
The exhaustiveness of classifications is also clear from the superoperator picture: if $\scomm=\scommt$, a system can only be non-universal in the presence of type I constraints, and if $\scomm\supsetneq\scommt$, the additional superoperator symmetries can be responsible for type III and IV constraints, if they affect the decomposition of the operator Hilbert space $\hend$ within the symmetric sector $\bond$ (see Lemma~\ref{lem4}).
A system that only displays type I constraints is referred to as ``semi-universal'' in Ref.~\cite{marvian2024abelian}, and although all weakly non-universal systems are semi-universal, in some rarer cases strongly non-universal systems can be semi-universal as well (e.g. matchgates acting on two qubits, see Footnote \ref{ft:matchgate2}).
Semi-universality can be understood in the superoperator language through the equivalence proven in Theorem \ref{theorem1}.
\section{Details of Non-Universality for Particular Sets of Gates}\label{sec:gates}
\subsubsection{Non-symmetric and $\mathbb Z_2$-symmetric Circuits}

In the main text we have presented in full how the commutant framework is applied for studying the non-universality of $U(1)$-symmetric circuits. 
Before moving forwards to the more general cases, we can briefly analyze some simpler examples.
In the case of a sufficiently large and generic set of local non-symmetric gates, we will have $\comm=\mrm{span}\{\1\}$ and $\scomm=\llangle\{\oketbra{\1}{\1}\}\rrangle$.
The Krylov subspace decomposition of the space of operators is therefore simply $\{\mrm{span}\{\oket\1\},\oket\1^\perp\}$, showing that in this case all traceless operators can in general be generated (corresponding to the special unitary group $SU(\dim\mc H)$).
If $\gen$ includes an operator with non-zero trace, then the set of gates will be even more universal, in that the unitary $U(\dim\mH)$ can be generated.
For example the set of 2-local operators
\begin{equation}
    \gen_\mrm{univ}=\{\1,X_j,Z_j,X_jX_{j+1},Z_jZ_{j+1}\}_{j=1}^L
\end{equation}
is found to generate all unitaries, and the set without $\1$ would also generate all unitaries except the overall global phase.
In this work, we will not be interested in the overall global phase, hence the traceless set of gates are also universal in our definition.
The simplest symmetry that can be implemented is the $\mb Z_2$ parity operator $P=\prod_{j=1}^L Z_j$, which is satisfied by the set of generators
\begin{equation}
    \gen_\mrm{\mb Z_2}=\{\1,Z_j,X_jX_{j+1},Z_jZ_{j+1}\}_{j=1}^L.
\end{equation}
In this case $\comm=\mrm{span}\{\1,P\}$ and $\scomm=\llangle\{P\ot\1,\1\ot P^T,\oketbra{\1}{\1}\}\rrangle$.
Scar operators form a two-dimensional subspace $\mrm{span}\{\oket\1,\oket P\}$, while the rest of their operators are split into four Krylov subspaces according to the $\mb Z_2\times\mb Z_2$ symmetry generated by left and right multiplication by the parity operator $\{P\ot\1,\1\ot P^T\}$.
This situation is analogous to the $U(1)$ case, with the only difference that, due to the small size of the group, there is only one additional scar $\oket P$, which never overlaps with $k$-local $\mb Z_2$-symmetric generators if $k<L$.
Hence the co-dimension of the controllable manifold $\unit$ within $\unitt$ is $1$.
\subsubsection{$t$-$J_z$ Model}
We now discuss the case of the $t$-$J_z$ model.
An orthogonal basis for the center $\cent$ of the gate set of Eq.~\eqref{eq:tjzgates}, which is also its commutant $\comm$, is described in detail in App.~B of Ref.~\cite{moudgalya2021hilbert}.
In terms of the operators of Eq.~(\ref{eq:tjzgates}), if we define $O_j\defeq\1_j-Z_j^2$, then this basis for $\cent$ is obtained by considering the following ``word operators'' for all strings $s=(s_1,...,s_n)$ of $n=0,...,L$ numbers $s_l\in\{1,2\}$:
\begin{equation}
    W_n(s)=\sum_{j_1<...<j_n} Z_{j_1}^{s_1}\cdot...\cdot Z_{j_k}^{s_n}\cdot\left(\prod_{j\notin\{j_1,...,j_n\}}O_j\right).
\label{eq:tjzwords}
\end{equation}
Let us now compute the necessary overlaps between these elements, which are the scars, and the generators $T_{j,j+1}$, $Z_j^2$, $Z_j$, and $Z_j Z_{j+1}$ of Eq.~(\ref{eq:tjzgates}).
We can then see that:
\begin{itemize}
    \item For all positions $j$ and strings $s$ we have $\obraket{T_{j,j+1}}{W_n(s)}=0$, since
    \begin{equation}
        \tr(T_{j,j+1}A_j\ot B_{j+1})=\sum_{\sigma=\uparrow,\downarrow}(\langle\sigma 0|A\ot B|0\sigma\rangle+\langle 0\sigma |A\ot B|\sigma 0\rangle),
    \end{equation}
    which is zero if $A,B$ are chosen among the on-site operators $\{O,Z,Z^2\}$, since they are diagonal in the $\{\ket\downarrow,\ket 0,\ket\uparrow\}$ basis.
    Note that these are the only operators that can appear in Eq.~(\ref{eq:tjzwords}). Therefore these generators will not contribute to the co-dimension formula Eq.~\eqref{eq:dimformula}.
    \item  $\obraket{Z^2_j}{W_n(s)}\neq 0$ only for the strings $s=(2,\cdots,2)$ for $n=1,\cdots,L$, since $Z^2\cdot O=0$ and $\tr(Z)=\tr(Z^3)=0$.
    Since the overlaps $\obraket{Z^2_j}{W_n((2,...,2))}=\binom{L-1}{n-1}2^n$ are independent of the position $j$, the projection of the gates $\{Z^2_j\}_{j=1,...,L}$ onto the center is one-dimensional, since the overlap matrix of Eq.~\eqref{eq:dimformulark} will have rank $1$.
    \item For the remaining gates we can restrict to the subspace of words of length $n=L$, to show that they all have a distinct projection onto the center. Indeed, since the number of word operators of length $n=L$ is greater than the number of $Z_j$ and $Z_jZ_{j+1}$ generators, we can show that the overlap matrix of Eq.~\eqref{eq:dimformulark} has maximum rank (i.e. $2L-1$, given the number of generators under consideration) by only considering the submatrix associated to this set of word operators.
    $\obraket{Z_j}{W_L(s)}\neq 0$ only for the word with $s_j=1$ and $s_{j'\neq j}=2$, while $\obraket{Z_jZ_{j+1}}{W_L(s)}\neq 0$  only for the word with $s_j=s_{j+1}=1$ and $s_{j'\notin\{j,j+1\}}=2$.
\end{itemize}
From this we can conclude that $\mrm{dim}(\dlie\cap\cent)\geq 0+1+(2L-1)=2L$.
\subsubsection{Translation Invariant Gates}
We consider the set of generators Eq.~\eqref{eq:gates-transinv} and its commutant $\comm=\mrm{span}(\{\oket\1,\oket T,\oket{T^2},...,\oket{T^{L-1}}\})$.
In order to apply the formula of Eq.~\eqref{eq:dimformulark} we calculate the overlaps between the generators and the given basis for the (Abelian) commutant:
\begin{itemize}
    \item Since all generators are traceless, the overlap with the identity operator $\oket \1$ is zero.
    \item By performing traces using the computational basis (i.e. simultaneous eigenstates of all the $Z_j$) for the qubit chain, one can see that $\oket{\sum_{j=1}^L X_j}$ does not overlap with any of the basis elements for the commutant: \begin{equation}
        \tr(\sum_{j=1}^L X_j T^n)=\sum_{j=1}^L\sum_{\{\sigma_j\}} \bra{\{\sigma_j\}}X_j T^n\ket{\{\sigma_j\}}=0
    \end{equation}
    since $X_j T^n\ket{\{\sigma_j\}}$ is orthogonal to $\ket{\{\sigma_j\}}$ (the two states have different spin parity under $P = \prod_j{Z_j}$).
    The same argument also works for $\oket{\sum_{j=1}^L X_j Z_{j+1}}$, and by analogy for all $\oket{\sum_{j=1}^L S^\alpha_j}$ and all $\oket{\sum_{j=1}^L S^\alpha_jS^\beta_{j+1}}$ with $\alpha\neq\beta$ (where we can evaluate them in different choices of the computational basis).
    \item 
    The operator $\oket{\sum_{j=1}^L Z_j Z_{j+1}}$, on the other hand, has a non-zero overlap with the translation operators $\oket T$ and $\oket{T^{L-1}}=\oket{T\+}$, and is orthogonal to the rest, as we show now.
    By performing the trace using the computational basis, the only states that contribute to $\tr(\sum_{j=1}^L Z_j Z_{j+1}T^{n})$ are $n$-periodic spin configurations. 
    For $n=1$ and $n=L-1$, these are the two fully polarized states $\ket{\uparrow...\uparrow}$ and  $\ket{\downarrow...\downarrow}$, which give:
    \begin{equation}
        \tr(\sum_{j=1}^L Z_jZ_{j+1}T^{n})_{|n\in\{1,L-1\}}=
        \sum_{j=1}^L(\bra{\uparrow...\uparrow}Z_jZ_{j+1}\ket{\uparrow...\uparrow}+\bra{\downarrow...\downarrow}Z_jZ_{j+1}\ket{\downarrow...\downarrow})=2L.
    \end{equation}
    For other values of $n$, the trace is zero since the number of $n$-periodic states for which each $Z_j Z_{j+1}$ is $+1$ is equal to the number of $n$-periodic states for which it is $-1$ (obtained by performing a global spin flip in the computational basis).
    The same reasoning also applies to $\oket{\sum_{j=1}^L X_jX_{j+1}}$ and $\oket{\sum_{j=1}^L Y_jY_{j+1}}$ by selecting an appropriate choice of computational spin basis.
    Therefore these three generators all share the same projection onto the center $\cent$.
\end{itemize}
From all this we can conclude that $\mrm{dim}(\dlie\cap\cent)=1$.
\subsubsection{\texorpdfstring{$\mathbb Z_2$}{} Matchgate Circuits}

\paragraph{Asymptotic value of OTOCs} 
Let us calculate Eq.~\eqref{eq:otoc-projector} for the case of matchgate circuits \eqref{eq:matchgate-gens} and $A=B=Z_j=-i\gamma_{2j-1}\gamma_{2j}$ for any $j$ on the lattice.
By following the conventions of Eq.~\eqref{eq:dimer-trivial-comm} we see that the projectors onto the subspaces spanned by Majorana strings $\oket a$ of fixed length $|a|=0,...,2L$ can be written in normalized form as:
\begin{equation}
    \ket{\mc P_n}_4 = \frac{1}{2^L\cdot\sqrt{\binom{2L}{n}}}\sum_{|a|=n} \tket{\varcdimercr{0.8}{0}{$a$}{white}\varcdimerc{0.4}{1.2}{$a\+$}{white}}_{\!4}.
\end{equation}
These are the only superoperators that contribute to  Eq.~\eqref{eq:otoc-projector}, since the rest\footnote{ These correspond to two types of superoperators. (i) The off-diagonal superoperators responsible for the degeneracy between the subspaces $\hend_n$ and $\hend_{2L-n}$ containing Majorana strings of length $n$ and $2L-n$ respectively, which have the form $(P\ot\1)\mc P_n$. (ii) The superoperator $\mc P_{L+}-\mc P_{L-}$, which is orthogonal to the projector $\mc P_L=\mc P_{L+}+\mc P_{L-}$, where $\mc P_{L\pm}$ are the projectors onto the parity $\pm 1$ sectors of the space spanned by all Majorana strings of length $L$.} are orthogonal to $\Big|\tp{\varcdimercr{0.4}{0}{$Z_j$}{white}\varcdimerc{0.8}{1.2}{$Z_j$}{white}}\Big\rangle_{\!4}$.
Note that $\oket{a\+}=(-)^{\lfloor n/2\rfloor}\oket a$ for $|a|=n$.
We therefore obtain the expression:
\begin{multline}
    \overline{C_{AB}}(\infty)= \sum_{n=0}^{2L}\prescript{}{4\!}{\tbra{\varcdimercr{0.4}{0}{$Z_j$}{white}\varcdimerc{0.8}{1.2}{$Z_j$}{white}}}\frac{\ket{\mc P_n}_4\!\bra{\mc P_n}}{2^L}
    \tket{\varcdimercr{0.8}{0}{$Z_j$}{white}\varcdimerc{0.4}{1.2}{$Z_j$}{white}}_{\!4}=\\=\frac{1}{4^L}\sum_{n=0}^{2L}\frac{1}{\binom{2L}{n}}\cdot\left(\sum_{|a|=n}\tr\left(\left(\gamma_{2j-1}\gamma_{2j}a\right)^2\right)\right)\cdot\left(\sum_{|a|=n}\left(\tr\left(\gamma_{2j-1}\gamma_{2j}a\right)\right)^2\right).
\end{multline}
In the sum, $\tr\left(\gamma_{2j-1}\gamma_{2j}a\right)=-2^L$ when $a=\gamma_{2j-1}\gamma_{2j}$ and zero otherwise; this implies that the only contribution to the sum comes from the superoperator $\ket{\mc P_2}_4$. 
For $n=2$ then, the term $\tr\left(\left(\gamma_{2j-1}\gamma_{2j}a\right)^2\right)$ is equal to either $+2^L$ or $-2^L$ for all strings $a$ of length 2; the result is $+2^L$ if the string $a$ contains either both or none of the operators $\gamma_{2j-1}$ and $\gamma_{2j}$, while it is $-2^L$ otherwise. 
Since the number of strings that strings that contain either $\gamma_{2j-1}$ or $\gamma_{2j}$ (but not both) is $4(L-1)$, we obtain the result of Eq.~\eqref{eq:how-to-otoc-mg}. The calculation to perform is very similar when $A$ and $B$ are both Majorana strings of any length.

\paragraph{Free-fermion Page curve} 
We now calculate Eq.~\eqref{eq:renyisuperoperatorformula} for the case of matchgate circuits \eqref{eq:matchgate-gens} and a Gaussian initial state. 
We can choose $\ket\psi$ to be the vacuum of the chosen Majorana basis, i.e. the polarized spin-down state $\ket{\downarrow\downarrow...\downarrow}$. 
We start by noting that in the $U\ot U^*\ot U\ot U^*$ convention that we used for the domain walls of Eq.~\eqref{eq:domainwalls}, for two Majorana strings $a$ and $b$:
\begin{equation}
    \oketbra{a}{b} = \tket{\varcdimercr{1.2}{0}{$a$}{white}\varcdimerc{0.4}{0.8}{$b\+$}{white}}_{\!4}.
\end{equation}
In the calculation, the only conserved superoperators that will contribute are the (partial transposed versions of) projectors onto the subspaces spanned by Majorana strings $\oket a$ of fixed length $|a|=0,...,2L$; these can be othonormalized as follows:
\begin{equation}
    \sket{\widetilde{\mc P}_n}_4 = \frac{1}{2^L\cdot\sqrt{\binom{2L}{n}}}\sum_{|a|=n} \tket{\varcdimercr{1.2}{0}{$a$}{white}\varcdimerc{0.4}{0.8}{$a\+$}{white}}_{\!4};
\end{equation}
the rest are orthogonal to $\ket{A\!:\!\bar A}_4$.
Note that $\oket{a\+}=(-)^{\lfloor n/2\rfloor}\oket a$ for $|a|=n$. We therefore obtain the expression:
\begin{equation}
    \overline{\langle\tr(\rho_A^2(\infty))}=\sum_{n=0}^{2L}\prescript{}{4\!}{\sbraket{A\!:\!\bar A}{\widetilde{\mc P}_n}_4}\!\sbraket{\widetilde{\mc P}_n}{\psi}^{\ot 4}=\frac{1}{4^L}\sum_{n=0}^{2L}\frac{1}{\binom{2L}{n}}\cdot\left(\sum_{|a|=n}\tr_{\bar A}\left((\tr_Aa)^2\right)\right)\cdot\left(\sum_{|a|=n}\langle\psi|a|\psi\rangle^2\right).
\end{equation}
In the sum, $\langle\psi|a|\psi\rangle^2$ is non-zero only when $n$ is even and the Majorana string corresponds to a product of $Z_j$ operators in the spin language; since $\gamma_{2j-1}\gamma_{2j}=iZ_j$, we have $\langle\psi|a|\psi\rangle^2=(-)^k$ if $n=2k$ and $a$ is composed of pairs of consecutive Majorana operators of the form $\gamma_{2j-1}\gamma_{2j}$ (there are $\binom{L}{k}$ strings of this form). 
Furthermore $\tr_Aa$ is non-zero only for the Majorana strings that have support on $\bar A$; there are $\binom{2(L-\ell)}{2k}$ such operators (when $k\leq L-\ell$), for whom $\tr_{\bar A}\left((\tr_Aa)^2\right)=(-)^{k}2^{L+\ell}$. 
By putting everything together we find the result of Eq.~\eqref{eq:how-to-page-mg}.
\subsubsection{\texorpdfstring{$U(1)$}{} Matchgate Circuits}
In the main text we have analyzed the block decomposition of the set of $U(1)$-symmetric operators under the action of $U(1)$-preserving matchgate circuits. We noted how this family of circuits preserves two separate non-trivial $U(1)$ charges, namely the number of creation and of annihilation operators $\mc N_{c/c\+}$, whereas $\mb Z_2$-symmetric circuits only preserve their sum $\mc N_\gamma$, defined in Eq.~\eqref{eq:majoranaN}.
Here we describe the operator Hilbert space in more detail.
Due to $\mc N_{c/c\+}$ conservation, we can split the operator Hilbert space into invariant subspaces $\hend_{n,\bar n}$ with $n,\bar n\in\{0,...,L\}$, composed of operators of the form $:\!(c\+)^{\bar n}(c)^n\!:$, where $:\!\bullet\!:$ denotes the normal ordering. 
When $n=\bar n$, the operators $\hend_{n,\bar n}$ are symmetric and belong to $\sbond$.
Subspaces with $n\neq\bar n$ contain operators that do not preserve particle number, and due to the symmetry between creation and annihilation operators, the structure of each $\hend_{n,\bar n}$ sector will be identical to that of its counterpart $\hend_{\bar n,n}$.
Due to number conservation, these sectors get further split in subspaces labeled by $\Gamma\in\{0,...,\mrm{min}\{n,\bar n,L-n,L-\bar n\}\}$. These subspaces contain operators of the form\footnote{Operators of the form \eqref{eq:basismgU1subspaces} for some value of $\Gamma$ might in some cases be obtained as linear combinations of operators with higher values of $\Gamma$. The precise constructions of the $(n,\bar n,\Gamma)$ proceeds iteratively, increasing the value $\Gamma$ starting from $0$ and always orthogonalizing new subspace with respect to the previous ones, in order to prevent overlap.}
\begin{equation}
    :\!N_{(\mrm{min}\{n,\bar n\}-\Gamma)} c\+_{j_1}...c\+_{j_{\bar n-\mrm{min}\{n,\bar n\}+\Gamma}}c_{j_1'}...c_{j_{n-\mrm{min}\{n,\bar n\}+\Gamma}'}\!:\qquad \text{where}\quad N_X=\sum_{j_1<...<j_X} :\!c\+_{j_1}c_{j_1}...c\+_{j_X}c_{j_X}\!:\label{eq:basismgU1subspaces}
\end{equation}
where $N_X$ are orthogonalized elements of the commutant $\commx{U(1)}$ (see Eq.~\eqref{eq:paulibasisu1}).
If we call $\Delta=n-\bar n$, then all sectors with the same value of $\Delta$ and $\Gamma$ will be degenerate in the operator Hilbert space decomposition of Eq.~\eqref{eq:hilbdecend}.
Notice the presence of two one-dimensional invariant subspaces, namely $\hend_{L,0}$ and $\hend_{0,L}$, containing $\prod_{j=1}^L c_j\+$ and $\prod_{j=1}^L c_j$ respectively, which do not correspond to elements of the commutant, since they are not annihilated by the action of $\ad{c_j\+c_j+h.c.}$.
These subspaces are a feature that also appears in general weakly non-universal $U(1)$-symmteric circuits, and in spin notation they correspond to the operators $\ketbra{\uparrow...\uparrow}{\downarrow...\downarrow}$ and $\ketbra{\downarrow...\downarrow}{\uparrow...\uparrow}$; these are the $\hend_{0,L}$ and $\hend_{L,0}$ subspaces in the operator Hilbert space decomposition described in Eq.~\eqref{eq:nmdec}.
\section{Mathematical Details on the Super-commutant algebras}
In this appendix, we collect various mathematical details of the super-bond and super-commutant algebras that are necessary for discussions in the main text.
\subsection{Full Decomposition of the Operator Hilbert Space for $\scommt$}\label{sec:app-scommt}
In Sec.~\ref{subsec:semiuniversal} we have discussed the structure of the operator Hilbert space decomposition from Eq.~\eqref{eq:hilbdecend}, in the case when the super-commutant is minimal $\scomm=\scommt$.
For the convenience for the remaining results in this section, we can schematically describe this structure as follows.
{Note that in the following, we will start with the physical Hilbert space decomposition of Eq.~(\ref{eq:fund-th}), and derive the operator Hilbert space decomposition from that. 
Hence the $\{\lambda, D_\lambda, d_\lambda\}$ in the following correspond to those in Eq.~(\ref{eq:fund-th}), whereas the $\widehat{\lambda}$'s in Eq.~(\ref{eq:hilbdecend}) can sometimes be labelled by ordered pairs of the former $\lambda$'s.}
\begin{itemize}
    \item For each ordered pair $(\lambda,\lambda')$ with $\lambda\neq\lambda'$ there is an irrep {$\widehat{\lambda} = (\lambda, \lambda')$} composed of $d_\lambda d_{\lambda'}$ degenerate $D_{\lambda}D_{\lambda'}$-dimensional invariant subspaces. These lie outside of the bond algebra $\bond$.
    \item For each $\lambda$ with $D_{\lambda}>1$ there is an irrep {$\widehat{\lambda} = (\lambda, \lambda)$} containing all operators of the form $\oket{M_\lambda\ot N_\lambda}$ and their linear combinations. 
    The {subspaces corresponding to different linearly independent operators $N_\lambda$ are all degenerate, and hence each irrep is} composed of $d_\lambda^2$ degenerate $(D_{\lambda}^2-1)$-dimensional invariant subspaces. 
    For each $\lambda$, only one such subspace belongs to the bond algebra $\bond$, the one associated to $N_\lambda=\1_{d_\lambda}$.
    \item There is one additional irrep {$\widehat{\lambda}$} containing all operators in the commutant $\oket{\1_{D_\lambda}\ot N_\lambda}\in\comm$ (this has been called $\hend_\mrm{scar}$ in the text).
    This irrep is composed of one-dimensional Krylov subspaces, and its intersection with the bond algebra $\bond$ is the center $\cent$.
\end{itemize}
This structure can be seen to mirror the structure conserved superoperators in Eq.~\eqref{eq:forsure} belonging to $\scommt$, where the $Q_1\ot Q_2^T$ superoperators are responsible for the ${\widehat{\lambda} =} (\lambda_1,\lambda_2)$ irreps, and the $\oketbra{Q_1}{Q_2}$ superoperators are responsible for the scar irrep.
\subsection{Lemmas on the Super-Commutant}\label{sec:app-math}

Here we collect a few simple lemmas proving facts about the super-commutant.
Although most of these facts are probably known from earlier literature, for convenience we state them and prove them within the language and framework of this paper.
%
%
\begin{lemma}[Independence from the particular choice of generators] 
If there are two sets of generators $\mc G$ and $\mc G'$ such that $\dlie=\dliex{\gen'}$, then $\scomm=\scommx{\gen'}$.
In particular the block-decomposition of the operator Hilbert space used in Sec.~\ref{sec:superoperatoralgebra}, is independent on the choice of generators for a given Lie algebra.\label{lem1} 
\end{lemma}
\begin{proof}
This is basically due to the fact that $\ad{(\bullet)}$ is a representation of Lie algebras, and in particular
\begin{equation}
    \ad{[h_{\alpha_1},h_{\alpha_2}]}=[\ad{h_{\alpha_1}},\ad{h_{\alpha_2}}]\defeq \ad{h_{\alpha_1}}\ad{h_{\alpha_2}}-\ad{h_{\alpha_2}}\ad{h_{\alpha_1}}.
\end{equation}
Since the expression on the RHS of the equation only contains addition and multiplication of operators, $\ad{[h_{\alpha_1},h_{\alpha_2}]}$ belongs to $\sbond$ whenever $h_{\alpha_{1}},h_{\alpha_{2}}\in \gen$. 
Since the DLA is generated through repeated commutators and linear combinations, we find that $\ad H\in\sbond$ for each $H\in\dlie$ (and similarly for $\dliex{\gen'}$).
{Since we know that $\dlie = \dliex{\gen'}$,} it follows immediately that  $\sbondx{\gen}=\sbondx{\gen'}$ and  hence $\scommx{\gen}=\scommx{\gen'}$, concluding the proof. 
Indeed this also shows that $\sbond$ is the adjoint representation of the ``universal enveloping algebra'' of $\dlie$.
\end{proof}


\begin{lemma}[Minimal super-commutant and maximal super-bond algebra] 
Here we prove the second equality in Eq.~\eqref{eq:sbont}, i.e. that $\sbondt=\mrm{comm}(\scommt)${, which shows that the maximal super-bond algebra is the commutant of the minimal super-commutant defined in Eq.~(\ref{eq:supercommtrivial}).}\label{lem2}
\end{lemma}
\begin{proof}
The fact that $\sbondt\subseteq\mrm{comm}(\scommt)$ is easy to show.
Note that any $K\in\bond=\gent$ commutes with {all} elements of the commutant $\comm$.
{Hence any such $\ad K$ commutes with all the generators of $\scommt$ shown in Eq.~(\ref{eq:supercommtrivial}).
In particular it commutes with all elements in $\comm \otimes \comm^T$, as well as $\oketbra{\mathds{1}}{\mathds{1}}$.}
Showing the converse $\mrm{comm}(\scommt)\subseteq\sbondt$ is equivalent to {showing} $\mrm{comm}(\sbondt)\subseteq\scommt$, i.e. $\sbondt$ has no symmetries other than the minimal ones {shown in Eq.~(\ref{eq:supercommtrivial}).}
We can prove this by showing that $\sbondt$ acts irreducibly on the invariant subspaces identified by $\scommt$ through Eq.~\eqref{eq:hilbdecend}.
This follows from the fact that the adjoint action of the algebra of traceless matrices $\mf{sl}_{\mb C}(d)$ on itself is irreducible but to show this explicitly, we use the matrix representation of Eqs.~\eqref{eq:matbasis} and \eqref{eq:matrixrep}; in particular we will use the following basis for {the space of} operators
\begin{equation}
    \oket{\lambda_1,\lambda_2,\alpha_1,\alpha_2,\gamma_1,\gamma_2}\defeq \left(\ket{\alpha_1}_{\lambda_1}\ot\ket{\gamma_1}_{\lambda_1}\right)\left(\prescript{}{\lambda_2}{\bra{\alpha_2}}\ot\prescript{}{\lambda_2}{\bra{\gamma_2}}\right).
\label{eq:opsubspace}
\end{equation}
According to the decomposition given by $\scommt$ in App.~\ref{sec:app-scommt}:
\begin{itemize}
    \item 
    {We first focus on invariant subspaces of the operator Hilbert space labelled by $\widehat{\lambda} = (\lambda_1, \lambda_2)$ for} $\lambda_1\neq\lambda_2$.
    Then {we show that} $\oket{\lambda_1,\lambda_2,\alpha_1,\alpha_2,\gamma_1,\gamma_2}$ should be connected to all other  $\oket{\lambda_1,\lambda_2,\alpha_1',\alpha_2',\gamma_1,\gamma_2}$.
    This can be achieved by applying the {super-}operator $\ad{K_2}\cdot\ad{K_1}\in\sbondt$, where we choose $K_1,K_2\in\bond$ such that in this tensored basis, $K_1$ and $K_2$ have the forms $K_1 = \ket{\alpha_1'}\prescript{}{\lambda_1}{\bra{\alpha_1}} \otimes \mathds{1}_{d_{\lambda_1}}$ and $K_2 = \ket{\alpha_2}\prescript{}{\lambda_2}{\bra{\alpha_2'}} \otimes \mathds{1}_{d_{\lambda_2}}$.
    {It is then easy to verify that for $K_1$ and $K_2$ of these forms, we obtain $\mL_{K_2} \cdot \mL_{K_1} \oket{\lambda_1, \lambda_2, \alpha_1, \alpha_2, \gamma_1, \gamma_2} = -\oket{\lambda_1, \lambda_2, \alpha_1', \alpha_2', \gamma_1, \gamma_2}$.}
    Hence for a given pair ${\widehat{\lambda} =} (\lambda_1,\lambda_2)$, the space spanned by the operators {of the form shown in Eq.~(\ref{eq:opsubspace}) for fixed $\lambda_1, \lambda_2, \gamma_1, \gamma_2$} is an irreducible representation.
    {For a given $(\lambda_1, \lambda_2)$, the corresponding operators of the form of Eq.~(\ref{eq:opsubspace}) [which we refer to as $\oket{O}_{\lambda_1, \lambda_2}$ for brevity] are} characterized by the relations $\ad{\Pi_{\lambda_1}}\oket O_{\lambda_1,\lambda_2}=\oket O_{\lambda_1,\lambda_2}$, $\ad{\Pi_{\lambda_2}}\oket O_{\lambda_1,\lambda_2}=-\oket O_{\lambda_1,\lambda_2}$, and $\ad{\Pi_{\lambda_3}}\oket O_{\lambda_1,\lambda_2}=0$ for $\lambda_3\notin\{\lambda_1,\lambda_2\}$.
    This shows that two irreps {$\widehat{\lambda} = (\lambda_1, \lambda_2)$ and $\widehat{\lambda'} = (\lambda_1', \lambda_2')$ with $\widehat{\lambda} \neq \widehat{\lambda'}$} are never degenerate {under the action of $\sbond$ and hence are distinct irreps}.
    With this, we have shown that the block decomposition of the operator Hilbert space according to $\scommt$ (as discussed in Sec.~\ref{sec:app-scommt}) is finer or equal to that of $\sbond$ within the $\lambda_1\neq\lambda_2$ sectors.
    \item 
    When $\lambda_1=\lambda_2=\lambda$, then the action of $\sbond$ on the $(\alpha_1,\alpha_2)$ indices of a given $(\gamma_1,\gamma_2)$ sector is exactly the adjoint action of the algebra of all complex matrices $M_\lambda\in\mf{gl}_{\mb C}(D_\lambda)$; this decomposes the algebra into two irreducible representations, which are the set of traceless matrices $\mf{sl}_{\mb C}(D_\lambda)$ and the identity component $\mrm{span}(\1_{D_\lambda})$ \cite{fulton2013representation}.
    This is exactly the decomposition obtained from the super-commutant $\scommt$, with one scar emerging from each degenerate subspace of the $(\lambda,\lambda)$ sector.
\end{itemize}
{In all, this completes the proof that the $\mrm{comm}(\sbondt)\subseteq\scommt$, and hence also that of $\mrm{comm}(\sbondt) = \scommt$.}
\end{proof}

\begin{definition}
We define $\mc P_{\bond}\in\scommt\subseteq\scomm$ to be the superoperator that projects the operator Hilbert space $\hend$ onto the bond algebra $\bond$:
\begin{equation}
    \mc P_{\bond}\oket{O} = \begin{cases}
        \oket O,\qquad&\mathrm{if\ }\oket O\in\bond,\\
        0,\qquad&\mathrm{if\ }\oket O\in\bond^\perp.\\
    \end{cases}
\end{equation}
\end{definition}
As mentioned in Sec.~\ref{sec:methodology}, the DLA $\dlie$ can be determined by only analyzing the block decomposition of the space of symmetric operators $\bond$, since $\dlie\subseteq\bond$.
The decomposition of this space is simply dictated by the superoperator conserved quantities $\mc Q$ that map $\bond$ onto itself; the set of such conserved quantities can be written as $\mc P_{\bond}\scomm\mc P_{\bond}\subseteq\scomm$. 
This subset of superoperator conserved quantities is the one responsible for the presence or absence of semi-universality, as explained in the following theorem.
\begin{theorem}\label{theorem1}
    A set of generators $\gen$ is semi-universal as defined in Ref.~\cite{marvian2024abelian}, or in other words it satisfies any of Eqs.~(\ref{eq:semiunivformula}-\ref{eq:dimformulark}) (see Sec.~\ref{subsec:semiuniversal}), if and only if $\mc P_{\bond}\scomm\mc P_{\bond}=\mc P_{\bond}\scommt\mc P_{\bond}$.
\end{theorem}
\begin{proof}
This is a direct consequence of the two Lemmas \ref{lem3} and \ref{lem4}, which we show below.
\end{proof}

\begin{lemma}[The generators overlap with all symmetry sectors up to central elements] Given a gate set $\gen$ with commutant $\comm$, for every symmetry sector $\mc H_\lambda=\mc H_\lambda^{\bond}\ot\mc H_\lambda^{\comm}$ with $D_\lambda>1$ in the decomposition of Eq.~\eqref{eq:fund-th} the generators overlap with at least one operator of the form $\oket{M_\lambda\ot\1_{d_\lambda}}$ with $\tr(M_\lambda)=0$ (following the matrix notation of Eq.~\eqref{eq:matrixrep}).
In particular $\gen$ is semi-universal as defined in Ref.~\cite{marvian2024abelian} (see Sec.~\ref{subsec:semiuniversal}) if it is weakly non-universal (i.e. $\scomm=\scommt$), or more generally if $\mc P_{\bond}\scomm\mc P_{\bond}=\mc P_{\bond}\scommt\mc P_{\bond}$.\label{lem3}
\end{lemma}
\begin{proof}
We start by noticing that the superoperator that projects onto the $\comm\ot\comm^T$-symmetric subspace composed of all operators of the form $\oket{M_\lambda\ot\1_{d_\lambda}}$, takes the simple form of $\Pi_\lambda\ot\Pi_\lambda^T$, where $\{\Pi_\lambda\}$ are the projectors onto the irreps of the physical Hilbert space $\mc H$ of Eq.~\eqref{eq:fund-th}.
We then prove the statement of the lemma by contradiction. Let us assume that for all generators $h_\alpha$ lies in the one-dimensional subspace generated by the central element $\oket{\Pi_\lambda}=\oket{\1_{D_\lambda}\ot\1_{d_\lambda}}$:
\begin{equation}
    \oket{\Pi_\lambda h_\alpha\Pi_\lambda}=\Pi_\lambda\ot\Pi_\lambda^T\oket{h_\alpha}\propto\oket{\Pi_\lambda}.
\end{equation}
Since $\Pi_\lambda\in\cent=\comm\cap\bond$, $h_\alpha$ commutes with $\Pi_\lambda$, and therefore
\begin{equation}
    \Pi_\lambda\ot\Pi_\lambda^T\oket{h_\alpha h_\beta}=\oket{\Pi_\lambda h_\alpha h_\beta\Pi_\lambda}=\oket{(\Pi_\lambda h_\alpha \Pi_\lambda)(\Pi_\lambda h_\beta\Pi_\lambda)}\propto\oket{\Pi_\lambda},
\end{equation}
hence for all $K\in\bond$ we get $\Pi_\lambda\ot\Pi_\lambda^T\oket{K}\propto\oket{\Pi_\lambda}$. But if this was the case, then for every $K\in\bond$ and for every state $\ket{\psi}\in\mc H_\lambda$, one would have
\begin{equation}
    K\ket\psi\propto\ket\psi,
\end{equation}
hence generating one-dimensional Krylov subspaces, which violate the hypothesis that $D_\lambda>1$.

As discussed in Sec.~\ref{subsec:semiuniversal}, when $\scomm=\scommt$ the bond algebra $\bond$ is decomposed into subspaces labeled by $\{\lambda\}$ containing all operators $\oket{M_\lambda\ot\1_{d_\lambda}}$ such that $\tr(M_\lambda)=0$, plus a subspace (the center $\cent$) containing all projectors $\oket{\Pi_\lambda}$. Therefore, through the techniques of Sec.~\ref{sec:methodology}, this result implies that the orthogonal complement of $\dlie$ in $\bond$ is a subset of the center $\cent$, which corresponds to semi-universality. 
Since the block decomposition of $\bond$ is exclusively determined by the conserved quantities in $\mc P_{\bond}\scomm\mc P_{\bond}\subseteq\scomm$, this statement remains also true if we relax the assumption $\scomm = \scommt$ to $\mc P_{\bond}\scomm\mc P_{\bond}=\mc P_{\bond}\scommt\mc P_{\bond}$.
\end{proof}


\newcommand{\scommtp}{\scommx{{\langle\!\langle}\gen'{\rangle\!\rangle}}}
\newcommand{\scommp}{\scommx{\gen'}}

\begin{lemma}[Conserved superoperators and strong non-universality] If $\mc P_{\bond}\scomm\mc P_{\bond}\supsetneq\mc P_{\bond}\scommt\mc P_{\bond}$, then at least one symmetric operator $\oket H\in\bond$ exists that is neither contained in the center $\cent$ nor in the DLA $\dlie$. In particular if $\gen$ is semi-universal as defined in Ref.~\cite{marvian2024abelian} (see Sec.~\ref{subsec:semiuniversal}) then $\mc P_{\bond}\scomm\mc P_{\bond}=\mc P_{\bond}\scommt\mc P_{\bond}$.\label{lem4}
\end{lemma}
\begin{proof}
Let us consider a superoperator $\mc Q$ that belongs to $\mc P_{\bond}\scomm\mc P_{\bond}$ but not to $\mc P_{\bond}\scommt\mc P_{\bond}$; then by definition of the super-commutants, for some $H'\in\bond$ one must have
\begin{equation}
    [\mc Q,\ad{H'}]\neq 0,\label{eq:Hcondition}
\end{equation}
from which it follows that $\oket{H'}\notin\dlie$ (according to Lemma \ref{lem1}). Let us now decompose $H'=H+Z$, with $Z\in\cent$ and $H\in\bondx{\gen}\cap\centx{\gen}^\perp$. Since $\mc Q\in \mc P_{\bond}\scomm\mc P_{\bond}$, we can write $\mc Q=\mc P_{\bond}\mc Q\mc P_{\bond}$, and therefore
\begin{equation}
    [\mc Q,\ad{Z}]=(\mc P_{\bond}\mc Q\mc P_{\bond})\ad{Z}-\ad{Z}(\mc P_{\bond}\mc Q\mc P_{\bond})=0,
\end{equation}
because $\ad Z$ annihilates all operators in $\bond$. Therefore Eq.~\eqref{eq:Hcondition} implies
\begin{equation}
    0\neq [\mc Q,\ad{H'}]=[\mc Q,\ad{H}]+[\mc Q,\ad{Z}]=[\mc Q,\ad{H}],
\end{equation}
which means that $H\notin\dlie$, thus proving the statement of the lemma.
The absence from $\dlie$ of a non-central operator $\oket H$ by definition indicates the failure of semi-universality, as described in Sec.~\ref{subsec:semiuniversal}, since Eq.~\eqref{eq:semiunivformula} is not satisfied.
\end{proof}

\begin{lemma}
    The projection of the DLA $\dlie$ onto the center $\cent$ is linearly generated by the projection of its generators $\gen$:
    \begin{equation}
        \Pi_{\cent}(\dlie)=\mrm{span}(\Pi_{\cent}(\gen)).
    \end{equation}\label{lem6}
    Note that as a general feature of the decomposition of Eq.~\eqref{eq:hilbdecend}, $\Pi_{\cent}(\dlie)=\dlie\cap\cent$.
\end{lemma}
\begin{proof}
    According to equation \eqref{eq:DLAgeneration} we obtain $\dlie$ by acting with $\ad{h_\alpha}$ on $\gen=\{\oket{h_{\alpha'}}\}$ and performing linear combinations; but the superoperators $\ad{h_\alpha}$ by definition annihilate all elements in the commutant, and therefore the center. Therefore they produce operators that do not overlap with the center: $\forall \oket Z\in\cent:\obra{Z}\ad{h_\alpha}\oket{K}=\left[\obra{Z}\ad{h_\alpha}\+\right]\oket{K}=0$. Hence, only the identity component $\1\in\sbond$ will be able to contribute to the projection onto $\cent$ when generating the DLA:
    \begin{equation}
        \Pi_{\cent}(\dlie)=\Pi_{\cent}(\sbond\cdot \mrm{span}(\gen))=\Pi_{\cent}(\mrm{span}(\gen))=\mrm{span}(\Pi_{\cent}(\gen)),
    \end{equation}
    where the last equality is a general property of linear maps.
\end{proof}

\subsection{Algebraic structure of Many-copy Conserved Quantities}\label{sec:app-manycopy}

According to the definitions in Sec.~\ref{sec:superoperatoralgebra}, it is clear that $\scomm$ and $\sbond$ are {finite-dimensional} von Neumann algebras, i.e. vector spaces containing the identity $\1$, closed under matrix multiplication, and under hermitian adjoint.
While this is the end of the story for general commutants, the form of the $\ad{h_\alpha}$ superoperators endows the super-commutants $\scomm$ with additional structure.
To understand how, it is convenient to use the many-copy notation from Sec.~\ref{sec:manycopy}.
First of all, $\scomm$ can be seen to be invariant under the exchange of the indices $\tp{\vdotc{0}{$1$}}\leftrightarrow\tp{\vdotc{0}{$4$}}$ or $\tp{\vdotc{0}{$2$}}\leftrightarrow\tp{\vdotc{0}{$3$}}$.
Indeed the condition that $\ket{\mc Q}_4$ is in $\scomm$ is equivalent to satisfying for any $U\in\unit$
\begin{equation}\label{eq:conservedsuperop}
    (U\ot U^*\ot U^*\ot U)\ket{\mc Q}_4=\ket{\mc Q}_4,
\end{equation}
and this condition will also be satisfied by their swapped versions (the superoperators in Eq.~\eqref{eq:dimer-trivial-comm} are for example related in this way).
Furthermore, the closure of $\scomm$ under matrix multiplication can be understood as follows: given two superoperators $\ket{\mc Q_1}_4$ and $\ket{\mc Q_2}_4$ that satisfy Eq.~\eqref{eq:conservedsuperop}, then contracting indices on which $U$ acts with indices on which $U^*$ acts gives superoperators that satisfy the same condition (since $\sum_j u_{ij}u_{kj}^*=\delta_{ik}$).
Graphically the contractions associated to superoperator products can be represented as:
\begin{equation}
\mc Q_1\cdot \mc Q_2=\ \ \tp{
\draw[thick,midarrow,dotted] (+.8,+.3) -- (+.4,-.3);
\draw[thick,midarrow,dotted] (-.4,-.3) -- (+.0,+.3);
\filldraw[thick,fill=white] (-.4,+.3) circle (0.06);
\filldraw[thick,fill=white] (-.4,-.3) circle (0.06);
\filldraw[thick,fill=white] (+.0,+.3) circle (0.06) node[black] at (+.0,+.5) {\footnotesize $*$};
\filldraw[thick,fill=white] (+.0,-.3) circle (0.06) node[black] at (+.0,-.6) {\footnotesize $*$};
\filldraw[thick,fill=white] (+.4,+.3) circle (0.06) node[black] at (+.4,+.5) {\footnotesize $*$};
\filldraw[thick,fill=white] (+.4,-.3) circle (0.06) node[black] at (+.4,-.6) {\footnotesize $*$};
\filldraw[thick,fill=white] (+.8,+.3) circle (0.06);
\filldraw[thick,fill=white] (+.8,-.3) circle (0.06);
\node[black] at (+1.2,+.25) {\footnotesize $\mc Q_1$};
\node[black] at (+1.2,-.35) {\footnotesize $\mc Q_2$};
}
\end{equation}
where the dotted lines indicate contractions, and the asterisks are used to remind which indices correspond to $U$ or to $U^*$.
Through the permutational symmetry observed above Eq.~(\ref{eq:conservedsuperop}), one can apply permutations on one of the operators and obtain from this product many other superoperators, which correspond to contracting \textit{any} $\tp{\vdotc{0}{}\vdotc{0.4}{$*$}}$ pair from the first superoperator with \textit{any} similar pair form the second.
But more can be done beyond this. One can for example contract the $\tp{\vdotc{0}{$*$}\vdotc{0.4}{$*$}}$ pair of one operator to the $\tp{\vdotc{0}{}\vdotc{0.4}{}}$ pair of the other; these types of products are the ones naturally performed in the commutant of a two-copy system (cf. Eqs.~\eqref{eq:twocopy} and \eqref{eq:unitary-patterns}), which was shown to be in a one-to-one correspondence with the super-commutant algebra.
Finally contractions may also be performed within a given conserved superoperator, thus obtaining an element of the operator-level commutant $\comm$, since the resulting object will be invariant under $U\ot U^*$.
More generally, by contracting along $q$ indices a $k$-copy conserved quantity with a $k'$-copy conserved quantity, one gets an object in the $(k+k'-q)$-copy commutant, with the structure of the higher copy commutants being consistent with the structure of the lower copy ones.
This gives rise to more structure in the super-commutant algebra -- not only should it be closed under operator multiplication, but it should also be closed under these more generalized operations.
This structure can be exploited while solving for the super-commutant of a set of gates.
For example, consider the conserved superoperator $\mc Q_1=\sum_{n=1}^{2L}\oketbra{\gamma_n}{\gamma_n}$ for the matchgate system of Eq.~\eqref{eq:matchgate-gens}, whose presence in the super-commutant $\scomm$ indicates that Majorana fermion $\gamma_n$ operators are mapped to other Majorana fermion operators when evolved through matchgates; since $\mc Q_1$ is a projector, under a normal product $\mc Q_1\cdot \mc Q_1=\mc Q_1$, but through the correspondence with the two-copy commutant discussed above, we get $\widetilde{\mc Q}_1\cdot \widetilde{\mc Q}_1=\widetilde{\mc Q}_2$ where
\begin{equation}
    \mc Q_2 = \sum_{n,m=1}^{2L}\oketbra{\gamma_n\gamma_m}{\gamma_m\gamma_n} = 2L\oketbra{\1}{\1} -2\sum_{n<m}\oketbra{\gamma_n\gamma_m}{\gamma_n\gamma_m}.
\end{equation}
Since $\oketbra{\1}{\1}$ is always present in $\scomm$, the presence of $\mc Q_2$ in the super-commutant implies that Majorana strings of length $2$ also preserve their length; by raising $\widetilde{\mc Q}_1$ to higher powers, one can ultimately show that the length of Majorana strings is preserved for all lengths, using no assumptions beyond the fact that Majorana strings of length one preserve their length.


\subsection{Proof of weak non-universality in the $U(1)$-symmetric case}\label{sec:app-nplusoneproof}
In this section we provide a proof of Eq.~\eqref{eq:U1supercomm}, which states that the set of $2$-local $U(1)$-symmetric gates of Eq.~\eqref{eq:U1generators} is weakly non-universal.
To show that a given set of gates is weakly non-universal, one must prove that all superoperator symmetries are linear combinations of superoperators of type (i) or (ii) in Eq.~\eqref{eq:forsure}.
These two sets of superoperators are generally non-orthogonal, and can actually have a non-zero intersection,\footnote{We will associate the letter $e$ to superoperators {that can be naturally interpreted to be} of type (i), the letter $\eta$ to superoperators {that can be naturally interpreted to be} of type (ii), and the letter $s$ to superoperators {that can be naturally interpreted to be either} of {the} types.} i.e. there can be superoperators that can be interpreted to be of both types.
We start the section by proving Lemma~\ref{lemmaU1proof} in order to deal with this issue. 
Otherwise, the main ingredient in the proof of weak non-universality is showing that in a system of length $L$, a superoperator symmetry cannot be composed by ``gluing together'' superoperators of different types on smaller sub-chains.
In particular, in the $U(1)$-symmetric case we show this by considering the $Z_jZ_{j+1}$ generator and by exploiting permutational invariance of the super-commutant.
\begin{lemma}\label{lemmaU1proof}
    {Consider} a finite-dimensional Hilbert space $V$ and three orthogonal subspaces $E,N,S\subseteq V$,  and a unitary group action $g\in G$ on $V$ such that:
    \begin{equation}
        \forall g: g\cdot E\perp N.
    \end{equation}
    Then if $\ket v\in E\oplus N\oplus S$ is invariant under the group action (i.e. $g\ket v=\ket v$), then it can be decomposed as
    \begin{equation}\label{eq:thesislemmaD7}
        \ket v = \ket{\mc E'} + \ket{\mc N'}, \quad\mathrm{where}\quad \ket{\mc E'}\in E\oplus S, \quad \ket{\mc N'}\in N\oplus S,
    \end{equation}
    and where both $\ket{\mc E'}$ and $\ket{\mc N'}$ are invariant {under $g$, i.e., $g\ket{\mc E'} = \ket{\mc E'}$ and $g\ket{\mc N'} = \ket{\mc N'}$}.
\end{lemma}
\begin{proof}
Let us define the orbit spaces
\begin{equation}
    E'=\mrm{span}\{g\ket{\mc E}\ \mrm{s.t.}\ g\in G, \ket{\mc E}\in E\},\;\;N'=\mrm{span}\{g\ket{\mc N}\ \mrm{s.t.}\ g\in G, \ket{\mc N}\in N\}.
\end{equation}
Since $\1\in G$, we have that $E\subseteq E'$ and $N\subseteq N'$.
$E'$ and $N'$ are orthogonal, because:
\begin{equation}
    \bigg(c'_1\bra{\mc N_1}{{g'_1}\+}+c'_2\bra{\mc N_2}{{g'_2}\+}\bigg)\bigg(c_1g_1\ket{\mc E_1}+c_2 g_2\ket{\mc E_2}\bigg) = \sum_{i,j=1,2} c'_ic_j \bra{\mc N_i}({{g'_i}\+}g_j)\ket{\mc E_j} = 0,
\end{equation}
{where we have used that the action of the element ${g'_i}\+ g_j$ (an element of $G$) on any vector in $E$ is orthogonal to $N$.}
We define $S'$ to be the orthogonal complement of $E'\oplus N'$ in $V$, so that $V=E'\oplus N'\oplus S'$.
Since these three subspaces in this decomposition are invariant under the group action of $G$, any invariant vector $\ket v\in V$ can be decomposed into:
\begin{equation}
    \ket v = \ket{\mc E'} + \ket{\mc N'}+ \ket{\mc S'}, \quad\mathrm{where}\quad \ket{\mc E'}\in E', \quad \ket{\mc N'}\in N',\quad \ket{\mc S'}\in S',
\end{equation}
and where $\ket{\mc E'}$, $\ket{\mc N'}$, and $\ket{\mc S'}$ are all invariant under the action of $G$.
If we suppose now that $\ket v\in E\oplus N\oplus S$, we can verify that in the decomposition above
\begin{equation}
    \ket{\mc E'}\in E\oplus S,\quad \ket{\mc N'}\in N\oplus S,\quad \ket{\mc S'}\in S,
\end{equation}
since $E'\perp N$, $N'\perp E$, and $S'\perp E\oplus N$. This proves the expression in Eq.~\eqref{eq:thesislemmaD7}, since the vector $\ket{\mc S'}$ can be absorbed into either $\ket{\mc E'}$ or $\ket{\mc N'}$.
\end{proof}
\begin{theorem}
    The set of generators in Eq.~\eqref{eq:U1generators} is weakly non-universal, i.e. its super-commutant is given by Eq.~\eqref{eq:U1supercomm}.
\end{theorem}
\begin{proof}
    We consider a spin-$\frac{1}{2}$ chain of length $L$, and its 4-copy superoperator Hilbert space $V=\mc H\ot\mc H^*\ot\mc H^*\ot\mc H$.
    On each site of the 4-copy chain, we can build a 16-dimensional spin basis $\{\ket{\sigma_1\sigma_2\sigma_3\sigma_4}_4\}_{\sigma_i=\uparrow,\downarrow}$. The super-commutant corresponds to the set of states which are invariant under the action of $U\ot U^*\ot U^*\ot U$ for any $U=\exp(i\theta H)$ where $H\in\dliex{\gen_{U(1)}}$ (see Lemma~\ref{lem1}).
    To show that the super-commutant is minimal we start by considering some particular choices of $U$ to constrain the space of superoperator symmetries:
    \begin{itemize}
        \item The local basis can be restricted by considering the generator $Z_j$ for a given site $j$.
        The set of one-site superoperators which is invariant under the action generated by choosing $U=U_{Z_j}(\theta)=e^{i\theta Z_j}$ is spanned by the orthonormal basis
        \begin{equation}\label{eq:u1scommbasis}
            \ket{s_0}_4=\ket{\downarrow\downarrow\downarrow\downarrow}_4\quad
            \ket{s_1}_4=\ket{\uparrow\uparrow\uparrow\uparrow}_4\quad
            \ket{e_0}_4=\ket{\downarrow\downarrow\uparrow\uparrow}_4\quad
            \ket{e_1}_4=\ket{\uparrow\uparrow\downarrow\downarrow}_4\quad
            \ket{\eta_0}_4=\ket{\downarrow\uparrow\downarrow\uparrow}_4\quad
            \ket{\eta_1}_4=\ket{\uparrow\downarrow\uparrow\downarrow}_4.
        \end{equation}
        Here $e$-states are invariant under $U_{Z_j}(\theta_1)\ot U_{Z_j}(\theta_1)^*\ot U_{Z_j}(\theta_2)^*\ot U_{Z_j}(\theta_2)$ for any choice of angles, $\eta$-states are invariant under $U_{Z_j}(\theta_1)\ot U_{Z_j}(\theta_2)^*\ot U_{Z_j}(\theta_1)^*\ot U_{Z_j}(\theta_2)$, and $s$-states are invariant under both.
        We will decompose superoperators in strings of these basis elements:
        \begin{equation}\label{eq:perminvbasis}
            \ket{\mc Q}_4=\sum_{\alpha^{(j)}\in\{e,s,\eta\}}\sum_{n^{(j)}\in\{0,1\}}c_{\alpha^{(1)}_{n^{(1)}}...\alpha^{(L)}_{n^{(L)}}}\ket{\alpha^{(1)}_{n^{(1)}}...\alpha^{(L)}_{n^{(L)}}}_4
        \end{equation}
        \item By choosing $U=\exp(i\frac{\pi}{4}(\1_j\1_{j+1}-X_jX_{j+1}-Y_jY_{j+1}-Z_jZ_{j+1}))$, which is the swap operator {between states on sites $j$ and $j+1$}, we find that the super-commutant must be permutationally invariant.
        Therefore if a given string of basis elements appears in the decomposition of a superoperator conserved quantity {such as in Eq.~\eqref{eq:perminvbasis}}, then all its permutations also appear with the same coefficient.
        \item 
        Finally we use the generator $Z_jZ_{j+1}$ to prove that in the decomposition of a superoperator symmetry, no strings can contain \textit{both} $e$-states and $\eta$-states.
        {We first observe that} due to permutational invariance, if a conserved superoperator has a string containing an $e$-state as well as an $\eta$-state in its decomposition, then it also contains a string where the $e$-state and the $\eta$-state are on neighboring sites $j$ and $j+1$.
        {Then} if we act on this string by {$U \otimes U^\ast \otimes U^\ast \otimes U$} with $U=\exp(i\frac{\pi}{4} Z_jZ_{j+1})$ we will obtain an overall minus sign (for any choice of $e$-state and $\eta$-state), {which} implies that the original superoperator could not be a symmetry, and therefore shows that superoperator symmetries only contain strings composed of combinations of $e/s$-states or $\eta/s$-states.
    \end{itemize}
    With these considerations in mind, we have shown that $\scommx{U(1)}\subseteq E\oplus N\oplus S$, where $E$ (resp. $N$) is spanned by all states of the form Eq.~\eqref{eq:perminvbasis} where $\alpha^{(j)}$ is chosen from $\{s,e\}$ (resp. $\{s,\eta\}$) such that not all $\alpha^{(j)}$ are equal to $s$, and where $S$ is spanned by all states of the form Eq.~\eqref{eq:perminvbasis} where $\alpha^{(j)}\equiv s$.
    
    Another characterization of the super-commutant $\scomm$ is given by considering the projectors $\Pi_{j,j+1}$, which act on two sites $j,j+1$ of the 4-copy superoperator Hilbert space by projecting onto the two-site super-commutant (i.e. the super-commutant for a chain of length $L=2$), in particular we have that\begin{equation}\label{eq:characterizationofthesupercommutant}
        \ket{\mc Q}_4 \in \scommx{U(1)} \;\;\iff\;\; \Pi_{j,j+1}\ket{\mc Q}_4=\ket{\mc Q}_4\;\;\forall\;\;j\in\{1,...,L-1\}, 
    \end{equation}
    since the generators $\gen_{U(1)}$ are 2-local (an equivalent fact is also exploited in the MPS method described in App.~\ref{app:mpsmethod}).
    But since all the projectors $\Pi_{j,j+1}$ are all translations of each other, and since all superoperators $\ket{\mc Q}_4\in E\oplus N\oplus S$ are invariant under permutations:
    \begin{equation}
        \ket{\mc Q}_4\in\scommx{U(1)} \iff \Pi_{1,2}\ket{\mc Q}_4=\ket{\mc Q}_4.
    \end{equation}
    Having reduced the problem to a single projector, we can now consider the group $G=\{\1,\1-2\Pi_{1,2}\}$ and apply Lemma~\ref{lemmaU1proof} to these subspaces $E$, $N$, $S$.
    {First, we show that the action of the group on $E$ is orthogonal to $N$, which can be seen by writing} a (non-orthogonal) basis for the two-site super-commutant in terms of the states in Eq.~\eqref{eq:u1scommbasis} as:
    \begin{equation}
    \left\{\begin{gathered}\label{eq:2sitestatesU1}
            \frac{\ket{s_0s_1}+\ket{s_1s_0}+\ket{\alpha_0\alpha_1}+\ket{\alpha_1\alpha_0}}{2},\\
            \frac{\ket{\alpha_0s_0}+\ket{s_0\alpha_0}}{\sqrt 2}, \qquad \frac{\ket{\alpha_1s_0}+\ket{s_0\alpha_1}}{\sqrt 2}, \qquad \frac{\ket{\alpha_0s_1}+\ket{s_1\alpha_0}}{\sqrt 2}, \qquad \frac{\ket{\alpha_1s_1}+\ket{s_1\alpha_1}}{\sqrt 2},\\
            \ket{s_0s_0},\qquad\ket{s_1s_1},\qquad \ket{\alpha_0\alpha_0},\qquad \ket{\alpha_1\alpha_1}
        \end{gathered}\right\}_{\alpha=e,\eta}
    \end{equation}
    Given a state $\ket{\mc E}_4\in E$, we can show that projection $\Pi_{1,2}\ket{\mc E}_4\perp N$ by performing a decomposition of $\ket{\mc E}_4$ between sites $j\in\{1,2\}$ and $j\in\{3,...,L\}$:
    \begin{equation}
        \ket{\mc E}_4 = \sum_i c_i \ket{i}_4^{(1,2)}\ot\ket{i}_4^{(3,...,L)}\;\;\implies\;\; \Pi_{1,2}\ket{\mc E}_4 = \sum_i c_i \left(\Pi_{1,2}\ket{i}_4^{(1,2)}\right)\ot\ket{i}_4^{(3,...,L)}.
    \end{equation}
    Here we choose the basis $\{\ket{i}_4^{(1,2)}\}$ to be $\{|\alpha_{n^{(1)}}^{(1)}\alpha_{n^{(2)}}^{(2)}\rangle\}$ where $\alpha^{(1)},\alpha^{(2)}\in\{s,e\}$ and $n^{(1)},n^{(2)}\in\{0,1\}$.
    If both $\alpha^{(1)}$ and $\alpha^{(2)}$ are $s$, then every string in the decomposition of $\ket{i}_4^{(3,...,L)}$ will contain at least one $e$-state, since $\ket{\mc E}_4\in E$ is orthogonal to $S$.
    If instead at least one of the two $\alpha^{(j)}$ are $e$, then $\ket{i}_4^{(1,2)}$ will be orthogonal to all $(\alpha=\eta)$-states in Eq.~\eqref{eq:2sitestatesU1}, and therefore $\Pi_{1,2}\ket{i}_4^{(1,2)}$ will be in the $(\alpha=e)$ subspace of Eq.~\eqref{eq:2sitestatesU1}, and will contain no $\eta$-states.
    In both cases, the state $\left(\Pi_{1,2}\ket{i}_4^{(1,2)}\right)\ot\ket{i}_4^{(3,...,L)}$ is orthogonal to $N$, and so we have shown that $\ket{\mc E}_4\in E$.
    
    Then Lemma~\ref{lemmaU1proof} then implies that
    \begin{equation}
        \begin{gathered}
            \mrm{If}\ \ \ket{\mc Q}_4\in E\oplus N\oplus S\ \ \mrm{is\ s.t.\ \ }\Pi_{1,2}\ket{\mc Q}_4=\ket{\mc Q}_4,\ \mrm{then}\\
            \ket{\mc Q}_4 = \ket{\mc E'}_4+\ket{\mc N'}_4\ \ \mrm{where}\ \ 
            \begin{cases}
                \Pi_{1,2}\ket{\mc E'}_4=\ket{\mc E'}_4,\quad \ket{\mc E'}_4\in E\oplus S,\\
                \Pi_{1,2}\ket{\mc N'}_4=\ket{\mc N'}_4,\quad \ket{\mc N'}_4\in N\oplus S.
            \end{cases}
        \end{gathered}
    \end{equation}
    Let us now call $\Pi_{1,2}^{(e)}$ and $\Pi_{1,2}^{(\eta)}$ the two-site projector on the states of Eq.~\eqref{eq:2sitestatesU1} with $\alpha=e$ and $\alpha=\eta$ respectively.\footnote{Note that $\Pi_{1,2}\neq \Pi_{1,2}^{(e)}+\Pi_{1,2}^{(\eta)}$, since the states $\ket{s_0s_0}$ and $\ket{s_1s_1}$ are in the image of both projectors, and since the states $\frac{\ket{s_0s_1}+\ket{s_1s_0}+\ket{\alpha_0\alpha_1}+\ket{\alpha_1\alpha_0}}{2}$ for $\alpha=e$ and $\alpha=\eta$ are not orthogonal.}
    Since strings containing $\eta$-states are orthogonal to $E\oplus S$ and strings containing $e$-states are orthogonal to $N\oplus S$, in the above decomposition:
    \begin{equation}
        \Pi_{1,2}\ket{\mc E'}_4=\ket{\mc E'}_4 \iff \Pi_{1,2}^{(e)}\ket{\mc E'}_4=\ket{\mc E'}_4,\qquad \Pi_{1,2}\ket{\mc N'}_4=\ket{\mc N'}_4 \iff \Pi_{1,2}^{(\eta)}\ket{\mc N'}_4=\ket{\mc N'}_4.
    \end{equation}
    Due to permutational symmetry, we also have that $\ket{\mc E'/\mc N'}_4$ are invariant under all $\Pi_{j,j+1}^{(e/\eta)}$.
    But with reference to Eq.~\eqref{eq:forsure}, the two-site projector $\Pi_{j,j+1}^{(e)}$ is exactly the projector on superoperators of type (i), while the two-site projector $\Pi_{j,j+1}^{(\eta)}$ is exactly the projector on superoperators of type (ii).
    This can be shown explicitly by choosing $Q_1$ and $Q_2$ in Eq.~\eqref{eq:forsure} to be elements of the $L=2$ projector basis $\{\Pi_0,\Pi_1,\Pi_2\}$ (defined above Eq.~\eqref{eq:U1basis}).
    In other words, $\Pi_{j,j+1}^{(e)}$ projects on $\{Q_1 \ot Q_2^T\}_{Q_1,Q_2\in\commx{U(1)}}$ and $\Pi_{j,j+1}^{(\eta)}$ projects on $\{\oketbra{Q_2}{Q_1}\}_{Q_1,Q_2\in\commx{U(1)}}$. Therefore $\Pi_{j,j+1}^{(e)}$ can be factorized into the product of two commuting projectors:\begin{equation}\label{eq:projdecprod}
    \Pi_{j,j+1}^{(e)} = \bigg(\Pi_{j,j+1}^{\commx{U(1)}}\ot\1\bigg)\bigg(\1\ot\left(\Pi_{j,j+1}^{\commx{U(1)}}\right)^T\bigg)
    \end{equation}
    where $\Pi_{j,j+1}^{\commx{U(1)}}$ is the superoperator which projects onto the two-site commutant $\commx{U(1)}$ (which is a subspace of $\hend=\mc H\ot\mc H^*$). In the decomposition of the superoprator Hilbert space $V=\mc H\ot\mc H^*\ot\mc H^*\ot\mc H$, the first factor in the product of Eq.~\eqref{eq:projdecprod} acts non-trivially only on Hilbert spaces 1 and 2, while the second factor acts non-trivially only on Hilbert spaces 3 and 4.  A similar factorization is also possible for $\Pi_{j,j+1}^{(\eta)}$, but with the first factor acting on Hilbert spaces 1 and 3, and the second factor acting on Hilbert spaces 2 and 4.
    
    Finally, since $\Pi_{j,j+1}^{(e)}\ket{\mc E'}_4=\ket{\mc E'}_4$ and $\Pi_{j,j+1}^{(\eta)}\ket{\mc N'}_4=\ket{\mc N'}_4$ for all $j\in\{1,...,L-1\}$, and since these projectors can be decomposed into two projectors onto the commutant $\commx{U(1)}$ as described above, and since similarly to Eq.~\eqref{eq:characterizationofthesupercommutant}
    \begin{equation}
        \oket{Q}\in\commx{U(1)}\;\;\iff\;\; \Pi_{j,j+1}^{\commx{U(1)}}\oket{ Q}=\oket{ Q}\;\;\forall\;\;j\in\{1,...,L-1\}, 
    \end{equation}
    we have that $\ket{\mc E'}_4$ is globally of the form Eq.~\eqref{eq:forsure}-(i) and  $\ket{\mc N'}_4$ is globally of the form Eq.~\eqref{eq:forsure}-(ii), thus showing that the super-commutant $\scommx{U(1)}$ is as described in Eq.~\eqref{eq:U1supercomm}.
\end{proof}


\subsection{Remark on the Compactness of the Generated Set of Unitaries}\label{app:compact}
In Sec.~\ref{sec:superoperatoralgebra} we have briefly discussed the possibility that for a given set of gates $\gen$, the set of generated unitaries $\unit$ is not compact.
We have noted that if one chooses each generator $h_\alpha$ so that it has rational spectrum, then the subgroup generated by exponentiating the DLA will be compact, and therefore all statements about non-universality will also hold when discussing \textit{approximate} non-universality.
This result follows from Theorem 2 in Sec.~3.3 of Ref.~\cite{onishchik1990lie}.

The set of all unitaries that can be either generated exactly or approximately (with arbitrary precision) is the topological closure of the set $\unit$ within the group of all unitaries.
The closure is a compact Lie group, and its algebra contains $\dlie$ as a subalgebra and can be generated through a set of generators $\gen'=\{h_\alpha'\}$ with rational spectrum.
Therefore if the gate set of interest $\gen$ contains some generators which produce non-compact one-parameter subgroups $u_\alpha(\theta)$, its closure may be studied by identifying a minimal gate set $\gen'$ such that all generators $h_\alpha\in\gen$ can be obtained as linear combinations of $h_\alpha'\in\gen'$, where all the generators $h_\alpha'$ have rational spectrum.
As a simple example of this, consider $\gen=\{h=\smqty(1&\\&\varphi)\}$ where $\varphi$ is irrational. The set of generated unitaries $\unit$ is one-dimensional, but can approximate with arbitrary precision any diagonal unitary matrix; its closure is obtained by exponentiating the Lie algebra of diagonal matrices, generated by the set $\gen'=\{h_1=\smqty(1&\\&0),h_2=\smqty(0&\\&1)\}$, which can be obtained as a way of writing $h=h_1+\varphi h_2$ as a linear combination of generators with rational spectrum.
However, note that rational spectra are not absolutely necessary for compactness, and in the presence of generators $h_\alpha\in\gen$ that produce non-compact one-parameter subgroups, $\gen$ can still generate a group $\unit$ that is compact.
\section{Brownian circuits}\label{sec:brown}
We now review Brownian circuits which we use in Sec.~\ref{sec:physicalimplications}.
These have been studied in many earlier works~\cite{lashkari2013towards, bauer2017stochastic,  sunderhauf2019quantum, xu2019locality, ogunnaike2023unifying, moudgalya2023symmetries, vardhan2024entanglement}, and are defined as follows.
Starting from a given set of generators $\gen=\{h_\alpha\}$, we define a Brownian circuit as a time-dependent Hamiltonian
\begin{equation}
    H(t)=\sum_{\alpha} J_\alpha(t) h_\alpha\label{eq:brownhamdef}
\end{equation}
where $\{J_\alpha(t)\}$ are i.i.d. Brownian random variables with $\langle J_\alpha(t)J_{\alpha'}(t')\rangle = 2\kappa\,\delta_{\alpha\alpha'}\,\delta(t-t')$.
By performing the statistical average, the averaged time-evolution operator takes the simple form of imaginary-time evolution of an effective positive semidefinite Hamiltonian
\begin{equation}
        \overline{U(t)} = \overline{ \mrm{T}\{e^{-i\int_0^t H(t')\,\dd t'}\}}=e^{-\kappa Pt},\;\;\;
        P\defeq\sum_\alpha h_\alpha^2\geq 0.
\end{equation}
At long times the average evolution operator $e^{-\kappa Pt}$ will approach the projector onto the ground states of the effective Hamiltonian $P$.

As discussed in Sec.~\ref{sec:physicalimplications}, many physical quantities can be understood in terms of multiple copies of the system.
The associated many-copy Hamiltonians are linear functions of the original single-copy Hamiltonians, hence the many-copy evolution of a Brownian system will also be Brownian.
For example the Heisenberg evolution of operators -- associated to the two-copy space $\mc H\ot\mc H$ -- is given by the Hamiltonian
\begin{equation}
    \ad{H(t)} = \sum_\alpha J_\alpha(t) \ad{h_\alpha},
\end{equation}
which leads to the averaged behavior of the two copies of the system to be of the form
\begin{equation}
    \overline{U(t) \otimes U^\ast(t)} = e^{-\kappa P_2 t},\;\;P_2=\sum_\alpha \ad{h_\alpha}^2.
\end{equation}
The ground state space of $P_2$ is the set of operators that are annihilated by each of the $\ad{h_\alpha}$, i.e. the commutant algebra $\comm$~\cite{moudgalya2023symmetries}, and one can show that symmetries predict the asymptotic value of two-point correlation functions.
In a similar fashion, at the level of four copies of the system, the average evolution of a superoperator is governed by the effective Hamiltonian defined as
\begin{equation}\label{eq:supercommutant-hamiltonian}
    \overline{U(t) \otimes U^\ast(t) \otimes U^\ast(t) \otimes U(t)} = e^{-\kappa P_4 t}\;\;\;P_4=\sum_\alpha \left(\ad{h_\alpha}\ot\1-\1\ot\ad{h_\alpha}^T\right)^2.
\end{equation}
The ground states of $P_4$ are exactly the superoperators in the super-commutant $\scomm$, and Eq.~\eqref{eq:otoc-projector} immediately follows.
Finally, for completeness, we note that the time evolution of a generic $k$-copy observables is given by:
\begin{equation}
    \begin{gathered}
        \overline{(U(t)\ot U^*(t))^{\ot k}} = \overline{\mrm{T}\{e^{-i\int_0^t H_{2k}(t')\,\dd t'}\}} = e^{-\kappa P_{2k}t},\qquad H_{2k}(t)\defeq\sum_\alpha J_\alpha(t)\mc L_{h_\alpha}^{(k)},\qquad
        P_{2k}\defeq\sum_\alpha \left(\mc L_{h_\alpha}^{(k)}\right)^2,\\
        \mc L_{h_\alpha}^{(k)}=\sum_{l=1}^k \1^{\ot 2(l-1)}\ot (h_\alpha\ot\1-\1\ot h_\alpha^T)\ot \1^{\ot 2(k-l)}.
    \end{gathered}
\end{equation}

\end{document}